\documentclass[runningheads]{llncs}

\usepackage[T1]{fontenc}

\usepackage{graphicx}
\usepackage{url}
\usepackage[small]{caption}

\usepackage{amsmath}
\usepackage{amsthm}
\usepackage{booktabs}
\usepackage{comment}
\usepackage{amsfonts}

\usepackage{multicol}

\usepackage{amsfonts}
\usepackage{amssymb}
\usepackage{latexsym}
\usepackage{stmaryrd}
\usepackage{wasysym}
\usepackage{mathrsfs}
\usepackage{graphicx}
\usepackage{rotating}
\usepackage{pifont}
\usepackage{booktabs}
\usepackage{array}
\newcolumntype{?}{!{\vrule width 1pt}}

\usepackage[numbers]{natbib}

\usepackage{multirow}
\usepackage{tikz}
\usetikzlibrary{arrows}
\usepackage{color}
\usepackage{xcolor}
\usepackage[inline]{enumitem}
\usepackage{xspace}
\usepackage{thm-restate} 
\usepackage{bussproofs}

\usepackage{tcolorbox}

\usepackage[edges]{forest}
\usepackage{varwidth}
\usepackage{adjustbox}

\usepackage[textsize=footnotesize]{todonotes}
\setuptodonotes{backgroundcolor=white, textcolor=red}

\usepackage[lined,linesnumbered,algoruled]{algorithm2e}

\makeatletter
\newsavebox{\@brx}
\newcommand{\llangle}[1][]{\savebox{\@brx}{\(\m@th{#1\langle}\)}%
  \mathopen{\copy\@brx\kern-0.5\wd\@brx\usebox{\@brx}}}
\newcommand{\rrangle}[1][]{\savebox{\@brx}{\(\m@th{#1\rangle}\)}%
  \mathclose{\copy\@brx\kern-0.5\wd\@brx\usebox{\@brx}}}
\makeatother

\newcommand{\ext}[1]{\llangle#1\rrangle}

\newcommand{\neighborhood}{neighbourhood\xspace} 
\newcommand{\p}{\varphi}

\newcommand{\reldomain}{\ensuremath{W}\xspace}
\newcommand{\relations}{\ensuremath{R}\xspace}
\newcommand{\role}{\ensuremath{r}\xspace}
\newcommand{\eset}{\emptyset}
\newcommand{\mdl}{\models}
\newcommand{\Int}{\ensuremath{\mathcal{I}}\xspace}

\newcommand{\dnot}{\ensuremath{\dot{\lnot}}\xspace}

\newcommand{\sbs}{\subseteq}
\newcommand{\sqs}{\sqsubseteq}

\newcommand{\tr}{^{\dagger}}
\newcommand{\ttr}{^{\ddagger}}

\newcommand{\valuation}{\ensuremath{\nu}\xspace}

\newcommand{\propmodel}{\ensuremath{\mathcal{M}^{\sf P}}\xspace}
\newcommand{\propdomain}{\ensuremath{\mathcal{W}}\xspace}
\newcommand{\propneigh}{\ensuremath{\mathcal{N}}\xspace}
\newcommand{\propassign}{\ensuremath{\mathcal{V}}\xspace}

\newcommand{\alcform}{\ensuremath{\hat{\varphi}}_{\Vmc, w}\xspace}

\newcommand{\con}{\ensuremath{\mathsf{con}}\xspace}
\newcommand{\conneg}{\ensuremath{\mathsf{con}_{\dot{\lnot}}}\xspace}
\newcommand{\for}{\ensuremath{\mathsf{for}}\xspace}
\newcommand{\forneg}{\ensuremath{\mathsf{for}_{\dot{\lnot}}}\xspace}
\newcommand{\rol}{\ensuremath{\mathsf{rol}}\xspace}
\newcommand{\ind}{\ensuremath{\mathsf{ind}}\xspace}

\newcommand{\fg}{\ensuremath{\mathsf{Fg}}\xspace}

\newcommand{\formtp}[1]{\ensuremath{{\boldsymbol{f}}^{\Vmc, w}_{#1}}\xspace}

\newcommand{\D}{\Diamond}

\newcommand{\B}{\Box}

\newcommand{\prop}[1]{\ensuremath{#1_{\sf prop} }\xspace} 
\newcommand{\elaxiom}{\ensuremath{\pi}\xspace}

\newcommand{\ALC}{\ensuremath{\smash{\mathcal{ALC}}}\xspace}

\newcommand{\NC}{\ensuremath{{\sf N_C}}\xspace}
\newcommand{\NI}{\ensuremath{{\sf N_I}}\xspace}
\newcommand{\NR}{\ensuremath{{\sf N_R}}\xspace}

\newcommand{\NPr}{\ensuremath{{\sf N_P}}\xspace}

\newcommand{\NV}{\ensuremath{{\sf N_{V}}}\xspace}

\newcommand{\MLn}{\ensuremath{\smash{\mathcal{ML}^{n}}}\xspace}

\newcommand{\MLALC}[1]{\ensuremath{\smash{\mathcal{ML}^{#1}_{\mathcal{ALC}}}}\xspace}

\newcommand{\MLnALC}{\ensuremath{\smash{\mathcal{ML}^{n}_{\mathcal{ALC}}}}\xspace}

\newcommand{\MLnALCg}{\ensuremath{\smash{\mathcal{ALC}\textnormal{-}\mathcal{ML}^{n}}\xspace}}

\newcommand{\Log}{\ensuremath{\smash{\mathsf{Pantheon}}}\xspace}

\newcommand{\logicnamestyle}[1]{\ensuremath{\smash{\mathbf{#1}}\xspace}}

\newcommand{\E}{\logicnamestyle{E}}
\newcommand{\EM}{\logicnamestyle{EM}}
\newcommand{\EC}{\logicnamestyle{EC}}
\newcommand{\EN}{\logicnamestyle{EN}}
\newcommand{\ET}{\logicnamestyle{ET}}
\newcommand{\ED}{\logicnamestyle{ED}}
\newcommand{\EP}{\logicnamestyle{EP}}
\newcommand{\EQ}{\logicnamestyle{EQ}}

\newcommand{\EMC}{\logicnamestyle{EMC}}
\newcommand{\EMN}{\logicnamestyle{EMN}}
\newcommand{\EMT}{\logicnamestyle{EMT}}
\newcommand{\EMD}{\logicnamestyle{EMD}}
\newcommand{\EMP}{\logicnamestyle{EMP}}

\newcommand{\ECN}{\logicnamestyle{ECN}}
\newcommand{\ECT}{\logicnamestyle{ECT}}
\newcommand{\ECD}{\logicnamestyle{ECD}}
\newcommand{\ECP}{\logicnamestyle{ECP}}
\newcommand{\ECQ}{\logicnamestyle{ECQ}}

\newcommand{\ENT}{\logicnamestyle{ENT}}
\newcommand{\END}{\logicnamestyle{END}}
\newcommand{\ENP}{\logicnamestyle{ENP}}

\newcommand{\ETD}{\logicnamestyle{ETD}}
\newcommand{\ETP}{\logicnamestyle{ETP}}
\newcommand{\ETQ}{\logicnamestyle{ETQ}}

\newcommand{\EDP}{\logicnamestyle{EDP}}
\newcommand{\EDQ}{\logicnamestyle{EDQ}}

\newcommand{\EPQ}{\logicnamestyle{EPQ}}

\newcommand{\EMCN}{\logicnamestyle{EMCN}}
\newcommand{\EMCT}{\logicnamestyle{EMCT}}
\newcommand{\EMCD}{\logicnamestyle{EMCD}}
\newcommand{\EMCP}{\logicnamestyle{EMCP}}

\newcommand{\EMNT}{\logicnamestyle{EMNT}}
\newcommand{\EMND}{\logicnamestyle{EMND}}
\newcommand{\EMNP}{\logicnamestyle{EMNP}}

\newcommand{\EMTD}{\logicnamestyle{EMTD}}
\newcommand{\EMTP}{\logicnamestyle{EMTP}}

\newcommand{\EMDP}{\logicnamestyle{EMDP}}

\newcommand{\ECNT}{\logicnamestyle{ECNT}}
\newcommand{\ECND}{\logicnamestyle{ECND}}
\newcommand{\ECNP}{\logicnamestyle{ECNP}}

\newcommand{\ECTD}{\logicnamestyle{ECTD}}
\newcommand{\ECTP}{\logicnamestyle{ECTP}}
\newcommand{\ECTQ}{\logicnamestyle{ECTQ}}

\newcommand{\ECDQ}{\logicnamestyle{ECDQ}}
\newcommand{\ECPQ}{\logicnamestyle{ECPQ}}

\newcommand{\ENTD}{\logicnamestyle{ENTD}}
\newcommand{\ENTP}{\logicnamestyle{ENTP}}

\newcommand{\ETDP}{\logicnamestyle{ETDP}}
\newcommand{\ETDQ}{\logicnamestyle{ETDQ}}
\newcommand{\ETPQ}{\logicnamestyle{ETDQ}}

\newcommand{\EDPQ}{\logicnamestyle{EDPQ}}

\newcommand{\EMCNT}{\logicnamestyle{EMCNT}}
\newcommand{\EMCND}{\logicnamestyle{EMCND}}
\newcommand{\EMCNP}{\logicnamestyle{EMCNP}}

\newcommand{\EMCTD}{\logicnamestyle{EMCTD}}
\newcommand{\EMCTP}{\logicnamestyle{EMCTP}}

\newcommand{\EMCDP}{\logicnamestyle{EMCDP}}

\newcommand{\EMNTD}{\logicnamestyle{EMNTD}}
\newcommand{\EMNDP}{\logicnamestyle{EMNDP}}
\newcommand{\EMNTP}{\logicnamestyle{EMNTP}}

\newcommand{\EMTDP}{\logicnamestyle{EMTDP}}

\newcommand{\ECTDP}{\logicnamestyle{ECTDP}}
\newcommand{\ECTDQ}{\logicnamestyle{ECTDQ}}
\newcommand{\ECTPQ}{\logicnamestyle{ECTPQ}}

\newcommand{\ECDPQ}{\logicnamestyle{ECDPQ}}

\newcommand{\ECNTD}{\logicnamestyle{ECNTD}}
\newcommand{\ECNTP}{\logicnamestyle{ECNTP}}

\newcommand{\ECNDP}{\logicnamestyle{ECNDP}}

\newcommand{\ENTDP}{\logicnamestyle{EMTDP}}

\newcommand{\ETDPQ}{\logicnamestyle{ETDPQ}}

\newcommand{\EMCNTD}{\logicnamestyle{EMCNTD}}
\newcommand{\EMCNTP}{\logicnamestyle{EMCNTP}}

\newcommand{\EMCNDP}{\logicnamestyle{EMCNDP}}

\newcommand{\EMCTDP}{\logicnamestyle{EMCTDP}}

\newcommand{\EMNTDP}{\logicnamestyle{EMNTDP}}

\newcommand{\ECNTDP}{\logicnamestyle{ECNTDP}}

\newcommand{\ECTDPQ}{\logicnamestyle{ECTDPQ}}

\newcommand{\EMCNTDP}{\logicnamestyle{EMCTDP}}

\newcommand{\LnALC}{\ensuremath{\mathit{L}^{n}_{\ALC}}\xspace}

\newcommand{\EnALC}[1]{\ensuremath{\smash{\mathbf{E}^{#1}_{\mathcal{ALC}}}}\xspace}

\newcommand{\MnALC}[1]{\ensuremath{\smash{\mathbf{EM}^{#1}_{\mathcal{ALC}}}}\xspace}
\newcommand{\KnALC}[1]{\ensuremath{\smash{\mathbf{K}^{#1}_{\mathcal{ALC}}}}\xspace}

\newcommand{\LnALCg}{\ensuremath{\smash{{\mathcal{ALC}}\textnormal{-}{\mathit{L}^{n}}}\xspace}}

\newcommand{\NP}{\textsc{NP}}
\newcommand{\PSpace}{\textsc{PSpace}}
\newcommand{\ExpTime}{\textsc{ExpTime}}
\newcommand{\NExpTime}{\textsc{NExpTime}}

\newcommand{\Fmf}{\ensuremath{F}\xspace}

\newcommand{\Mmf}{\ensuremath{M}\xspace}

\newcommand{\Cmc}{\ensuremath{\mathcal{C}}\xspace}

\newcommand{\Fmc}{\ensuremath{\mathcal{F}}\xspace}

\newcommand{\Imc}{\ensuremath{\mathcal{I}}\xspace}

\newcommand{\Mmc}{\ensuremath{\mathcal{M}}\xspace}
\newcommand{\Nmc}{\ensuremath{\mathcal{N}}\xspace}

\newcommand{\Pmc}{\ensuremath{\mathcal{P}}\xspace}

\newcommand{\Tmc}{\ensuremath{\mathcal{T}}\xspace}

\newcommand{\Vmc}{\ensuremath{\mathcal{V}}\xspace}
\newcommand{\Wmc}{\ensuremath{\mathcal{W}}\xspace}

\newcommand{\W}{\ensuremath{\mathcal{W}}\xspace}

\newcommand{\Lvar}{\mathit{L}}
\newcommand{\falseprop}{\mathsf{ff}}
\newcommand{\trueprop}{p \lor \lnot p}

\newcommand{\ax}{\AxiomC}
\newcommand{\llab}{\LeftLabel}

\newcommand{\uinf}{\UnaryInfC}
\newcommand{\disp}{\DisplayProof}

\newcommand{\T}{\mathcal{T}}
\spnewtheorem{claim}{Claim}{\itshape}{\rmfamily}
\usepackage{etoolbox}
\AtEndEnvironment{document}{\setcounter{claim}{0}}
\begin{document}


\title{Non-Normal Modal Description Logics \\ (Extended Version)}


\author{%
Tiziano Dalmonte$^1$
\and
Andrea Mazzullo$^2$
\and
Ana Ozaki$^3$
\and
Nicolas Troquard$^1$
}

\authorrunning{T. Dalmonte et al.}

\institute{
Free University of Bozen-Bolzano
\email{
\{name.surname\}@unibz.it}
\and
University of Trento
\email{
andrea.mazzullo@unitn.it}
\and
University of Oslo \& University of Bergen
\email{
anaoz@ifi.uio.no}
}

\maketitle

\begin{abstract}

Modal logics are widely used in multi-agent systems to reason about actions, abilities, norms, or epistemic states.
Combined with description logic languages, they
are also a powerful tool to formalise modal aspects of ontology-based reasoning over an object domain.
However, the standard relational semantics for modalities is known to validate principles
deemed problematic in agency, deontic, or epistemic applications.
To overcome these difficulties, weaker systems of so-called \emph{non-normal} modal logics, equipped with \emph{neighbourhood semantics} that generalise the relational one, have been investigated both at the propositional and at the description logic level.
We present here
a family of
\emph{non-normal modal description logics}, obtained by extending $\ALC$-based languages with non-normal modal operators.
For formulas interpreted on neighbourhood models over varying domains, we provide a modular framework of terminating, correct, and complete tableau-based satisfiability checking algorithms in $\NExpTime$.
For a subset of these systems, we also consider a reduction to satisfiability on constant domain relational models.
Moreover, 
we investigate the satisfiability problem in fragments obtained by disallowing the application of modal operators to description logic concepts,
providing tight $\ExpTime$ complexity results.

\end{abstract}

\section{Introduction}

\emph{Modal logics} are powerful tools used to represent and reason about actions and abilities~\cite{Brown,Elg}, coalitions~\cite{Pau,Tro}, knowledge and beliefs~\cite{Ago,Bal,Var1,LismontMongin}, obligations and permissions~\cite{AngEtAl,Gob,Wright}, etc.
In combination with \emph{description logics}, they give rise to \emph{modal description logics}~\cite{WolZak98,GabEtAl03}, knowledge representation formalisms used for modal reasoning over an object domain and with a good trade-off between expressive power and decidability.

The standard \emph{relational semantics} for modal operators is given in terms of \emph{frames} consisting of a set of \emph{possible worlds} equipped with binary \emph{accessibility relations}.
The foundations of modal description logics, so far, have also mostly been studied with relational semantics.
However, all the modal systems interpreted
with respect to
this semantics, known as \emph{normal}, validate principles that have been considered problematic or debatable for agency-based, coalitional, epistemic, or deontic applications, in that they lead to
unacceptable conclusions,
e.g.,
\emph{logical omniscience} in epistemic settings~\cite{Var1}, as well as
\emph{agency} or \emph{deontic paradoxes} in the representation of agents' abilities~\cite{Elg} and obligations~\cite{Ross,Aqv,For}.

To overcome these problems,
a generalisation of relational semantics, known as \emph{neighbourhood semantics}, was introduced by Scott~\cite{Sco} and Montague~\cite{Mon}.
Since it avoids in general the problematic principles
validated by relational semantics, it has been used to interpret a number of \emph{non-normal} modal logics, first studied by C.I. Lewis~\cite{CIL}, Lemmon~\cite{Lem}, Kripke~\cite{Kripke},
Segerberg~\cite{Seg}, and Chellas~\cite{Che}, among others.
A \emph{neighbourhood frame} consists of a set of worlds, each one associated with a ``neighbourhood'', i.e., a set of subsets of worlds.
Intuitively, a subset of worlds can be
thought of as
representing a fact
in a model, namely, those worlds where that fact holds.
Hence, the idea is that every world is assigned to a collection of
facts,
those
that are
brought about,
known, obligatory, etc.,
in that world of the model.
%

These are the neighbourhood semantics ingredients for \emph{propositional} non-normal modal logics.
A further line of research focuses on the behaviour of modal operators interpreted on neighbourhood frames in combination with \emph{first-order} logic.
In this direction, completeness results for first-order non-normal modal logics have been provided~\cite{Cos,CosPac}.
In addition, \emph{non-normal modal} \emph{description logics}, extending standard description logics, 
with modal operators interpreted on neighbourhood frames, have been considered for knowledge representation applications~\cite{SeyErd09,DalEtAl19,DalEtAl22}, also in multi-agent coalitional settings~\cite{SeyJam09,SeyJam10}.

To illustrate the expressivity of non-normal modal description logic languages, as well some of the limitations of relational frames behind adoption of neighbourhood semantics, we provide an example based on a classic multi-agent purchase choreography scenario~\cite{MonEtAl10} (see the Appendix for a detailed version).
Our multi-agent
setting involves
a customer $\mathit{c}$ and a seller $\mathit{s}$, as well as
agency operators $\mathbb{D}_i$ and $\mathbb{C}_i$, for $i \in \{ c, s \}$, read as `agent $i$ does/makes' and `agent $i$ can do/make', respectively~\cite{Elg,GovernatoriRotolo}.
The formula
$
\label{eq:1eq}
\mathsf{Ord} \equiv \mathbb{D}_{c}\exists \mathsf{req}.(\mathsf{Prod} \sqcap \mathsf{InCatal} )
$
defines an order $\mathsf{Ord}$ as a request made by customer $c$ of an in-catalogue
product.

By stating
$
\label{eq:2eq}
\exists \mathsf{req}.(\mathsf{Prod} \sqcap \mathsf{InCatal} ) \sqsubseteq \mathsf{Confirm} \sqcup \lnot \mathsf{Confirm},
$
we can also enforce that any request of an in-catalogue product is either confirmed or not confirmed.
However, relational semantics validates the so-called \emph{$\mathbf{M}$-principle} (often called \emph{monotonicity}) as well, according to which $C \sqsubseteq D$ always entails $\mathbb{D}_{c} C \sqsubseteq \mathbb{D}_{c} D$, for any concepts $C, D$.
Thus,
from the $\mathbf{M}$-principle
and $\mathsf{Ord}$ definition,
we obtain
$
\label{eq:monconc}
\mathsf{Ord} \sqsubseteq \mathbb{D}_{c} ( \mathsf{Confirm} \sqcup \lnot \mathsf{Confirm} ),
$
meaning that any order
is made confirmed or not confirmed by $c$. This is an unwanted conclusion in our agency-based scenario, since customers' actions should be unrelated to order confirmation aspects.\footnote{Other approaches (out of the scope of this paper) to avoid such consequences would involve rejecting the principle of \emph{excluded middle}, as done e.g. in \emph{intuitionistic description logics}~\cite{Dep06,BozEtAl07,Sch15}.}

Moreover, the formula
$
\label{eq:agglprem}
\mathsf{SubmitOrd} \sqsubseteq \mathbb{C}_{s}\mathsf{Confirm} \sqcap \mathbb{C}_{s}\mathsf{PartConf} \sqcap \mathbb{C}_{s}\mathsf{Reject}
$
states that a submitted order can be confirmed, can be partially confirmed, and can be rejected by the seller $s$.
On relational frames,
$ \mathbb{C}_{s} C \sqcap \mathbb{C}_{s} D \sqsubseteq \mathbb{C}_{s}( C \sqcap D ) $
 is a valid formula, for any concepts $C, D$, known as the \emph{$\mathbf{C}$-principle} (or \emph{agglomeration}).
 Therefore, by 
 the $\mathbf{C}$-principle, under relational semantics we would be forced to conclude that
$
\label{eq:agglconc}
\mathsf{SubmitOrd} \sqsubseteq \mathbb{C}_{s}(\mathsf{Confirm} \sqcap \mathsf{PartConf} \sqcap \mathsf{Reject}),
$
meaning that any submitted order is such that the seller $s$ has the ability to make it confirmed, partially confirmed, and rejected, all \emph{at once}, which is unreasonable.

Finally, consider the formula
$
\label{eq:necprem}
\top \sqsubseteq \mathsf{Confirm} \sqcup \lnot \mathsf{Confirm},
$
i.e., the truism stating that anything is either confirmed or not confirmed.
By the so called \emph{$\mathbf{N}$-principle} (or \emph{necessitation}) of relational semantics, we have that if $\top \sqsubseteq C$ is valid on relational frames, then $\top \sqsubseteq \mathbb{D}_{c} C$ holds as well, for any concept $C$.
Thus, from
the $\mathbf{N}$-principle
it would follow
on relational semantics
that
$
\label{eq:necprem}
\top \sqsubseteq \mathbb{D}_{c} (\mathsf{Confirm} \sqcup \lnot \mathsf{Confirm}),
$
thereby forcing us to the consequence that every object is made by customer $c$ to be either confirmed or not confirmed, hence leading again to an unreasonable connection between customer's actions and confirmation of orders.

The $\mathbb{D}_{i}$ and $\mathbb{C}_{i}$ modalities are axiomatised similarly to~\protect\cite{Elg}, by means of additional principles as well:
$\mathbb{D}_{i}$ obeys the $\mathbf{C}$- (seen above) and \emph{$\mathbf{T}$-principle}
($
\mathbb{D}_{w} C  \sqsubseteq C
$),
stating a \emph{factivity of actions} principle, well-known also in epistemic logic;
and both satisfy the \emph{$\mathbf{Q}$-principle}
($
\label{eq:qprinc}
\top \sqsubseteq \lnot \mathbb{D}_{c} \top
$),
asserting a principle of \emph{impotence towards tautologies} that is unsatisfiable in relational frames, but admissible over neighbourhood ones, and the \emph{$\mathbf{E}$-principle} ($C \equiv D$ entails $\mathbb{D}_{i} C \equiv \mathbb{D}_{i} D$ and $\mathbb{C}_{i} C \equiv \mathbb{C}_{i} D$), valid both on relational and neighbourhood frames.

In this paper, which is an extension of~\cite{DalEtAl19,DalEtAl22},
we investigate reasoning in a family of non-normal modal description logics, providing terminating, sound, and complete tableau algorithms for checking formula satisfiability on neighbourhood models based on \emph{varying domains} of objects.
Moreover,
we study the complexity of reasoning in a restricted fragment that disallows modalities on description logic concepts. Finally, for two modal description logics interpreted on \emph{constant domain} neighbourhood models, we adjust a reduction (known from the propositional case) to satisfiability with respect to standard relational semantics.

The paper is structured as follows.
Section~\ref{sec:prelim} provides the necessary definitions and the preliminary results on non-normal modal description logics.
In Section~\ref{sec:tableaux} we present the tableau algorithms for the family of logics here considered.
The case of fragments without modalised concepts is then studied in Section~\ref{sec:fragvardom}.
Section~\ref{sec:reasoncondom} contains the results for the constant domain case.
Finally, Section~\ref{sec:discuss} concludes the paper, discussing related work and possible future research directions.

\section{Preliminaries}
\label{sec:prelim}

Here we introduce
modal description logics, first presenting their syntax, and then their semantics based on neighbourhood and relational models, respectively.
Finally, we introduce the family of frame conditions here considered. 

\subsection{Syntax}

Let \NC, \NR and \NI be countably infinite and pairwise disjoint 
sets of \emph{concept}, \emph{role}, and \emph{individual names}, respectively.
An $\MLALC{n}$ \emph{concept} is an expression of the form
$
C ::= A \mid \lnot C \mid C \sqcap C \mid \exists \role.C \mid \B_{i} C,
$
where $A \in \NC$, $\role \in \NR$, and $\B_{i}$ such that
$i \in J = \{ 1, \ldots, n \}$.
%
%
An \emph{\MLALC{n} atom} is
a \emph{concept inclusion} (\emph{CI}) of
the form $(C \sqsubseteq D)$, or an \emph{assertion} of the form $C(a)$ or
$\role(a,b)$, with $C, D$ \MLALC{n} concepts,
$\role \in \NR$, and $a, b \in \NI$.
An \emph{\MLALC{n} formula}
has the form
$
\varphi::= \pi \mid \neg \varphi\mid \varphi\land \varphi\mid \B_{i} \p,
$
where
$\pi$ is an \MLALC{n} atom
and
$i \in J$.
We use the following standard definitions for concepts:
$\forall \role.C :=  \lnot \exists \role.\lnot C$;
$(C \sqcup D) :=  \lnot(\lnot C \sqcap \lnot D)$;
$\bot := A \sqcap \lnot A$,
$\top := A \sqcup \lnot A$
(for an arbitrarily fixed $A \in \NC$);
$\D_{i} C := \lnot \B_{i} \lnot C$.
Concepts of the form $\B_{i} C$, $\D_{i} C$ are \emph{modalised concepts}.
Analogous conventions
hold for formulas,
writing
$C \equiv D$ for $(C \sqsubseteq D) \land (D \sqsubseteq C)$
and setting
$\mathsf{false} := (\top \sqsubseteq \bot)$, $\mathsf{true} := (\bot \sqsubseteq \top)$.
%

\subsection{Semantics}\label{sec:sem}

%
We now define neighbourhood semantics, which (as already mentioned) can be seen as a generalisation of the relational semantics, introduced immediately after.

\subsubsection{Neighbourhood Semantics}
 A \emph{neighbourhood frame}, or simply \emph{frame},
is a pair
$\Fmc = ( \Wmc, \{\Nmc_i \}_{i \in J})$,
where 
$\Wmc$ is a non-empty set
of \emph{worlds}
and
$\Nmc_{i} \colon \W \rightarrow 2^{2^{\Wmc}}$ is   a \emph{neighbourhood function}, for each
\emph{agent}
$i \in J = \{1, \ldots, n\}$. 
%
%
An \emph{\MLALC{n} varying domain neighbourhood model}, or simply \emph{model}, based on a neighbourhood frame $\Fmc$ is a pair
$\Mmc = (\Fmc, \Int)$,
where
$\Fmc = (\Wmc,  \{\Nmc_i \}_{i \in J})$ is a neighbourhood frame
and $\Imc$ is a function associating with every $w \in \Wmc$ an \emph{$\ALC$ interpretation}
$\Imc_{w} = (\Delta_{w}, \cdot^{\Imc_{w}})$,
with non-empty \emph{domain} $\Delta_{w}$,
and where $\cdot^{\Imc_{w}}$ is a function such that:
for all $A \in \NC$, $A^{\Imc_{w}} \subseteq \Delta_{w}$;
for all $\role \in \NR$, $\role^{\Imc_{w}} \subseteq \Delta_{w} {\times} \Delta_{w}$;
for all $a \in \NI$, $a^{\Imc_{w}} \in \Delta_{w}$.
An \MLALC{n} \emph{constant domain neighbourhood model}
is defined in the same way, except that, for all $w,w'\in\Wmc$,
we have that $\Delta_{w}=\Delta_{w'}$
and, for all $u, v \in \Wmc$,
we require $a^{\Imc_{u}}= a^{\Imc_{v}}$ (denoted by $a^{\Imc}$), that is, individual names are \emph{rigid designators}.
We often write $\Mmc = (\Fmc, \Delta, \Imc)$ to denote a constant domain neighbourhood model $\Mmc = (\Fmc, \Imc)$ with domain $\Delta = \Delta_{w}$, for every $w \in \Wmc$.
%
%
%
Given a model $\Mmc = (\Fmc, \Int)$ and a world $w \in \Wmc$ of $\Fmc$ (or simply \emph{$w$ in $\Fmc$}), the \emph{interpretation $C^{\Imc_{w}}$ of a concept $C$ in $w$}
is defined
as:
$
		(\neg D)^{\Imc_{w}} = \Delta_{w} \setminus D^{\Imc_{w}}, \quad (D \sqcap E)^{\Imc_{w}}  = D^{\Imc_{w}} \cap E^{\Imc_{w}},
		$
		$
		(\exists r.D)^{\Imc_{w}} = \{d \in \Delta_{w} \mid \exists  e \in D^{\Imc_{w}}{:} (d,e) \in r^{\Imc_{w}}\},
		$
		$
		 (\B_{i} D)^{\Imc_{w}} = \{ d \in \Delta_{w} \mid \llbracket D \rrbracket^{\Mmc}_{d} \in \Nmc_{i}(w) \},
		 $
where, for all 
$d \in \bigcup_{w \in \Wmc} \Delta_{w}$, the set
$\llbracket D \rrbracket^{\Mmc}_{d} = \{ v \in \Wmc \mid  d \in D^{\Imc_{v}} \}$
is called the \emph{truth set of $D$ with respect to \Mmc and $d$}. 
We say that a concept $C$ is \emph{satisfied in $\Mmc$} if there is $w$ in $\Fmc$ such that $C^{\Imc_{w}} \neq \eset$, and that $C$ is \emph{satisfiable} (over varying or constant neighbourhood models, respectively) if there is a (varying or constant domain, respectively) neighbourhood model in which it is satisfied.
The \emph{satisfaction of an $\MLALC{n}$ formula~$\p$ in $w$ of $\Mmc$}, written $\Mmc, w  \models \p$, is defined
as follows:
\begin{alignat*}{6}
	& \Mmc, w \models C\sqsubseteq D && \text{ iff } && C^{\Imc_{w}} \subseteq D^{\Imc_{w}}, 
	&& \ \
	\Mmc, w  \models C(a) && \text{ iff } && a^{\Imc_{w}} \in C^{\Imc_{w}}, \\
	& \Mmc, w \models \role(a,b) && \text{ iff } && (a^{\Imc_{w}},b^{\Imc_{w}}) \in \role^{\Imc_{w}},
	&& \ \
	\Mmc, w  \models \neg \psi && \text{ iff } && \Mmc, w  \not \models \psi, \\
	& \Mmc, w  \models \psi \land \chi && \text{ iff } && \Mmc, w  \models \psi \text{ and } \Mmc, w  \models \chi, 
	&& \ \
	\Mmc, w  \models \B_{i} \psi && \text{ iff } && \llbracket \psi \rrbracket^{\Mmc} \in \Nmc_{i}(w),
 \end{alignat*} 
%
where
$\llbracket \psi \rrbracket^{\Mmc} = \{ v \in \Wmc \mid \Mmc, v \models \psi \}$ is the \emph{truth set of $\psi$}.
As a consequence of the above definition, we obtain the following
condition for $\Diamond_{i}$ formulas:
$\Mmc, w \models \Diamond_{i} \psi$  iff  $\llbracket \neg \psi \rrbracket^{\Mmc} \notin \Nmc_{i}(w)$.
Given a
neighbourhood
frame $\Fmc = (\Wmc,  \{\Nmc_i \}_{i \in J})$
and a
neighbourhood
model $\Mmc = (\Fmc, \Imc)$,
we say that $\varphi$ is \emph{satisfied in $\Mmc$} if there is $w \in \Wmc$ such that
$\Mmc, w \models \varphi$,
and that $\p$ is \emph{satisfiable} (over varying or constant domain neighbourhood models, respectively) if it is satisfied in some (varying or constant domain, respectively) neighbourhood model.
Also, $\p$ is    \emph{valid in $\Mmc$}, $\Mmc \models \p$, if it is satisfied in all $w$ 
of $\Mmc$, and it is \emph{valid on $\Fmc$} if, for all $\Mmc$ based on $\Fmc$,
$\p$ is valid in $\Mmc$,
writing $\Fmc \models \p$.
\subsubsection{Relational Semantics}
A
\emph{relational frame} 
is a pair
$\Fmf = ( \reldomain, \{\relations_i\}_{i \in J})$,
with 
$\reldomain$ non-empty set and $\relations_i$ 
binary relation on $\reldomain$,
for $i \in J = \{ 1, \ldots, n \}$.
%
An \emph{$\MLALC{n}$ (constant domain) relational model} 
based on a relational frame
$\Fmf = ( W, \{ R_{i} \}_{i \in J})$
is a
pair
$\Mmf = ( \Fmf, I)$,
where
$I$ is a function associating with every $w \in W$ an \ALC \emph{interpretation}
$I_w = (\Delta, \cdot^{I_{w}})$,
having non-empty \emph{constant domain} $\Delta$,
and where $\cdot^{I_{w}}$ is a function such that:
for all $A \in \NC$, $A^{I_{w}} \subseteq \Delta$;
for all $\role \in \NR$, $\role^{I_{w}} \subseteq \Delta {\times} \Delta$;
for all $a \in \NI$, $a^{I_{w}} \in \Delta$, and for all $u, v \in \reldomain$, $a^{I_u}=a^{I_{v}}$(denoted by $a^{I}$).
%
Given a relational model $M = (F, I)$ and a world $w \in W$ of $F$ (or simply $w$ in $F$), the 
\emph{interpretation of a concept $C$ in $w$}, written $C^{I_{w}}$, is defined by taking:
$
	(\neg C)^{I_{w}} = \Delta \setminus C^{I_{w}},$
	$
	(C \sqcap D)^{I_{w}}  = C^{I_{w}} \cap D^{I_{w}},
	$
	$
	(\exists \role.C)^{I_{w}} =
		\{d \in \Delta \mid \exists  e \in C^{I_{w}}{:}(d,e) \in \role^{I_{w}}\},
		$
		$
	(\B_{i} C)^{I_{w}} = 
		\{ d \in \Delta \mid \ \forall v \in W:
		w R_{i} v
\Rightarrow
	d \in C^{I_{v}} \}.
$
%

A concept $C$ is \emph{satisfied in $\Mmf$} if there is $w$ in $\Fmf$ such that $C^{I_{w}} 
\neq \eset$, and that $C$ is \emph{satisfiable on relational models} if there is a relational model in which it is satisfied.
The \emph{satisfaction of a $\MLALC{}$ formula~$\p$ in $w$ of $\Mmf$}, written $\Mmf, w  \models \p$, is defined, for atoms, negation and conjunction, similarly to the previous case, and as follows for the $\Box_{i}$ case:
$
		\Mmf, w  \models \B_{i} \varphi
		\text{ \ iff \ }
		\forall v \in \reldomain: w \relations_{i} v \Rightarrow \Mmf, v  \models \p.
		$
%
Given a relational frame $\Fmf = (\reldomain, \{\relations_i\}_{i \in J})$
and a relational model $\Mmf = (\Fmf, \Delta, I)$,
we say that $\varphi$ is \emph{satisfied in $\Mmf$} if there is $w \in \reldomain$ such that
$\Mmf, w \models \varphi$,
and that $\p$ is \emph{satisfiable on relational models} if it is satisfied in some relational model.
Also, $\p$ is said to be \emph{valid in $\Mmf$}, $\Mmf \models \p$, if it is satisfied in all $w$ 
of $\Mmf$, and it is \emph{valid on $\Fmf$} if, for all $\Mmf$ based on $\Fmf$,
$\p$ is valid in $\Mmf$,
writing $\Fmf \models \p$.
\subsection{Frame Conditions and Formula Satisfiability}

We consider the following conditions on neighbourhood frames $\Fmc = ( \Wmc, \{\Nmc_i \}_{i \in J})$. We say that  \emph{$\Fmc$ satisfies the}:
%
\begin{alignat*}{3}
	\text{\emph{$\mathbf{E}$-condition}} && \text{ \ iff \ } & \text{ $\Nmc_{i}$ is a neighbourhood function; } \\
	\text{\emph{$\mathbf{M}$-condition}} && \text{ \ iff \ }& \text{ $\alpha\in \Nmc_{i}(w)$ and $\alpha\subseteq\beta$ implies $\beta\in \Nmc_{i}(w)$; } \\
	\text{\emph{$\mathbf{C}$-condition}} && \text{ \ iff \ } & \text{ $\alpha\in \Nmc_{i}(w)$ and $\beta\in \Nmc_{i}(w)$ implies $\alpha\cap\beta\in \Nmc_{i}(w)$; } \\
	\textnormal{\emph{$\mathbf{N}$-condition}}  && \text{ \ iff \ } & \text{ $\Wmc \in \Nmc_{i}(w)$; } \\
	\text{\emph{$\mathbf{T}$-condition}} && \text{ \ iff \ } & \text{ $\alpha \in \Nmc_{i}(w)$ implies $w \in \alpha$; } \\
	\text{\emph{$\mathbf{D}$-condition}} && \text{ \ iff \ } & \text{  $\alpha \in \Nmc_{i}(w)$ implies $\Wmc \setminus \alpha \not \in \Nmc_{i}(w)$; } \\
	\text{\emph{$\mathbf{P}$-condition}} && \text{ \ iff \ } & \text{ $\emptyset \not \in \Nmc_{i}(w)$; } \\
	\text{\emph{$\mathbf{Q}$-condition}} && \text{ \ iff \ } & \text{  $\Wmc \not \in \Nmc_{i}(w)$; }
\end{alignat*}
for every $w\in \Wmc$, $\alpha,\beta\subseteq \Wmc$.
%
Combinations of conditions, such as the $\mathbf{EMCN}$-condition, are obtained by suitably joining the ones above.
Moreover, since the $\mathbf{E}$-condition is always satisfied by any neighbourhood frame, we often omit the letter $\mathbf{E}$ from this naming scheme, writing for instance `$\mathbf{MCN}$' in place of `$\mathbf{EMCN}$'.

On the relationships among (combinations of) neighbourhood frame conditions, we make the following observations.

\begin{restatable}{theorem}{PropImplicationSystem}\label{prop:implicationsystem}
Given a neighbourhood frame
$\Fmc = ( \Wmc, \{\Nmc_i \}_{i \in J})$, the following statements hold, for $i \in J$.
\begin{enumerate}
	\item If $\Nmc_i$ satisfies the $\mathbf{MQ}$-condition then, for every $w \in \Wmc$, $\Nmc_{i}(w) = \emptyset$. Hence, $\Nmc_i$ satisfies all but the $\mathbf{N}$-condition.
	\item\label{item:P-cond} $\Nmc_i$ satisfies the $\mathbf{P}$-condition, if $\Nmc_i$ satisfies one of the following:\\
		\begin{enumerate*}[label=(\roman*)]
			\item $\mathbf{MD}$-condition;
			\item $\mathbf{ND}$-condition; or
			\item $\mathbf{T}$-condition.
		\end{enumerate*}
	\item\label{item:D-cond} $\Nmc_i$ satisfies the $\mathbf{D}$-condition, if $\Nmc_i$ satisfies one of the following:\\
		\begin{enumerate*}[label=(\roman*)]
			\item $\mathbf{CP}$-condition; or
			\item $\mathbf{T}$-condition.
		\end{enumerate*}	
	\item $\Nmc_i$ does not satisfy the $\mathbf{NQ}$-condition.
\end{enumerate}
\end{restatable}

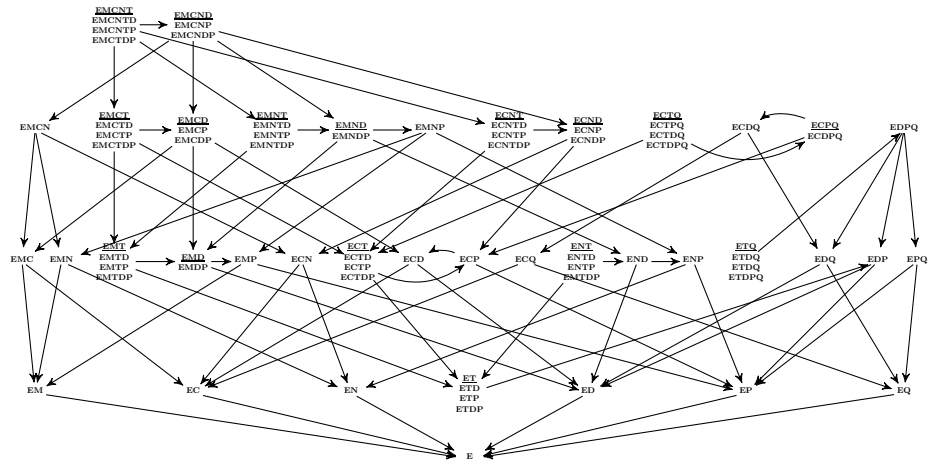
\begin{figure*}[h]
\centering
\begin{tikzpicture}[
scale=0.35, every node/.style={scale=0.35},
->,>=stealth',
box/.style={draw=none, text height=1cm, align=center}
]



%

\node[box, text height=1.4cm] (EMCNT) at (-13.5,16) {\underline{$\EMCNT$} \\ $\EMCNTD$ \\ $\EMCNTP$ \\ $\EMCNTDP$};

\node[box, text height=1cm] (EMCND) at (-10.5,16) {\underline{$\EMCND$} \\ $\EMCNP$ \\ $\EMCNDP$};


\node (EMCN) at (-16.5,12) {$\EMCN$};

\node[box, text height=1.4cm] (EMCT) at (-13.5,12) {\underline{$\EMCT$} \\ $\EMCTD$ \\ $\EMCTP$ \\ $\EMCTDP$};

\node[box, text height=1cm] (EMCD) at (-10.5,12) {\underline{$\EMCD$} \\ $\EMCP$ \\ $\EMCDP$};

\node[box, text height=1.4cm] (EMNT) at (-7.5,12) {\underline{$\EMNT$} \\ $\EMNTD$ \\ $\EMNTP$\\ $\EMNTDP$};

\node[box, text height=0.65cm] (EMND) at (-4.5,12) {\underline{$\EMND$} \\ $\EMNDP$};

\node (EMNP) at (-1.5,12) {$\EMNP$};

\node[box, text height=1.4cm] (ECNT) at (1.5,12) {\underline{$\ECNT$} \\ $\ECNTD$ \\ $\ECNTP$ \\ $\ECNTDP$};

\node[box, text height=1cm] (ECND) at (4.5,12) {\underline{$\ECND$} \\ $\ECNP$ \\ $\ECNDP$};

\node[box, text height=1.4cm] (ECTQ) at (7.5,12) {\underline{$\ECTQ$} \\ $\ECTPQ$ \\ $\ECTDQ$ \\ $\ECTDPQ$};

\node (ECDQ) at (10.5,12) {$\ECDQ$};

\node[box, text height=0.65cm] (ECPQ) at (13.5,12) {\underline{$\ECPQ$} \\ $\ECDPQ$};

\node (EDPQ) at (16.5,12) {$\EDPQ$};


\node (EMC) at (-17,7) {$\EMC$};

\node (EMN) at (-15.5,7) {$\EMN$};

\node[box, text height=1.4cm] (EMT) at
(-13.5,7)
{\underline{$\EMT$} \\ $\EMTD$ \\ $\EMTP$ \\ $\EMTDP$};

\node[box, text height=0.65cm] (EMD) at
(-10.5,7)
{\underline{$\EMD$} \\ $\EMDP$};

\node (EMP) at (-8.5,7) {$\EMP$};

\node (ECN) at (-6.375,7) {$\ECN$};

\node[box, text height=1.4cm] (ECT) at (-4.25,7) {\underline{$\ECT$} \\ $\ECTD$ \\ $\ECTP$ \\ $\ECTDP$};

\node (ECD) at (-2.125,7) {$\ECD$};

\node (ECP) at (0,7) {$\ECP$};

\node (ECQ) at (2.125,7) {$\ECQ$};

\node[box, text height=1.4cm] (ENT) at (4.25,7) {\underline{$\ENT$} \\ $\ENTD$ \\ $\ENTP$ \\ $\ENTDP$};

\node (END) at (6.375,7) {$\END$};

\node (ENP) at (8.5,7) {$\ENP$};

\node[box, text height=1.4cm] (ETQ) at
(10.5,7) 
{\underline{$\ETQ$} \\ $\ETDQ$ \\ $\ETPQ$ \\ $\ETDPQ$};

\node (EDQ) at (13.5,7)  {$\EDQ$};

\node (EDP) at (15.5,7) {$\EDP$};

\node (EPQ) at (17,7) {$\EPQ$};


\node (EM) at (-16.5,2) {$\EM$};

\node (EC) at (-10.5,2) {$\EC$};

\node (EN) at (-4.5,2) {$\EN$};

\node[box, text height=1.4cm] (ET) at (0,2) {\underline{$\ET$} \\ $\ETD$ \\ $\ETP$ \\ $\ETDP$};

\node (ED) at (4.5,2) {$\ED$};

\node (EP) at (10.5,2) {$\EP$};

\node (EQ) at (16.5,2) {$\EQ$};


\node (E) at (0,-0.5) {$\E$};


\draw[->, shorten >=1mm] (EM) -- (E);

\draw[->, shorten >=1mm] (EC) -- (E);

\draw[->, shorten >=1.5mm] (EN) -- (E);


\draw[->, shorten >=1.5mm] (ED) -- (E);

\draw[->, shorten >=1mm] (EP) -- (E);

\draw[->, shorten >=1mm] (EQ) -- (E);

\draw[->] (ET) -- (EDP);

\draw[->, shorten >=1mm] (EMC) -- (EM);

\draw[->, shorten >=1mm] (EMC) -- (EC);

\draw[->, shorten >=1mm] (EMN) -- (EM);

\draw[->, shorten >=1mm] (EMN) -- (EN);

\draw[->] (EMT) -- (ET);

\draw[->] (EMD) -- (ED);

\draw[->, shorten >=0.5mm] (EMP) -- (EP);

\draw[->, shorten >=1mm] (EMP) -- (EM);

\draw[->, shorten >=1mm] (ECN) -- (EC);

\draw[->, shorten >=1mm] (ECN) -- (EN);

\draw[->, shorten <=-0.5mm, shorten >=-1mm] (ECT) -- (ET);

\draw[->, shorten >=1mm] (ECD) -- (EC);

\draw[->, shorten >=1mm] (ECD) -- (ED);

\draw[->, shorten >=1mm] (ECP) -- (EP);

\draw[->, shorten >=1mm] (ECQ) -- (EC);

\draw[->, shorten >=0.5mm] (ECQ) -- (EQ);

\draw[->, shorten >=-1mm] (ENT) -- (ET);

\draw[->, shorten >=1mm] (END) -- (ED);

\draw[->, shorten >=1mm] (ENP) -- (EN);

\draw[->, shorten >=1mm] (ENP) -- (EP);

\draw[->, shorten >=1mm] (EDP) -- (ED);

\draw[->, shorten >=1mm] (EDP) -- (EP);


\draw[->, shorten >=1mm] (EDQ) -- (ED);

\draw[->, shorten >=1mm] (EDQ) -- (EQ);

\draw[->, shorten >=1mm] (EPQ) -- (EP);

\draw[->, shorten >=1mm] (EPQ) -- (EQ);

\draw[->] (EMT) -- (EMD);

\draw[->] (EMD) -- (EMP);

\draw[->] (ECT) edge[bend right] (ECP);

\draw[->, shorten <=1.5mm, shorten >=1.5mm] (ECP) edge[bend right] (ECD);

\draw[->] (ENT) -- (END);

\draw[->] (END) -- (ENP);

\draw[->, shorten <=-1.5mm] (ETQ) -- (EDPQ);


\draw[->, shorten >=1.5mm] (EMCN) -- (EMC);

\draw[->, shorten >=1.5mm] (EMCN) -- (EMN);

\draw[->, shorten >=1.5mm] (EMCN) -- (ECN);


\draw[->, shorten >=-0.5mm] (EMCT) -- (EMT);

\draw[->, shorten >=-1mm] (EMCT) -- (ECT);

\draw[->, shorten >=1.5mm] (EMCD) -- (EMC);

\draw[->] (EMCD) -- (EMD);

\draw[->, shorten >=1.5mm] (EMCD) -- (ECD);

\draw[->] (EMND) -- (EMD);

\draw[->, shorten >=1.5mm] (EMND) -- (END);


\draw[->, shorten >=-1mm] (EMNT) -- (EMT);

\draw[->, shorten >=1.5mm] (EMNP) -- (EMN);

\draw[->, shorten >=1.5mm] (EMNP) -- (EMP);

\draw[->, shorten >=1.5mm] (EMNP) -- (ENP);


\draw[->, shorten >=1.5mm] (ECND) -- (ECN);

\draw[->, shorten >=1.5mm] (ECND) -- (ECP);

\draw[->, shorten >=-1.5mm] (ECNT) -- (ECT);


\draw[->, shorten >=1.5mm] (ECDQ) -- (ECQ);

\draw[->, shorten >=1.5mm] (ECDQ) -- (EDQ);


\draw[->] (ECTQ) -- (ECT);


\draw[->, shorten >=1.5mm] (ECPQ) -- (ECP);


\draw[->, shorten >=1.5mm] (EDPQ) -- (EDP);

\draw[->, shorten >=1.5mm] (EDPQ) -- (EPQ);

\draw[->, shorten >=1.5mm] (EDPQ) -- (EDQ);

\draw[->] (EMCT) -- (EMCD);

\draw[->] (EMNT) -- (EMND);

\draw[->] (EMND) -- (EMNP);

\draw[->] (ECNT) -- (ECND);

\draw[->] (ECTQ) edge[bend right] (ECPQ);

\draw[->, shorten >=1.5mm,] (ECPQ) edge[bend right] (ECDQ);

\draw[->] (EMCNT) -- (EMCND);


\draw[->, shorten >=1.5mm] (EMCND) -- (EMCN);
\draw[->] (EMCND) -- (EMCD);
\draw[->] (EMCND) -- (EMND);
\draw[->] (EMCND) -- (ECND);
\draw[->] (EMCNT) -- (EMCT);
\draw[->, shorten >=-1.5mm] (EMCNT) -- (EMNT);
\draw[->] (EMCNT) -- (ECNT);

\end{tikzpicture}
\caption{Implications among $L$-conditions in $\Log$ (equivalent ones are listed in the same nodes, with underlined
representatives).}
\label{fig:pantheon}
\end{figure*}

Based on these results, Figure~\ref{fig:pantheon} depicts the relations between combinations of frame conditions: nodes are (groups of equivalent) conditions (with the canonical representative underlined), and arrows represent logical implications.
Any combination containing the $\mathbf{NQ}$-condition has been omitted, as it leads to inconsistency (Theorem~\ref{prop:implicationsystem}, Point 4). Moreover, due to Theorem~\ref{prop:implicationsystem}, Point 1, any combination that includes the $\mathbf{MQ}$-condition is not considered, since
for any neighbourhood frame $\Fmc$ satisfying such condition and any
$\MLnALC$ concept $C$, we have $\Fmc \models \Box_{i} C \equiv \bot$,
and similarly for formulas,
hence trivialising the modal operators.
Thus, we consider in the remainder the set $\Log$ of 39 non-equivalent combinations shown (as nodes or canonical representatives) in Figure~\ref{fig:pantheon}.

For $\Lvar \in \Log$,
we say that a neighbourhood frame
$\Fmc = ( \Wmc, \{\Nmc_i \}_{i \in J})$,
with $J = \{1, \ldots, n\}$, is an \emph{$L^{n}$ frame} iff its neighbourhood functions $\Nmc_i$, for $i \in J$, satisfy the \emph{$\Lvar$-condition},
obtained by  combining the conditions associated with letters in $\Lvar$.
For a class of neighbourhood frames $\Cmc$, the \emph{satisfiability in $\MLALC{n}$ on} (\emph{varying} or \emph{constant domain}, resp.) \emph{neighbourhood models based on a frame in $\Cmc$} is the problem of deciding whether an $\MLALC{n}$ formula is satisfied in a (varying or constant domain, resp.) neighbourhood model based on a frame in $\Cmc$.
%
Satisfiability in \emph{$\LnALC$ on} (\emph{varying} or \emph{constant domain}, respectively) \emph{neighbourhood models} is satisfiability in $\MLALC{n}$ on (varying or constant domain, resp.) neighbourhood models based on a frame in the class of $L^{n}$ frames.
Finally,  \emph{satisfiability  in $\KnALC{n}$ on}
(\emph{constant domain})
\emph{relational models} is satisfiability in $\MLALC{n}$  on
relational models based on any relational frame.

\begin{table*}
\begin{center}
\footnotesize
\begin{tabular}{l l l}
\toprule
\multirow{2}{*}{\emph{${\mathbf{E}}$-principle}} & $S \models C \equiv D$ implies $S \models \Box_{i} C \equiv \Box_{i} D$. \\
 & $S \models \varphi\leftrightarrow \psi$ implies $S \models \Box_{i} \varphi\leftrightarrow \Box_{i} \psi$. \\
\midrule
\multirow{2}{*}{\emph{${\mathbf{M}}$-principle}} & $S \models C \sqsubseteq D$ implies $S \models \Box_{i} C \sqsubseteq \Box_{i} D$. \\
 & $S \models \varphi\to \psi$ implies $S \models \Box_{i} \varphi\to \Box_{i} \psi$. \\
\midrule
\multirow{2}{*}{\emph{${\mathbf{C}}$-principle}} & $S \models \Box_{i} C \sqcap \Box_{i} D \sqsubseteq \Box_{i} (C \sqcap D)$. \\
 & $S \models \Box_{i} \varphi\land \Box_{i} \psi \to \Box_{i} (\varphi \land \psi)$. \\
\midrule
\multirow{2}{*}{\emph{${\mathbf{N}}$-principle}} & $S \models \top \sqsubseteq C$ implies $S \models \top \sqsubseteq \Box_{i} C$. \\
 & $S \models \varphi$ implies $S \models \Box_{i} \p$. \\
\bottomrule
\end{tabular}
\quad
\begin{tabular}{l l l}
\toprule
\multirow{2}{*}{\emph{${\mathbf{T}}$-principle}} & $S \models \Box_{i} C \sqsubseteq  C $. \\
& $S \models \Box_{i} \varphi \to \varphi$. \\
\midrule
\multirow{2}{*}{\emph{${\mathbf{D}}$-principle}} & $S \models \Box_{i} C \sqsubseteq \Diamond_{i} C $. \\
& $S \models \Box_{i} \varphi \to \Diamond_{i} \varphi$. \\
\midrule
\multirow{2}{*}{\emph{${\mathbf{P}}$-principle}} & $S \models \top \sqsubseteq  \lnot \Box_{i} \bot$. \\
& $S \models \lnot \Box_{i} \mathsf{false}$. \\
\midrule
\multirow{2}{*}{\emph{${\mathbf{Q}}$-principle}} & $S \models \top \sqsubseteq  \lnot \Box_{i} \top$. \\
& $S \models  \lnot \Box_{i} \mathsf{true}$. \\
\bottomrule
\end{tabular}
\end{center}
\caption{Principles over neighbourhood or relational frames and models $S$.}
\label{tab:principles}
\end{table*}

We now study the correspondence between 
conditions presented in Section~\ref{sec:sem} and
  the principles in
Table~\ref{tab:principles},
where
$S$ is either a
(neighbourhood or relational) frame 
or a (neighbourhood or relational) model
and the $L$-principle is obtained by suitably combining the basic principles.
We say that the $L$-principle holds in $S$ if the corresponding
expressions  in Table~\ref{tab:principles} are satisfied.
On the correspondence between the principles in
Table~\ref{tab:principles} and conditions over frames and models,
we have the following results (see e.g.~\cite{Pac}
for the propositional case).

%


\begin{restatable}{proposition}{PropCorresp}\label{prop:corresp}
	Given 
	a neighbourhood frame $\Fmc$, 
the $\Lvar$-principle holds in $\Fmc$ iff $\Fmc$ satisfies the $\Lvar$-condition.
\end{restatable}
%

\begin{restatable}{proposition}{PropValid}
\label{prop:valid}
%
The following statements hold.
\begin{enumerate}
	\item For a (varying or constant domain) neighbourhood model $\Mmc$, we have that if $\Mmc$ satisfies the $\Lvar$-condition, then the $\Lvar$-principle holds in $\Mmc$.
	However, in general, the converse is not true.
	\item For a relational frame $\Fmf$ and a relational model $\Mmf$ based on $\Fmf$, the $\mathbf{EMCN}$-principle holds in $\Mmf$, hence in $\Fmf$.
	Moreover, in $\Mmf$, hence in $\Fmf$,  the $\mathbf{D}$-principle holds iff the
	$\mathbf{P}$-principle holds, and the $\mathbf{Q}$-principle does not hold.
\end{enumerate}
\end{restatable}
%



\label{sec:reasonvardom}

\section{Tableaux for Formula Satisfiability}
\label{sec:tableaux}


We provide terminating, sound, and complete tableau algorithms to check satisfiability of formulas in varying domain neighbourhood models. The notation partly adheres to that of~\cite{GabEtAl03},
while the model construction in the soundness proof is based on the strategy of~\cite{DalHyp}.
In this section, we use concepts and formulas in \emph{negation normal form} 
(\emph{NNF}) and, for this reason, we consider all the logical connectives $\sqcup, \lor, \forall, \Diamond$ as primitive, rather than defined. 
For a concept or formula $\gamma$, we denote by $\dot{\lnot}\gamma$ the negation of $\gamma$ put in {NNF},
 defined as follows:
a concept is in \emph{NNF} if negation occurs in it only in front of concept names; a formula is in \emph{NNF} if all concepts in it are in NNF and negation occurs in the formula only in front of CIs or assertions of the form $r(a,b)$ (regarding
assertions of the form $A(a)$, we recall that a formula $\lnot \psi$, with $\psi = C(a)$, is equivalent to the assertion $D(a)$, with $D = \lnot C$).
Given an $\MLALC{n}$ formula $\p$, we assume without loss of generality that $\p$ is in NNF (using De Morgan laws) 
and it contains CIs only of the form $\top \sqsubseteq C$,
since $C \sqsubseteq D$ is equivalent to $\top \sqsubseteq \lnot C \sqcup D$). 
We denote by $\con(\p)$ and $\for(\p)$ the set of subconcepts  and subformulas of $\p$, respectively, and then we set
$\conneg(\p) = \con(\p) \cup \{ \dot{\lnot}C \mid C \in \con(\p) \} \cup  \{ \top \}$ and $\forneg(\p) =  \for(\p) \cup \{ \dot{\lnot}\psi \mid \psi \in \for(\p) \}$.
The sets $\rol(\p)$ and $\ind(\p)$ are, respectively, the sets of role names and of  individual names  occurring in $\p$.
Let 
$\fg(\p) = \forneg(\p) \cup \conneg(\p) \cup \rol(\p) \cup \ind(\p)$ be the \emph{fragment induced by $\varphi$}.


Moreover, let $\NV$ be a countable set of \emph{variables}, denoted by the letters $u, v$.
The \emph{terms for $\p$}, denoted by the letters $x, y$, are either individual names in $\ind(\p)$ or variables in $\NV$.
We assume that the set of terms for $\p$ is
strictly well-ordered by the relation $<$.
In addition, let $\mathsf{N_{L}}$ be a countable set of \emph{labels}.
Given an $\MLnALC$ formula $\p$ and a label $n \in \mathsf{N_{L}}$, an \emph{$n$-labelled constraint for $\p$} takes the form $n: \psi$, or $n: C(x)$, or $n: r(x, y)$, where $\psi \in \forneg(\p)$,
$x, y$ are terms for $\p$,
$C \in \conneg(\p)$, 
and $r \in \rol(\p)$.
For every $n \in \mathsf{N_{L}}$, an \emph{$n$-labelled constraint system for $\p$} is a set $S_{n}$ of $n$-labelled constraints for $\p$.
A \emph{labelled constraint for $\p$} is an $n$-labelled constraint for $\p$, for some $n \in \mathsf{N_{L}}$, and similarly for a \emph{labelled constraint system for $\p$}.
A
\emph{completion set $\T$ for $\p$} is a non-empty
union
of labelled constraint systems for $\p$,
and we set $\mathsf{L}_{\T} = \{ n \in \mathsf{N_{L}} \mid S_{n} \subseteq \T \}$.

About terms, we adopt the following terminology.
A {{term}} $x$ \emph{occurs in $S_{n}$} if $S_{n}$ contains $n$-labelled constraints of the form $n: C(x)$ or $n: r(\tau,\tau')$, where $\tau = x$, or $\tau' = x$, and $n \in \mathsf{N_{L}}$.
In addition, a variable $u$ is said to be \emph{fresh for $S_{n}$} if $u$ does not occur in $S_{n}$.
(These notions can be used with respect to $\T$, whenever $S_{n} \subseteq \T$).
Finally, given variables $u, v$ in an $n$-labelled constraint system $S_{n}$, we say that $u$ is \emph{blocked by $v$ in $S_{n}$} if $u > v$ and $\{ C \mid n : C(u) \in S_{n} \} \subseteq \{ C \mid n : C(v) \in S_{n} \}$.

A completion set $\T$ contains a \emph{clash} if,
for some $m \in \mathsf{N_{L}}$,
concept $C$, role $r$, and terms $x, y$, one of the following holds:
$\{m: (\top \sqsubseteq C), m: \lnot (\top \sqsubseteq C) \} \subseteq \T$;
or
$\{m: A(x), m: \lnot A(x)\} \subseteq \T$;
or
$\{m: r(x,y), m: \lnot r(x,y)\} \subseteq \T$.
A completion set that does not contain a clash is \emph{clash-free}.
%

For every $\mathit{L} \in \Log$, we associate to $\Lvar$ the set of \emph{$\LnALC$-rules} from Figure~\ref{fig:rules} (bottom part) containing
$\mathsf{R}_{\land}$,
$\mathsf{R}_{\sqcap}$,
$\mathsf{R}_{\lor}$,
$\mathsf{R}_{\sqcup}$,
$\mathsf{R}_{\exists}$,
$\mathsf{R}_{\forall}$,
$\mathsf{R}_{\sqsubseteq}$,
$\mathsf{R}_{\not\sqsubseteq}$,
$\mathsf{R}_{\mathit{L}}$,
and $\mathsf{R}_{\mathit{L}\mathbf{X}}$, for every $\mathbf{X}\in\{\mathbf{N,T,P,Q,D}\}$ such that $\mathbf{X}\in\Lvar$. 
Given $\mathit{L} \in \Log$, a completion set $\T$ is $\LnALC$-\emph{complete} if no 
$\LnALC$-rule is applicable to $\T$,
where $\gamma_{j}$ is either $\psi_{j} \in \forneg(\p)$ or $C_{j}(x_{j})$, with $C_{j} \in \conneg(\p)$, for $j = 1, \ldots, k$, and $\delta$ is either $\chi \in \forneg(\p)$ or $D(y)$, with $D \in \conneg(\p)$,
with respect to the \emph{application conditions} associated to each $\LnALC$-rule
from Figure~\ref{fig:rules} (top part).
%
\begin{figure*}
\scriptsize
\centering

Application conditions

\begin{enumerate}[leftmargin=*, align=left]
	\item[$(\mathsf{R}_{\land})$] $\{n :\psi,n : \chi\} \not\subseteq \T$;
	\item[$(\mathsf{R}_{\sqcap})$] $\{n : C(x),n : D(x)\} \not\subseteq \T$;
\item[$(\mathsf{R}_{\lor})$] $\{n :\psi, n : \chi \} \cap \T = \emptyset$;
	\item[$(\mathsf{R}_{\sqcup})$] $\{n : C(x),n : D(x)\} \cap \T = \emptyset$;
	\item[$(\mathsf{R}_{\exists})$]
	$x$ is not blocked by any variable in $S_{n}$, there is no $y$ such that  $\{n : r(x,y), n : C(y)\} \subseteq \T$, and $v$ is the $<$-minimal
	variable  fresh for $S_{n}$;
	\item[$(\mathsf{R}_{\forall})$] $n : C(y) \notin\T$;
	\item[$(\mathsf{R}_{\sqsubseteq})$]
	either
	$x$ occurs in
	$S_{n}$
	and $n : C(x) \notin\T$;
or no term occurs in $S_{n}$ and $x$ is the $<$-minimal variable fresh for $S_{n}$;
\item[$(\mathsf{R}_{\not\sqsubseteq})$]
there is no $x$ such that $n : \dnot C(x) \in \T$ and $v$ is
the $<$-minimal
variable 
fresh
for $S_{n}$;
\item[$(\mathsf{R}_{\mathit{L}})$]
$m$ is fresh
for $\T$, and there is no $o\in \mathsf{N_{L}}$ such that
$\{ o: \gamma_1, \ldots, o: \gamma_k, o: \delta\}\subseteq \T$, or $\{ o: \dot{\lnot}\gamma_j, o: \dot{\lnot}\delta\}\subseteq \T$, for some $j\leq k$;
\item[$(\mathsf{R}_{\mathit{L}\mathbf{N}})$]
$m$ is fresh for $\T$, and there is no $o\in \mathsf{N_{L}}$ such that
$o: \gamma\in \T$;
\item[$(\mathsf{R}_{\mathit{L}\mathbf{T}})$]
$n: \gamma\not\in \T$;
\item[$(\mathsf{R}_{\mathit{L}\mathbf{P}})$]
$m$ is fresh for $\T$, and there is no $o\in \mathsf{N_{L}}$ such that
$\{ o: \gamma_1, \ldots, o: \gamma_k\}\subseteq \T$;
\item[$(\mathsf{R}_{\mathit{L}\mathbf{Q}})$]
$m$ is fresh for $\T$, and there is no $o\in \mathsf{N_{L}}$  such that
$o: \dot{\lnot}\gamma_j\in \T$, for some $j\leq k$;
\item[$(\mathsf{R}_{\mathit{L}\mathbf{D}})$]
$m$ is fresh for $\T$, and there is no $o\in \mathsf{N_{L}}$ such that
$\{ o: \gamma_1, \ldots, o: \gamma_k, o: \delta_1, \ldots, o: \delta_k\}\subseteq \T$, or $\{ o: \dot{\lnot}\gamma_j, o: \dot{\lnot}\delta_\ell\}\subseteq \T$, for some $j\leq k$, $\ell\leq h$.
\end{enumerate}
%

Rules

\bigskip

\ax{$n: \psi \land \chi$}
\solidLine
\llab{($\mathsf{R}_{\land}$)}
\uinf{$n: \psi$ {,} $n:\chi$}
\disp
\hfill
\ax{$n: \psi \lor \chi$}
\solidLine
\llab{($\mathsf{R}_{\lor}$)}
\uinf{\begin{tabular}{c?c}
$n:\psi$ & $n:\chi$ \\
\end{tabular}}
\disp
\hfill
\ax{$n: C \sqcap D (x)$}
\solidLine
\llab{($\mathsf{R}_{\sqcap}$)}
\uinf{$n: C(x)${,}  $n: D(x)$}
\disp
\hfill
\ax{$n: C \sqcup D (x)$}
\solidLine
\llab{($\mathsf{R}_{\sqcup}$)}
\uinf{\begin{tabular}{c?c}
$n: C(x)$ & $n: D(x)$ \\
\end{tabular}}
\disp

\vspace{0.3cm}
\ax{$n: \exists r.C(x)$}
\solidLine
\llab{($\mathsf{R}_{\exists}$)}
\uinf{ $n: r (x , v)${,} $n: C(v)$ }
\disp
\hspace{1cm}
\ax{$n: \forall r.C(x)${,} $n: r (x , y)$}
\solidLine
\llab{($\mathsf{R}_{\forall}$)}
\uinf{$n: C(y)$}
\disp
\\
\vspace{0.3cm}
\ax{$n: \top \sqsubseteq C$}
\solidLine
\llab{($\mathsf{R}_{\sqsubseteq}$)}
\uinf{$n: C(x)$}
\disp
\hspace{2.5cm}
\ax{$n: \lnot (\top \sqsubseteq C)$}
\solidLine
\llab{($\mathsf{R}_{\not\sqsubseteq}$)}
\uinf{ $n:\dot{\lnot}C(v)$}
\disp

\vspace{0.3cm}
\ax{$n: \Box_i\gamma_1 ${,} $ \ldots ${,} $ n: \Box_i\gamma_k ${,} $ n: \Diamond_{i}\delta$}
\solidLine
\llab{($\mathsf{R}_{\mathit{L}}$)}
\uinf{\begin{tabular}{c?c?c?c}
	\vspace{-0.2cm}
	$m: \gamma_1 ${,} $ \ldots ${,} $ m: \gamma_k $  {,} $ m: \delta$ & 
	$ m: \dot{\lnot}\gamma_1 ${,} $ m: \dot{\lnot}\delta$ &
	$\ \dots$ &
	$ m: \dot{\lnot}\gamma_k ${,} $ m: \dot{\lnot}\delta$ \\
	\multicolumn{4}{c}{\ \quad\qquad\qquad\qquad\qquad\qquad\  $\underbrace{{\color{white}{\qquad\qquad\qquad\qquad - - \qquad\qquad\qquad\quad\ \  }}}_{\text{if } \mathbf{M} \not \in L}$} 
\end{tabular}}
\disp
\hfill
\ax{$n: \Diamond_i\gamma$}
\solidLine
\llab{($\mathsf{R}_{\mathit{L}\mathbf{N}}$)}
\uinf{$m: \gamma$}
\disp

\vspace{0.3cm}
\ax{$n: \Box_i\gamma$}
\solidLine
\llab{($\mathsf{R}_{\mathit{L}\mathbf{T}}$)}
\uinf{$n: \gamma$}
\disp
\hfill
\ax{$n: \Box_i\gamma_1${,} $\ldots${,} $ n: \Box_i\gamma_k$}
\solidLine
\llab{($\mathsf{R}_{\mathit{L}\mathbf{P}}$)}
\uinf{$ m: \gamma_1${,} $ \ldots${,} $ m: \gamma_k$}
\disp
\hfill
\ax{$n: \Box_i\gamma_1${,} $\ldots${,} $ n: \Box_i\gamma_k$}
\solidLine
\llab{($\mathsf{R}_{\mathit{L}\mathbf{Q}}$)}
\uinf{\begin{tabular}{c?c?c}
	$m: \dot{\lnot} \gamma_1$ &
	$\ \dots$ &
	$ m: \dot{\lnot}\gamma_k$ \\
\end{tabular}}
\disp

\vspace{0.3cm}
\ax{$n: \Box_i\gamma_1$  {,} $\ldots$  {,} $ n: \Box_i\gamma_k$  {,} $ n: \Box_i\delta_1$  {,} $\ldots$  {,} $ n: \Box_i\delta_h$}
\solidLine
\llab{($\mathsf{R}_{\mathit{L}\mathbf{D}}$)}
\uinf{\begin{tabular}{c?c?c?c?c?c?c?c}
	$ m: \gamma_1${,} $ \ldots${,} $ m: \gamma_k${,} & 
	$m: \dot{\lnot}\gamma_1${,} &
	$\ \dots$ &
	$ m: \dot{\lnot}\gamma_k ${,} &
	$\ \dots$ &
	$ m: \dot{\lnot}\gamma_1 ${,} & 
	$\ \dots$ &
	$ m: \dot{\lnot}\gamma_k ${,}
	\\
	\vspace{-0.2cm}
	$ m: \delta_1${,} $\ldots${,} $ m: \delta_h$ & 
	$ m: \dot{\lnot}\delta_{1}$ &
	&
	$ m: \dot{\lnot}\delta_1$ &
	&
	$ m: \dot{\lnot}\delta_h$ & 
	&
	$ m: \dot{\lnot}\delta_h$\\
	\multicolumn{8}{c}{\qquad\qquad\qquad\qquad\ \ \ $\underbrace{{\color{white}{\qquad\qquad\qquad\qquad - - \qquad\qquad\qquad\qquad\qquad\qquad\qquad }}}_{\text{if } \mathbf{M} \not \in L}$} 
\end{tabular}}
\disp

\caption{\label{fig:rules} $\LnALC$-rules,
where $k, h \geq 1$ if $\mathbf{C}\in\mathit{L}$ and $k = h = 1$ if $\mathbf{C}\not\in\mathit{L}$.
In the rules $\mathsf{R}_{\mathit{L}}$ and $\mathsf{R}_{\mathit{L}\mathbf{D}}$,
the number of possible 
expansions depend on whether $\mathbf{M}\in\mathit{L}$:
if $\mathbf{M}\in\mathit{L}$ only the first expansion is possible, 
if $\mathbf{M}\notin\mathit{L}$ all other expansions are also possible.
}
\end{figure*}
%
%
The $\LnALC$-rules essentially state how to extend a completion set on the basis of the information contained in it.
Branching rules entail a \emph{non-deterministic choice} in the  completion set expansion.


For each $\Lvar\in \Log$,
we now present a
tableau-based non-deterministic decision procedure 
for the $\LnALC$ formula satisfiability problem on varying domain neighbourhood models,
based on
Algorithm~\ref{alg:tableau}
(simply referred to as
\emph{$\LnALC$ tableau algorithm}).
We have that a formula $\p$ is $\LnALC$ satisfiable if and only if there exists at least one execution of the
$\LnALC$ tableau algorithm
that constructs an $\LnALC$-complete and clash-free completion set for $\p$.
This non-deterministic algorithm
gives priority to non-generating $\LnALC$-rules,
i.e., those that do not introduce new variables or labels,
with respect to generating ones,
so to minimise the size of the completion set constructed by its application,
and
terminates in exponential time 
for every formula $\p$.
Thus, we obtain the following.


\begin{theorem}
	\label{thm:upperbound}
	Satisfiability in $\LnALC$  on varying domain neighbourhood models is decidable in $\NExpTime$.
\end{theorem}

\begin{algorithm}[t]
  \KwIn{the initial completion set $\T := \{0 : \p\}$ of an $\MLALC{n}$ formula $\p$ in NNF.}
  \KwOut{a completion set for $\p$, extending the initial one, that either contains a clash, or is complete and clash-free.}
    \BlankLine
  \While{$\T$ is clash-free and not $\LnALC$-complete}{
  	\uIf{a rule $\mathsf{R} \in\{\mathsf{R}_{\land}, \mathsf{R}_{\lor}, \mathsf{R}_{\sqcap}, \mathsf{R}_{\sqcup}, \mathsf{R}_{\forall}, \mathsf{R}_{\sqsubseteq}, \mathsf{R}_{\mathit{L}\mathbf{T}} \}$ is applicable to $\T$}{apply $\mathsf{R}$ to $\T$\;}
      \uElseIf{a rule $\mathsf{R} \in\{\mathsf{R}_{\exists},  \mathsf{R}_{\mathit{L}}, \mathsf{R}_{\mathit{L}\mathbf{N}}, \mathsf{R}_{\mathit{L}\mathbf{P}}, \mathsf{R}_{\mathit{L}\mathbf{Q}}, \mathsf{R}_{\mathit{L}\mathbf{D}} \}$ is applicable to $\T$}{apply $\mathsf{R}$ to $\T$\;}
    }
\caption{
$\LnALC$ tableau algorithm
on varying domain neighbourhood models
for $\p$
 }
\label{alg:tableau}
\end{algorithm}

%
%
%
%

%

As an immediate consequence of the 
correctness of the tableau we also obtain 
a (constructive) proof of the following kind of \emph{exponential model property}.
\begin{restatable}{corollary}{Fmp}
Every $\LnALC$ satisfiable formula $\p$ has a model 
with at most $p(|\fg(\p)|)$ worlds,
if $\mathbf{C}\notin\Lvar$, 
and at most $2^{q(|\fg(\p)|)}$) worlds, 
if $\mathbf{C}\in\Lvar$,
each of them having a domain with at most 
$2^{r(|\fg(\p)|)}$ elements, with $p$, $q$, $r$ polynomial functions.
\end{restatable}
\section{Fragments without Modalised Concepts}
\label{sec:fragvardom}
Here we study fragments of $\MLnALC$ without modalised concepts.
An \emph{$\MLnALCg$ formula} is defined similarly to the $\MLnALC$ case,
by disallowing modalised concepts.
Given
$\mathit{L} \in \Log$,
 \emph{satisfiability in $\LnALCg$ on varying} 
  \emph{domain neighbourhood models} is $\MLnALCg$ satisfiability 
  on varying 
  domain neighbourhood models based on neighbourhood frames in the respective class for $\mathit{L}$.
%
An \emph{$\MLn$ formula}, instead, is defined analogously to $\MLnALCg$, except that we build it from the standard propositional (rather than $\ALC$) language over a countable set of \emph{propositional letters} $\mathsf{N_{P}}$, disjoint from \NC, \NR, and \NI.
The semantics of
$\MLn$ formulas
is given in terms of \emph{propositional 
neighbourhood models} (or simply \emph{models}) $\Mmc^{\sf P} = (\Wmc, \{ \Nmc_{i} \}_{i \in J}, \Vmc)$,
where $(\Wmc, \{ \Nmc_{i} \}_{i \in J})$ is a neighbourhood frame,
with $J = \{ 1, \ldots, n \}$ in the following,
and $\Vmc: \NPr \rightarrow 2^{\Wmc}$ is a function 
mapping propositional letters
to
sets of worlds
(see~\cite{Che,Var2}).
\emph{Satisfiability in $\ensuremath{\smash{\mathit{L}^{n}}}$} is satisfiability in $\MLn$ on propositional neighbourhood models based on neighbourhood frames in the respective class for $\mathit{L}$.
A propositional neighbourhood model based on a neighbourhood frame in the respective class for $\mathit{L}$ is called \emph{$\mathit{L}^{n}$ model}.

\newcommand{\setsymbols}{\ensuremath{\Sigma}\xspace}

We
prove
tight complexity results 
for  
$\LnALCg$ satisfiability on varying domain neighbourhood models, where $L\in \Log$,
using the notion of 
a propositional abstraction of a formula 
(as in, e.g.,~\cite{Baader:2012:LOD:2287718.2287721}).
Here, one can separate the satisfiability test into two parts, 
one for the description logic dimension and 
one for the \neighborhood frame dimension.
%
For an $\MLnALCg$ formula $\varphi$, the 
\emph{propositional abstraction} $\prop{\varphi}$ is  
the result of replacing each \ALC atom $\pi$ in $\varphi$ by 
a propositional variable $p_{\pi} \in \NPr$.
Define the set $\setsymbols_\varphi = \{p_{\elaxiom}\in\NPr\mid \elaxiom \text{ is an \ALC atom in }
\varphi \}$.
%
A (propositional neighbourhood) $L^{n}$ model 
$\propmodel = (\Wmc, \{ \Nmc_{i} \}_{i \in J}, \Vmc)$
is \emph{$\setsymbols_\varphi$-consistent}
if,
for all $w\in \Wmc$,
the following \ALC formula is satisfiable:
$
\alcform = \bigwedge_{p_{\elaxiom}\in \formtp{\varphi}} {\elaxiom} \ \wedge \bigwedge_{p_{\elaxiom} \in
	\setsymbols_\varphi \setminus \formtp{\varphi}}
\neg {\elaxiom},
$
where
$\formtp{\varphi} = \{p_{\elaxiom} \in \setsymbols_\varphi \mid w\in \Vmc(p_{\elaxiom})\}$.
%
We  formalise the connection between 
$\LnALCg$ satisfiable formulas and their propositional abstractions 
with the following lemma.


\begin{restatable}{lemma}{LemmapropL}\label{lem:propL}
A
formula $\varphi$ is
$\LnALCg$
satisfiable
on varying domain neighbourhood models
iff
$\prop{\varphi}$ is satisfied in a $\setsymbols_\varphi$-consistent $L^{n}$ model.  
\end{restatable}
%


We now introduce definitions and notation used to prove our complexity result on fragments without modalised concepts.
Let $\setsymbols = \{p_{\elaxiom}\in\NPr\mid \elaxiom \text{ is an \ALC }$ $\text{atom in }\varphi 
 \}$,  for a fixed but arbitrary  $\MLnALCg$ formula $\varphi$,
and let $\phi$ be an $\MLn$ formula built from symbols in $\setsymbols$.
We
denote by  ${\sf sub}(\phi)$ the set 
of subformulas of $\phi$ closed under single negation.  
A \emph{valuation} for $\phi$ 
is a function $\nu: {\sf sub} (\phi)\rightarrow \{0,1\}$ that 
satisfies the   conditions:
(1) for all $\neg\psi\in {\sf sub} (\phi)$,
$\nu(\psi)=1$ iff $\nu(\neg\psi)=0$;
(2) for all $\psi_1\wedge \psi_2\in {\sf sub} (\phi)$,
$\nu(\psi_1\wedge \psi_2) = 1$ iff $\nu(\psi_1) = 1$
and $\nu(\psi_2) = 1$; 
and (3) $\nu(\phi)=1$. 
A valuation $\nu$ for $\phi$
is \emph{$\setsymbols$-consistent}
if
the following \ALC formula is satisfiable:
$
\textstyle\bigwedge_{\nu(p_\elaxiom)=1} \ {\elaxiom} \ \wedge \
\bigwedge_{\nu(p_\elaxiom)=0}\ \neg {\elaxiom},
$
where $p_\elaxiom\in\setsymbols$.  
Lemma~\ref{lem:proplemmaL} establishes that satisfiability of 
$\phi$ in a $\setsymbols$-consistent model is characterised  by the existence of a $\setsymbols$-consistent valuation satisfying suitable properties.
In the following, we use $\falseprop$ as an abbreviation for $p \land \neg p$, for a fixed but arbitrary   $p \in \mathsf{N_{P}}$. 


\begin{restatable}{lemma}{Lemmapropvardi}\label{lem:proplemmaL}
	Given $\Lvar$ and  an $\MLn$ formula $\phi$ built from symbols in $\setsymbols$ (defined as above), 
	let:
	\[
	\boldsymbol{\kappa} =
	\begin{cases}
		| {\sf sub}({\phi}) | , & \text{if $\mathbf{C} \in \Lvar$} \\
		1, & \text{if $\mathbf{C} \not \in \Lvar$}
	\end{cases}.
	\]
	A formula $\phi$ is satisfied in a $\setsymbols$-consistent $\Lvar^{n}$ 
	model
	iff
	there is   a $\setsymbols$-consistent valuation \valuation 
	for $\phi$ 
	such that,
	for every $1 \leq k \leq \boldsymbol{\kappa}$,
	if $\B_i\psi_1, \dots, \B_i\psi_k, \B_i\chi\in{\sf sub}(\phi)$,
	$\valuation(\B_i\psi_j)=1$ for all $1\leq j \leq k$,
	and $\valuation(\B_i\chi)=0$, then
\begin{enumerate}
\item
$
	(\bigwedge^{k}_{j=1}\psi_j\wedge\neg\chi) \vee \boldsymbol\vartheta
$
	is satisfied in a $\setsymbols$-consistent $\Lvar^{n}$ 
	model, 
	where:
$\boldsymbol\vartheta = \falseprop$, if $\mathbf{M}\in\Lvar$;
$\boldsymbol\vartheta = \bigvee^{k}_{j=1} (\neg\psi_j\wedge\chi)$, if $\mathbf{M}\not\in\Lvar$;
	and 
\item
	for $\mathbf{X}\in\{\mathbf{N,T,P,Q,D}\}$,
	if $\mathbf{X}\in\Lvar$, then $\nu$ satisfies the condition $(\mathbf{X})$ below, for every $1 \leq k, h \leq \boldsymbol{\kappa}$:
	\begin{itemize}
		%
		\item[($\mathbf{N}$)] if $\B_i\psi\in{\sf sub}(\phi)$ and
		$\valuation(\B_i\psi)=0$, then $\neg \psi$ 
		is satisfied in a $\setsymbols$-consistent $\Lvar^{n}$ model; 
		
		\item[($\mathbf{T}$)] if $\B_i\psi\in{\sf sub}(\phi)$ and 
		$\valuation(\B_i\psi)=1$   then
		$\valuation(\psi)=1$;
		
		\item[($\mathbf{P}$)] if $\B_i\psi_1, \dots, \B_i\psi_k\in{\sf sub}(\phi)$ and
		$\valuation(\B_i\psi_j)=1$ for all $1\leq j \leq k$, 
		then $\bigwedge^{k}_{j=1}\psi_j$
		is satisfied in a $\setsymbols$-consistent $\Lvar^{n}$ model; 
		
		\item[($\mathbf{Q}$)] if $\B_i\psi_1, \dots, \B_i\psi_k\in{\sf sub}(\phi)$ and
		$\valuation(\B_i\psi_j)=1$ for all $1\leq j \leq k$, 
		then $\bigvee^{k}_{j=1}\neg\psi_j$
		is satisfied in a $\setsymbols$-consistent $\Lvar^{n}$ model; 
		
		\item[($\mathbf{D}$)] if $\B_i\psi_1, \dots, \B_i\psi_k, \B_i\chi_1, \dots, \B_i\chi_h\in{\sf sub}(\phi)$,
		$\valuation(\B_i\psi_j)=1$ for all $1\leq j \leq k$, and
		$\valuation(\B_i\chi_\ell)=1$ for all $1\leq \ell \leq h$, 
		then $(\bigwedge^{k}_{j=1}\psi_j \land \bigwedge^{h}_{\ell=1}\chi_\ell) \vee \boldsymbol\eta$
		is satisfied in a $\setsymbols$-consistent $\Lvar^{n}$ model,
		where:
		\[
		\boldsymbol\eta =
		\begin{cases}
			\falseprop, & \text{if $\mathbf{M}\in\Lvar$} \\
			\neg(\bigwedge^{k}_{j=1}\psi_j) \land \neg(\bigwedge^{h}_{\ell=1}\chi_\ell), & \text{if $\mathbf{M}\not\in\Lvar$}
		\end{cases}.
		\]
	\end{itemize}
\end{enumerate}
	%
\end{restatable}

By using Lemmas~\ref{lem:propL}-\ref{lem:proplemmaL},
the following theorem provides a procedure that runs in exponential time to check $\LnALCg$ satisfiability on varying domains.
Since $\ALC$ formula satisfiability is already $\ExpTime$-hard, our upper bound is tight.

\begin{restatable}{theorem}{Satfragvardomexp}
\label{thm:satfragvardomexp}
	Satisfiability in $\LnALCg$   on varying domain neighbourhood models is \ExpTime-complete.
\end{restatable}
%



\section{Reasoning on Constant Domain}
\label{sec:reasoncondom}

%
\label{sec:relation}

	
We now study the complexity of the formula satisfiability problem in $\EnALC{n}$ and $\MnALC{n}$
on constant domain neighbourhood models.
We provide a $\NExpTime$ upper bound for satisfiability in $\EnALC{n}$ and $\MnALC{n}$ by using a reduction, lifted from the propositional case, to multi-modal $\KnALC{m}$.
The translation $\cdot\tr$ from \MLALC{n} to \MLALC{3n} is defined as~\cite{KraWol,GasHer}: 
	$A\tr = A$,
	$(\lnot C)\tr = \lnot C\tr$,
	$(C \sqcap D)\tr = C\tr \sqcap D\tr$,
	$(\exists \role.C)\tr = \exists \role.C\tr$;
$(C(a))\tr = C\tr(a)$,
$(r(a,b))\tr = r(a,b)$,
$(C \sqsubseteq D)\tr = C\tr \sqsubseteq D\tr$,
$(\lnot \psi)\tr  = \lnot \psi\tr$,
$(\psi \land \chi)\tr  = \psi\tr \land \chi\tr$;
$(\B_{i} \gamma)\tr = \D_{i_{1}} (\B_{i_{2}} \gamma\tr \circ \B_{i_{3}} \lnot \gamma\tr)$;
where $A \in \NC$, $\role\in\NR$,
$\gamma$ is either an $\MLALC{n}$ concept or formula,
and $\circ \in \{ \sqcap, \land \}$ accordingly.
Using this translation, one can show that 
 satisfiability on neighbourhood models
is reducible to 
satisfiability on the
relational
models~\cite{KraWol,GasHer}.
Since satisfiability in $\KnALC{3n}$ constant domain relational models is
$\NExpTime$-complete~\cite[Theorem 15.15]{GabEtAl03},
 we obtain the following complexity result.

\begin{restatable}{theorem}{Theoremcomplealc}\label{theor:complealc}
Satisfiability in $\EnALC{n}$  on constant domain neighbourhood models is decidable in $\NExpTime$.
\end{restatable}
The translation $\cdot\ttr$ from $\MLALC{n}$ to $\MLALC{2n}$ is defined as $\cdot\tr$ on all concepts and formulas, except for the modalised concepts or formulas $\gamma$:
	$(\B_{i} \gamma)\ttr = \D_{i_1} \B_{i_2} \gamma\ttr$.
We obtain an upper bound analogous to the one for $\EnALC{n}$ by a reduction 
of the formula satisfiability problem for $\MnALC{n}$
to the $\KnALC{2n}$ one~\cite{KraWol,GasHer,GabEtAl03}.

\begin{restatable}{theorem}{Theoremcomplmalc}\label{theor:complmalc}
Satisfiability in $\MnALC{n}$   on constant domain neighbourhood models is decidable in $\NExpTime$.
\end{restatable}

\section{Discussion}
\label{sec:discuss}

We investigated reasoning in non-normal modal description logics,
focussing on:
$(i)$
tableaux algorithms to check satisfiability of multi-modal description logics formulas in varying domain neighbourhood models based on classes of frames
for
39 different
non-normal
systems;
$(ii)$ complexity of satisfiability restricted to fragments with modal operators applied only over formulas,
and interpreted on varying domain models;
$(iii)$ preliminary reduction of formula satisfiability for two non-normal modal description logics to satisfiability in the standard relational semantics on a constant domain.
We now discuss possible future work.
%
%

First, we intend to devise tableaux for formula satisfiability on neighbourhood models with constant domain, by solving the problem of newly introduced variables that do not occur in other previously expanded labelled constraints systems.
For instance, by applying the $\mathbf{M}^{n}_{\ALC}$-rules to the $n$-labelled constraint system $S_{n} = \{ n : \Diamond_{i} \exists r. A(x), \Box_{i} \lnot A(x) \}$, we get the $m$-labelled constraint system $S_{m} = \{ m : \exists r. A(x),  m : \lnot A(x), m : r(x,y), m : A(y) \}$.
The fresh variable $y$ in $S_{m}$ does not allow for the direct extraction of a constant domain model,
as no object in the domain of the world associated with $S_{n}$ would be capable of representing $y$ correctly.
An alternative approach involves \emph{quasimodels}~\cite{GabEtAl03}, to characterise satisfiability on constant domain
models in terms of structures representing ``abstractions'' of the actual models of a formula.
Objects across worlds can be represented by means of \emph{runs}, i.e., functions to guarantee
their modal properties and the constant domain assumption.
A similar strategy is presented in~\cite{SeyErd09,SeyJam09,SeyJam10},
where the definition of runs (which is not carried out in detail) involves the introduction of suitable world ``copies''.
We conjecture that a quasimodel-based approach with \emph{marked variables}, as
illustrated in~\cite{GabEtAl03}, can also be adopted to solve the constant domain model extraction issue.

Moreover, we aim at tight complexities for $\LnALC$ satisfiability, both in varying and in constant domain 
models.
While $\ALC$ formula satisfiability is $\ExpTime$-complete, it is unclear whether the upper bound for $\LnALC$ on varying or constant domain
neighbourhood
models can be improved to $\ExpTime$-membership, for any $\Lvar \in \Log$.
At the propositional level, the formula satisfiability problem for the systems based on the $L$-condition, with $\mathbf{C} \not \in L$, is $\NP$-complete, rising to $\PSpace$ if the $\mathbf{C}$-condition is included~\cite{Var2}.
For normal modal description logics, instead, the (tight) $\NExpTime$-hardness results are based on complexity proofs of \emph{product logics} over relational product frames~\cite{GabEtAl03}, and cannot be immediately adapted to neighbourhood semantics, where an analogous notion of product is not yet well understood.
Nonetheless, we conjecture that the $\NExpTime$-hardness known for,
e.g., $\mathbf{K}_\mathcal{ALC}$ on constant domain relational models, also holds in the neighbourhood case, at least in presence of the $\mathbf{C}$-condition.


Finally, we plan to study: non-normal modal description logics in \emph{coalitional} and \emph{strategic} settings~\cite{Pau,Tro,SeyJam09}, with an interplay between abilities and powers of \emph{groups} of agents, rather than single ones;
additional description logics constructs (e.g. \emph{nominals}, \emph{inverse roles}, or \emph{number restrictions}~\cite{BaaEtAl17}); and \emph{interactions between modalities}, with axioms expressing e.g. that an agent \emph{can do} anything they \emph{actually do}, by means of formulas of the form $\mathbb{D}_{i}C \sqsubseteq \mathbb{C}_{i}C$ or $\mathbb{D}_{i}\varphi \to \mathbb{C}_{i}\varphi$.

\section*{Acknowledgements}

This research has been partially supported
by the Province of Bolzano and DFG through the project D2G2 (DFG grant n.~500249124).
Andrea Mazzullo acknowledges the support of the MUR PNRR project FAIR - Future AI Research (PE00000013) funded by the NextGenerationEU.
Ana Ozaki is supported by the Research Council of Norway, project number 316022.

\bibliographystyle{splncs04}
\bibliography{bib_nnmdl}

\newpage

\appendix

\section{Modelling Scenario
}
\label{sec:model}

In the following,
we present the modelling of an example scenario in the classic domain of multi-agent purchase choreography~\cite{MonEtAl10}.
The aim is twofold.
First, it displays some of the limitations of modalities defined on relational frames, motivating the adoption of neighbourhood semantics in modal extensions of description logics.
Second, it illustrates the expressivity of (non-normal) modal description logic languages, showing interactions between modalities and the constructs of the standard description logic $\ALC$.

Our multi-agent setting involves
a customer $\mathit{c}$,
a marketplace
$\mathit{m}$,
a seller $\mathit{s}$,
and
a warehouse $\mathit{w}$,
with
agency operators $\mathbb{D}_i$ and $\mathbb{C}_i$, for $i \in \{ c, m, s, w \}$, read as `agent $i$ does/makes' and `agent $i$ can do/make', respectively~\cite{Elg,GovernatoriRotolo}. 
Concept names $\mathsf{Ord}$, $\mathsf{Prod}$, and $\mathsf{InCatal}$ are used to represent, respectively, orders, products, and the class of objects displayed as in-catalogue, while $\mathsf{req}$ is a role name for the request relation.
The formula
$
\label{eq:1eq}
\mathsf{Ord} \equiv \mathbb{D}_{c}\exists \mathsf{req}.(\mathsf{Prod} \sqcap \mathsf{InCatal} )
$
defines an order $\mathsf{Ord}$ as a request made by customer $c$ of an in-catalogue
product.
Using the concept name $\mathsf{Confirm}$ to represent the class of objects that are confirmed, we also have that
$
\label{eq:2eq}
\exists \mathsf{req}.(\mathsf{Prod} \sqcap \mathsf{InCatal} ) \sqsubseteq \mathsf{Confirm} \sqcup \lnot \mathsf{Confirm},
$
meaning that any request of an in-catalogue product is either confirmed or not confirmed.
However, relational semantics validates the so-called \emph{$\mathbf{M}$-principle} (often called \emph{monotonicity}) as well, according to which $C \sqsubseteq D$ always entails $\mathbb{D}_{c} C \sqsubseteq \mathbb{D}_{c} D$, for any concepts $C, D$.
Thus,
from the $\mathbf{M}$-principle 
we would
obtain
$
\label{eq:monconc}
\mathsf{Ord} \sqsubseteq \mathbb{D}_{c} ( \mathsf{Confirm} \sqcup \lnot \mathsf{Confirm} ),
$
meaning that any order
is made confirmed or not confirmed by $c$. This is an unwanted conclusion in our agency-based scenario, since customers' actions should be unrelated to any aspect of order confirmation.

Moreover, assume that the concept name $\mathsf{SubmitOrd}$ stands for the class of submitted orders, and that $\mathsf{PartConf}$ and $\mathsf{Reject}$ are used in our knowledge base to represent, respectively, the partially confirmed and the rejected
entities.
Now consider the formula
$
\label{eq:agglprem}
\mathsf{SubmitOrd} \sqsubseteq \mathbb{C}_{s}\mathsf{Confirm} \sqcap \mathbb{C}_{s}\mathsf{PartConf} \sqcap \mathbb{C}_{s}\mathsf{Reject},
$
stating that a submitted order can be confirmed, can be partially confirmed, and can be rejected by the seller $s$.
On relational frames, however, we have that
$ \mathbb{C}_{s} C \sqcap \mathbb{C}_{s} D \sqsubseteq \mathbb{C}_{s}( C \sqcap D ) $
 is a valid formula, for any concepts $C, D$, known as the \emph{$\mathbf{C}$-principle} (or \emph{agglomeration}).
 Therefore, by 
 the $\mathbf{C}$-principle, under relational semantics we would be forced to conclude that
$
\label{eq:agglconc}
\mathsf{SubmitOrd} \sqsubseteq \mathbb{C}_{s}(\mathsf{Confirm} \sqcap \mathsf{PartConf} \sqcap \mathsf{Reject}),
$
meaning that any submitted order is such that the seller $s$ has the ability to make it confirmed, partially confirmed, and rejected, all \emph{at once}, which is unreasonable.


Consider now the formula
$
\label{eq:necprem}
\top \sqsubseteq \mathsf{Confirm} \sqcup \lnot \mathsf{Confirm},
$
i.e., the truism stating that anything is either confirmed or not confirmed.
By the so called \emph{$\mathbf{N}$-principle} (or \emph{necessitation}) of relational semantics, we have that if $\top \sqsubseteq C$ is valid on relational frames, then $\top \sqsubseteq \mathbb{D}_{c} C$ holds as well, for any concept $C$.
Thus, from
the $\mathbf{N}$-principle of relational semantics
it would follow
that
$
\label{eq:necprem}
\top \sqsubseteq \mathbb{D}_{c} (\mathsf{Confirm} \sqcup \lnot \mathsf{Confirm}),
$
thereby forcing us to the consequence that every object is made by customer $c$ to be either confirmed or not confirmed, and hence leading again to an unreasonable connection between customer's actions and confirmation of orders.
In fact, since customer $c$ plays no role in confirmation actions, it is sensible to assume that, for any object of the domain, it is not the case that $c$ makes it confirmed or not confirmed. This can be achieved by the formula
$
\label{eq:qprinc}
\top \sqsubseteq \lnot \mathbb{D}_{c} (\mathsf{Confirm} \sqcup \lnot \mathsf{Confirm}),
$
an instance of a principle sometimes known as the \emph{$\mathbf{Q}$-principle}, which is \emph{unsatisfiable} in relational frames, while admissible over neighbourhood ones.

Finally, we consider additional principles that can be adopted both in relational and neighbourhood semantics.
The formula
$
\label{eq:tprinc}
\mathbb{D}_{w} \mathsf{Avail}  \sqsubseteq \mathsf{Avail}
$
states that
anything that is \emph{made} available by the warehouse $w$ is \emph{actually} available.
This is an instance of the so-called \emph{$\mathbf{T}$-principle} (also \emph{factivity principle}), well-known in modal logic, particularly for its epistemic applications (if an agent \emph{knows} something, it has to be true).
Both in relational and neighbourhood semantics, the $\mathbf{T}$-principle entails the so-called \emph{$\mathbf{D}$-principle}.
This is instantiated by
$
\label{eq:dprinc}
\mathbb{D}_{w} \mathsf{Avail}  \sqsubseteq \lnot \mathbb{D}_{w} \lnot \mathsf{Avail},
$
a formula asserting that
anything that is \emph{made available} by the warehouse $w$ is \emph{not made unavailable} by $w$.
This principle, also well-known for its epistemic implications (anything that is \emph{known} by an agent is \emph{compatible} with their knowledge),
in relational semantics is \emph{equivalent} to a much lesser known principle, sometimes called \emph{$\mathbf{P}$-principle}.
An example of it is given by the formula
$
 \top \sqsubseteq \lnot \mathbb{D}_{w} (\mathsf{Avail} \sqcap \lnot \mathsf{Avail}),
$
which states that, for any object, it is not the case that the warehouse $w$ makes it both available \emph{and} unavailable.
Neighbourhood semantics, under which the $\mathbf{D}$- and the $\mathbf{P}$-principle are \emph{not equivalent}, allows for distinctions between modal constraints that are hence more fine-grained than in relational semantics.

\begin{figure*}
\scriptsize
  \[
      \begin{array}{@{}r@{~~}c@{~~}l@{\qquad}r@{~~}c@{~~}l@{}}
        \top & \sqsubseteq & \forall \mathsf{req}.\mathsf{Prod}
        &
		\mathsf{InvalidOrd} & \sqsubseteq &\mathbb{C}_{s} \mathsf{Cancel}
        \\
         \mathsf{Confirm} & \sqsubseteq & \lnot \mathsf{PartConf}
        &
		\mathsf{SubmitOrd} & \sqsubseteq & \mathbb{C}_{s} \forall \mathsf{req}.
(\mathbb{D}_{w} \mathsf{Avail} \sqcup \mathbb{D}_{w} \lnot \mathsf{Avail})
        \\
         \mathsf{Confirm} & \sqsubseteq & \lnot \mathsf{Reject}
        &
       \mathsf{SubmitOrd} & \sqsubseteq & \mathbb{C}_{s}\mathsf{Confirm} \sqcap \mathbb{C}_{s}\mathsf{PartConf} \sqcap \mathbb{C}_{s}\mathsf{Reject}
        \\
         \mathsf{PartConf} & \sqsubseteq & \lnot \mathsf{Reject}
        &
         \mathbb{D}_{s} \mathsf{Confirm} & \sqsubseteq & \forall \mathsf{req}.
\mathbb{D}_{w}\mathsf{Avail}
        \\
        \mathsf{Ord} & \equiv & \mathbb{D}_{c}\exists \mathsf{req}.(\mathsf{Prod} \sqcap \mathsf{InCatal} )
        &
        \mathbb{C}_{s} \mathsf{PartConf} & \sqsupseteq & \exists \mathsf{req}.
\mathbb{D}_{w}\mathsf{Avail} \ \sqcap
\exists \mathsf{req}.
\mathbb{D}_{w} \lnot \mathsf{Avail}
        \\
       \mathsf{SubmitOrd} & \equiv & \mathsf{Ord} \sqcap \mathbb{D}_{c} \mathsf{Submit}
        &
        \mathbb{C}_{s} \mathsf{Reject} & \sqsupseteq & \forall \mathsf{req}.
\mathbb{D}_{w}\lnot\mathsf{Avail}
        \\
        \mathsf{IncomplOrd} & \equiv & \mathsf{Ord} \sqcap \lnot \mathbb{D}_{c} \mathsf{Submit}
       &
       \mathbb{D}_{m} (\mathsf{ConfirmOrd} & \sqsubseteq & \lnot \mathbb{C}_{s} \mathsf{Reject})
       \\               
       \mathsf{InvalidOrd} & \equiv & \mathsf{Ord} \sqcap \lnot \mathbb{C}_{c} \mathsf{Submit} 
       &
        \mathsf{DispatchOrd} & \equiv & \mathsf{Ord} \sqcap \mathbb{D}_{s} \mathbb{D}_{w} \mathsf{Dispatch}
         \\
\mathsf{ConfirmOrd} & \equiv & \mathsf{Ord} \sqcap \mathbb{D}_{s} \mathsf{Confirm}
		&
		\mathbb{D}_{m} (\mathsf{ConfirmOrd} & \sqsubseteq & \mathsf{DispatchOrd})
      \end{array}
  \]
  \caption{Purchase choreography knowledge base.
  }
  \label{fig:purchasekb}
\end{figure*}

The full knowledge base describing the purchase choreography scenario
is reported in Figure~\ref{fig:purchasekb}.
Notably, it displays a great range of interactions between modal and description logics constructs.
The modal operators $\mathbb{D}_{i}$ and $\mathbb{C}_{i}$ are axiomatized similarly to~\protect\cite{Elg}:
$\mathbb{D}_{i}$ obeys the $\mathbf{C}$- and $\mathbf{T}$-principles;
  and both satisfy the $\mathbf{Q}$-principle, and the $\mathbf{E}$-principle:
  $C \equiv D$ entails $\mathbb{D}_{i} C \equiv \mathbb{D}_{i} D$ and $\mathbb{C}_{i} C \equiv \mathbb{C}_{i} D$.

On the left column, the first four axioms impose that the range of the request relation is a product, and that 
the classes of confirmed, partially confirmed, and rejected objects are all disjoint.
We have already discussed the fifth axiom, which is the definition of an order as a request, made by a customer, for some in-catalogue product.
The last four axioms on the left column define, respectively:
a submitted order as an order that is actually submitted by a customer;
an incomplete order as an order that is not submitted by the customer;
an invalid order as an order that cannot be submitted by the customer;
and a confirmed order as an order that is confirmed by the seller.

On the right column, the first formula states that an invalid order can be cancelled by the seller.
The second axiom asserts that, given a submitted order, the seller can check at the warehouse for availability of all of its requested products.
The third formula, as already discussed, requires that a submitted order can be confirmed, can be partially confirmed, and can be rejected by the seller.
The subsequent three axioms impose, respectively, the following constraints:
anything that is confirmed by the seller has to request only products made available by the warehouse;
anything that requests products that are made available, as well as products that are brought about to be unavailable by the warehouse, can be partially confirmed by the seller;
finally, any request of products that are all unavailable at the warehouse can be rejected by the seller.
The third formula from the bottom of the right column states that the marketplace enforces that a confirmed order cannot be rejected by the seller.
The second last axiom defines a dispatched order as an order that the seller makes the warehouse dispatch.
Finally, the last formula imposes that the marketplace brings it about that any confirmed order is also dispatched.

\section{Proofs for Section~\ref{sec:prelim}}

\PropImplicationSystem*
\begin{proof}
Point 1. Suppose that, for every $w\in \Wmc$, $\alpha,\beta\subseteq \Wmc$, we have: ($\mathbf{M}$-condition) $\alpha\in \Nmc_{i}(w)$ and $\alpha\subseteq\beta$ implies $\beta\in \Nmc_{i}(w)$; and ($\mathbf{Q}$-condition) $\Wmc \not \in \Nmc_{i}(w)$. This means that $\alpha \not \in \Nmc_{i}(w)$, for every $\alpha \subseteq \Wmc$, i.e., $\Nmc_{i}(w) = \emptyset$.
Thus, every condition, except for the $\mathbf{N}$-condition, is satisfied by $\Nmc_{i}$.

Point 2.
%
Suppose~$(i)$ that, for every $w\in \Wmc$, $\alpha,\beta\subseteq \Wmc$, we have: ($\mathbf{M}$-condition) $\alpha\in \Nmc_{i}(w)$ and $\alpha\subseteq\beta$ implies $\beta\in \Nmc_{i}(w)$; and ($\mathbf{D}$-condition) $\alpha \in \Nmc_{i}(w)$ implies $\Wmc \setminus \alpha \not \in \Nmc_{i}(w)$.
Towards a contradiction, suppose that $\emptyset \in \Nmc_{i}(w)$. Then, we have $\Wmc = \Wmc \setminus \emptyset \not \in \Nmc_{i}(w)$ as well. By contraposition, this implies in particular that $\emptyset \not \in \Nmc_{i}(w)$, a contradiction. Thus, ($\mathbf{P}$-condition) $\emptyset \not \in \Nmc_{i}(w)$

Moreover, suppose~$(ii)$ that, for every $\alpha\subseteq \Wmc$, we have: ($\mathbf{N}$-condition) $\Wmc \in \Nmc_{i}(w)$, i.e., $\Wmc \setminus \emptyset \in \Nmc_{i}(w)$; and ($\mathbf{D}$-condition) $\alpha \in \Nmc_{i}(w)$ implies $\Wmc \setminus \alpha \not \in \Nmc_{i}(w)$. By contraposition, we obtain that ($\mathbf{P}$-condition) $\emptyset \not \in  \Nmc_{i}(w)$.

Finally, suppose~$(iii)$ that, for every $w\in \Wmc$, $\alpha,\beta\subseteq \Wmc$, we have: ($\mathbf{T}$-condition) $\alpha\in \Nmc_{i}(w)$ implies $w \in \alpha$. Then, we have ($\mathbf{P}$-condition) $\emptyset \not \in \Nmc_{i}(w)$, for otherwise we would get a contradiction.

Point 3.
Suppose~$(i)$ that, for every $w\in \Wmc$, $\alpha,\beta\subseteq \Wmc$, we have: ($\mathbf{C}$-condition)
$\alpha\in \Nmc_{i}(w)$ and $\beta\in \Nmc_{i}(w)$ implies $\alpha\cap\beta\in \Nmc_{i}(w)$; and ($\mathbf{P}$-condition) $\emptyset \not \in  \Nmc_{i}(w)$. Given $\alpha \in \Nmc_{i}(w)$, suppose towards a contradiction that $\Wmc \setminus \alpha \in \Nmc_{i}(w)$. From this, we obtain that $\emptyset = \alpha \cap (\Wmc \setminus \alpha) \in \Nmc_{i}(w)$, a contradiction. Thus, we have that ($\mathbf{D}$-condition) $\alpha \in \Nmc_{i}(w)$ implies $\Wmc \setminus \alpha \not \in \Nmc_{i}(w)$.

Moreover, suppose~$(ii)$ that, for every $w\in \Wmc$, $\alpha,\beta\subseteq \Wmc$, we have: ($\mathbf{T}$-condition) $\alpha\in \Nmc_{i}(w)$ implies $w \in \alpha$. Consider $\alpha\in \Nmc_{i}(w)$ and suppose, towards a contradiction, that also $\Wmc \setminus \alpha \in \Nmc_{i}(w)$. We obtain that $w \in \alpha$ 
and $w \in \Wmc \setminus \alpha$ as well, a contradiction. Hence, we have that ($\mathbf{D}$-condition) $\alpha \in \Nmc_{i}(w)$ implies $\Wmc \setminus \alpha \not \in \Nmc_{i}(w)$.

Point 4. Straightforward, because otherwise we immediately have a contradiction.
\end{proof}


\PropCorresp*
\begin{proof}
Here we present a proof  only for $\MLALC{n}$ concept inclusions 
 and only for the basic principles $L \in \{ \mathbf{E, M, C, N, P, Q, D, T} \}$.
For $\MLALC{n}$ formulas, the proof is similar to the case of propositional non-normal modal logics (see e.g.~\cite{Pac}).
More complex principles  (e.g. $\mathbf{EMCN}$) 
can be obtained by suitably combining the basic principles. 
%
%
%
Let $\Fmc = (\Wmc, \{ \Nmc_{i} \}_{i \in J})$ be a neighbourhood frame and let 
$L$ be as above.
\begin{description}
	\item[\textnormal{$\mathit{L = \mathbf{E}}$.}]
	If $C \equiv D$ is valid on $\Fmc$, then for all
	$\Mmc = (\Fmc, \Imc)$
	based on $\Fmc$, and all $w$ in $\Mmc$, we have 
$C^{\Int_w}  = D^{\Int_w}$.
Thus, for all $d \in \Delta_{w}$, $\llbracket C \rrbracket^{\Mmc}_{d} = \llbracket D \rrbracket^{\Mmc}_{d}$.
So for all $v\in\W$, $\llbracket C \rrbracket^{\Mmc}_{d} \in \Nmc_{i}(v)$ iff $\llbracket D \rrbracket^{\Mmc}_{d} \in \Nmc_{i}(v)$,
which implies $d\in (\B_{i} C)^{\Int_v}$ iff $d\in(\B_{i} D)^{\Int_v}$,
that is $(\B_{i} C)^{\Int_v} = (\B_{i} D)^{\Int_v}$.
Then, $\B_{i} C \equiv \B_{i} D$ is valid on $\Fmc$, for all $i\in J$.
	\item[\textnormal{$\mathit{L = \mathbf{M}}$.}]
From right to left,
assume that $\Fmc$ is supplemented and $C \sqsubseteq D$ is valid on $\Fmc$.
Then, for all
$\Mmc = (\Fmc, \Imc)$
based on $\Fmc$, and all $w$ in $\Mmc$, we have 
$C^{\Int_w} \subseteq D^{\Int_w}$.
Thus, for all $d\in\Delta_{w}$, $\llbracket C \rrbracket^{\Mmc}_{d} \subseteq \llbracket D \rrbracket^{\Mmc}_{d}$.
By supplementation we have that, for all $v\in\W$, $\llbracket C \rrbracket^{\Mmc}_{d} \in \Nmc_{i}(v)$ implies $\llbracket D \rrbracket^{\Mmc}_{d} \in \Nmc_{i}(v)$.
So $d\in (\B_{i} C)^{\Int_v}$ implies $d\in(\B_{i} D)^{\Int_v}$,
that is $(\B_{i} C)^{\Int_v} \subseteq (\B_{i} D)^{\Int_v}$.
Then $\B_{i} C \sqsubseteq \B_{i} D$ is valid on $\Fmc$.
For the left-to-right direction,
assume that $\Fmc$ is not supplemented. Then 
there are $w\in\Wmc$, $\alpha,\beta\subseteq\Wmc$ such that $\alpha\subseteq\beta$, $\alpha\in\Nmc_{i}(w)$ and $\beta\notin\Nmc_{i}(w)$.
We define over $\Fmc$ the model
$\Mmc= ( \Fmc, \Int )$,
where
$\Delta_{w} = \{d\}$, for every $w \in \Wmc$,
and the interpretation of two concept names $A, B \in \NC$ is defined as follows:
$d\in A^{\Int_v}$ iff $v\in\alpha$, and
$d\in B^{\Int_v}$ iff $v\in\beta$
(and defined arbitrarily on all other symbols in $( \NC \cup \NR ) \setminus \{ A, B \}$).
As a consequence, we have that $\llbracket A \rrbracket_d^{\Mmc} = \alpha$ and $\llbracket B \rrbracket_d^{\Mmc} = \beta$,
which implies $\llbracket A \rrbracket_d^{\Mmc} = \llbracket A \rrbracket_d^{\Mmc} \cap \llbracket B \rrbracket_d^{\Mmc} = \llbracket A \sqcap B \rrbracket_d^{\Mmc}$.
Thus $\llbracket A \sqcap B \rrbracket_d^{\Mmc}\in\Nmc_{i}(w)$ and $\llbracket B \rrbracket_d^{\Mmc}\notin \Nmc_{i}(w)$. 
By definition we have $d\in(\B_{i}(A\sqcap B))^{\Int_w}$ and $d\notin(\B_{i} B)^{\Int_w}$.
	\item[\textnormal{$\mathit{L = \mathbf{C}}$.}]
From right-to-left,
assume that $\Fmc$ is closed under intersection. 
Moreover, let
$\Mmc = (\Fmc, \Imc)$
be a model based on $\Fmc$, with $w$ world of $\Mmc$, and $d \in \Delta$
such that $d\in(\B_{i} C \sqcap \B_{i} D)^{\Int_w}$.
Thus $d\in(\B_{i} C)^{\Int_w}$ and $d\in(\B_{i} D)^{\Int_w}$,
that is $[C]_d^{\Mmc}, [D]_d^{\Mmc} \in\Nmc_{i}(w)$.
By closure under intersection, $[C]_d^{\Mmc} \cap [D]_d^{\Mmc} =  [C\sqcap D]_d^{\Mmc}\in\Nmc_{i}(w)$.
Then $d\in(\B_{i} ( C \sqcap D))^{\Int_w}$.
For the left-to-right direction,
assume that $\Fmc$ is not closed under intersection.
Then, there are $w\in\Wmc$, $\alpha,\beta\subseteq\Wmc$ such that $\alpha, \beta\in\Nmc_{i}(w)$ and $\alpha\cap\beta\notin\Nmc_{i}(w)$.
We define over $\Fmc$ the model
$\Mmc = ( \Fmc, \Int )$,
where
$\Delta_{w} = \{d\}$, for all $w \in \Wmc$,
and the interpretation of two concept names $A, B \in \NC$ is defined as follows:
$d\in A^{\Int_v}$ iff $v\in\alpha$, and
$d\in B^{\Int_v}$ iff $v\in\beta$
(and defined arbitrarily on all other symbols in $( \NC \cup \NR ) \setminus \{ A, B \}$).
We have $\llbracket A \rrbracket_d^{\Mmc} = \alpha$ and $\llbracket B \rrbracket_d^{\Mmc} = \beta$, which implies
$d\in(\B_{i} A)^{\Int_w}$ and $d\in(\B_{i} B)^{\Int_w}$.
Moreover, $\llbracket A \sqcap B \rrbracket_d^{\Mmc} = \llbracket A \rrbracket_d^{\Mmc} \cap \llbracket B \rrbracket_d^{\Mmc} = \alpha\cap\beta \notin\Nmc_{i}(w)$.
Thus $d\notin(\B_{i}(A\sqcap B))^{\Int_w}$.
	\item[\textnormal{$\mathit{L = \mathbf{N}}$.}]
From right-to-left,
assume that $\Fmc$ contains the unit and $\top \sqsubseteq C$ is valid on $\Fmc$.
Then for all
$\Mmc = (\Fmc, \Imc)$
based on $\Fmc$, and all $w$ in $\Mmc$, we have 
$C^{\Int_w} = \Delta_{w}$. 
As a consequence, for all
$d\in\Delta_{w}$,
$\llbracket C \rrbracket^{\Mmc}_{d} = \W$.
By the property of containing the unit we have that, for all $v\in\W$, $\llbracket C \rrbracket^{\Mmc}_{d} \in \Nmc_{i}(v)$.
So $d\in (\B_{i} C)^{\Int_v}$ for all
$d\in\Delta_{v}$, 
that is, $\top \sqs \B_{i} C$ is valid on $\Fmc$.
For the left-to-right direction,
assume that $\Fmc$ does not contain the unit, i.e.,
there is $w\in\Wmc$ such that $\W\notin\Nmc_{i}(w)$.
Then, for all models
$\Mmc= ( \Fmc, \Int )$
based on $\Fmc$,
all $w \in \Wmc$,
and all
$d\in\Delta_{w}$,
we have $d\in\top^{\Int_w}$.
Moreover, 
since $d\in (\B_{i}\top)^{\Int_w}$ iff $\llbracket \top \rrbracket_d^{\Mmc}\in\Nmc_{i}(w)$ iff $\W\in\Nmc_{i}(w)$,
we also have $d\notin (\B_{i}\top)^{\Int_w}$.
	\item[\textnormal{$\mathit{L = \mathbf{P}}$.}]
          From right-to-left, assume that $\Fmc = ( \Wmc, \{\Nmc_i \}_{i \in J})$ satisfies the $\mathbf{P}$-condition. I.e., for all $w \in \Wmc$ and $i \in J$, $\emptyset \not\in \Nmc_i(w)$.
          Then, for all models
          $\Mmc= ( \Fmc, \Int )$
          based on $\Fmc$,
          all $w \in \Wmc$,
          and all
          $d\in\Delta_{w}$,
          we have $\emptyset = \llbracket \bot \rrbracket^\Mmc_d \not\in \Nmc_i(w)$. So $d \not\in (\Box_i\bot)^{\Int_w}$, or equivalently $d \in (\lnot \Box_i\bot)^{\Int_w}$. Also,
          $d \in \Delta_{w} = \top^{\Int_w}$.
          So $\top^{\Int_w} \subseteq (\lnot \Box_i\bot)^{\Int_w}$. Then $\Mmc, w \models \top \sqsubseteq \lnot \Box_i \bot$. Hence $\top \sqsubseteq \lnot \Box_i \bot$ is valid on $\Fmc$.
          For the
          left-to-right direction,
          assume that $\Fmc$ does not satisfy the $P$-condition for some $i \in J$. This means that there is $w \in \Wmc$ such that $\emptyset \in \Nmc_i(w)$. So there exists a model
             $\Mmc= ( \Fmc, \Int )$
          based on $\Fmc$,
          a $w \in \Wmc$, and a
          $d\in\Delta_{w}$,
          such that $\emptyset = \llbracket \bot \rrbracket^\Mmc_d \in \Nmc_i(w)$. So $d \in (\Box_i \bot)^{\Int_w}$, or equivalently $d \not\in (\lnot\Box_i \bot)^{\Int_w}$. But $d \in \Delta_{w} = \top^{\Int_w}$. So $\top^{\Int_w} \not\subseteq (\lnot\Box_i \bot)^{\Int_w}$. Then $\Mmc, w \not\models \top \sqsubseteq \lnot\Box_i\bot$. Hence $\top \sqsubseteq \lnot \Box_i \bot$ is not valid on $\Fmc$.

	\item[\textnormal{$\mathit{L = \mathbf{Q}}$.}]
          From right to left, assume that $\Fmc = ( \Wmc, \{\Nmc_i \}_{i \in J})$ satisfies the $\mathbf{Q}$-condition. I.e., for all $w \in \Wmc$ and $i \in J$, $\Wmc \not\in \Nmc_i(w)$.
          Then, for all models
          $\Mmc= ( \Fmc, \Int )$
          based on $\Fmc$,
          all $w \in \Wmc$,
          and all
          $d\in\Delta_{w}$,
          we have $\Wmc = \llbracket \top \rrbracket^\Mmc_d \not\in \Nmc_i(w)$. So $d \not\in (\Box_i\top)^{\Int_w}$, or equivalently $d \in (\lnot \Box_i\top)^{\Int_w}$. Also $d \in \Delta = \top^{\Int_w}$. So $\top^{\Int_w} \subseteq (\lnot \Box_i\top)^{\Int_w}$. Then $\Mmc, w \models \top \sqsubseteq \lnot \Box_i \top$. Hence $\top \sqsubseteq \lnot \Box_i \top$ is valid on $\Fmc$.
          For the other direction, assume that $\Fmc$ does not satisfy the $\mathbf{Q}$-condition for some $i \in J$. This means that there is $w \in \Wmc$ such that $\Wmc \in \Nmc_i(w)$.
          So there exists a model
          $\Mmc= ( \Fmc, \Int )$
          based on $\Fmc$,
          a $w \in \Wmc$,
          and a
               $d\in\Delta_{w}$,
          such that $\Wmc = \llbracket \top \rrbracket^\Mmc_d \in \Nmc_i(w)$. So $d \in (\Box_i \top)^{\Int_w}$, or equivalently $d \not\in (\lnot\Box_i \top)^{\Int_w}$. But $d \in \Delta = \top^{\Int_w}$. So $\top^{\Int_w} \not\subseteq (\lnot\Box_i \top)^{\Int_w}$. Then $\Mmc, w \not\models \top \sqsubseteq \lnot\Box_i\top$. Hence $\top \sqsubseteq \lnot \Box_i \top$ is not valid on $\Fmc$.
          
	\item[\textnormal{$\mathit{L = \mathbf{D}}$.}]
          From right-to-left, assume that $\Fmc  = ( \Wmc, \{\Nmc_i \}_{i \in J})$ satisfies the $\mathbf{D}$-condition. Moreover, let
          $\Mmc= ( \Fmc, \Int )$
          based on $\Fmc$, with a world $w \in \Wmc$, and
          $d \in \Delta_{w}$,
          a concept $C$, and $i \in J$, all arbitrarily chosen. Suppose that $d \in (\Box_i C)^{\Int_w}$. It means that $[C]^\Mmc_d \in \Nmc_i(w)$. By the $\mathbf{D}$-condition, $\Wmc \setminus [C]^\Mmc_d \not\in \Nmc_i(w)$. Equivalently, $[\lnot C]^\Mmc_d \not\in \Nmc_i(w)$. This means that $d \not\in (\Box_i\lnot C)^{\Int_w}$, or equivalently, that $d \in (\lnot\Box_i\lnot C)^{\Int_w}$. So $(\Box_iC)^{\Int_w} \subseteq (\lnot \Box_i\lnot C)^{\Int_w}$. Then $\Mmc,w \models \Box_i C \sqsubseteq \Diamond_i C$. Hence, $\Box_i C \sqsubseteq \Diamond_i C$ is valid on $\Fmc$.
          For the left-to-right direction, assume that $\Fmc$ does not satisfy the $\mathbf{D}$-condition for $i \in J$. So there is a $w \in \Wmc$, such that, for some $\alpha \subseteq \Wmc$, we have $\alpha \in \Nmc_i(w)$ and $\Wmc \setminus \alpha \in \Nmc_i(w)$.
          We define
            $\Mmc= (\Fmc, \Int )$
          based on $\Fmc$,
          where,
          for any $v \in \Wmc$,
          $\Delta_{v} = \{d\}$,
          and $A^{\Int_v} = \{d\}$ iff $v \in \alpha$ (and defined arbitrarily on all other symbols in $( \NC \cup \NR ) \setminus \{A\}$). We have $\llbracket A \rrbracket^\Mmc_d = \{v \in \Wmc \mid d \in A^{\Int_w}\} = \alpha$, and $\llbracket \lnot A \rrbracket^\Mmc_d = \Wmc \setminus \alpha$. So $d \in (\Box_i A)^{\Int_w}$ and $d \in (\Box_i \lnot A)^{\Int_w}$, and then $d \not \in (\lnot \Box_i \lnot A)^{\Int_w}$. So $(\Box_iA)^{\Int_w} \not \subseteq (\lnot \Box_i \lnot A)^{\Int_w}$. This means that, $\Mmc, w \not\models \Box_i A \sqsubseteq \Diamond_i A$. Hence, $\Box_i C \sqsubseteq \Diamond_i C$ is not valid on $\Fmc$.

	\item[\textnormal{$\mathit{L = \mathbf{T}}$.}]
          From right-to-left, assume that $\Fmc = ( \Wmc, \{\Nmc_i \}_{i \in J})$ satisfies the $\mathbf{T}$-condition. Moreover, let
          $\Mmc= ( \Fmc, \Int )$
          based on $\Fmc$, with a world $w \in \Wmc$, and
          $d \in \Delta_{w}$,
          a concept $C$, and $i \in J$, all arbitrarily chosen. Suppose that $d \in (\Box_i C)^{\Int_w}$. It means that $\llbracket C \rrbracket^\Mmc_d \in \Nmc_i(w)$. By the $\mathbf{T}$-condition, $w \in \llbracket C \rrbracket^\Mmc_d$. So $d \in C^{\Int_w}$. So $(\Box_i C)^{\Int_w} \subseteq C^{\Int_w}$. Then $\Mmc, w \models \Box_i C \sqsubseteq C$. Hence, $\Box_i C \sqsubseteq C$ is valid on $\Fmc$.
          For the left-to-right direction, assume that $\Fmc$ does not satisfy the $\mathbf{T}$-condition for $i \in J$. So there is a $w \in \Wmc$, such that for some $\alpha \subseteq \Wmc$ we have $\alpha \in \Nmc_i(w)$ and $w \not\in \alpha$. 
          We define
          $\Mmc= (\Fmc, \Int )$
          based on $\Fmc$, where
          $\Delta_{v} = \{d\}$, for any $v \in \Wmc$,
          and $A^{\Int_v} = \{d\}$ iff $v \in \alpha$ (and defined arbitrarily on all other symbols in $( \NC \cup \NR ) \setminus \{A\}$). We have $\llbracket A \rrbracket^\Mmc_d = \{v \in \Wmc \mid d \in A^{\Int_w}\} = \alpha$. So $d \in (\Box_i A)^{\Int_w}$. Since $w \not \in \alpha$, $w \in \Wmc \setminus \alpha$. That is, $w \in \Wmc \setminus \llbracket A \rrbracket^\Mmc_d$, or equivalently $w \in \llbracket \lnot A \rrbracket^\Mmc_d$. Then, $d \in (\lnot A)^{\Int_w}$, or equivalently $d \not \in (A)^{\Int_w}$. So, 
$(\Box_i A)^{\Int_w} \not\subseteq A^{\Int_w}$,
          meaning that  $\Mmc, w \not\models \Box_i A \sqsubseteq A$. Hence, $\Box_i C \sqsubseteq C$ is not valid on $\Fmc$, as required.
\qedhere
\end{description}
\end{proof}

\PropValid*
\begin{proof}
\emph{Point~1.}
{{The $(\Rightarrow)$ direction follows from the proof of Proposition~\ref{prop:corresp}.
To see that the $(\Leftarrow)$ direction does not hold in general,
we provide the following counterexample
showing that the $\mathbf{T}$-principle holds in a model that does not satisfy the $\mathbf{T}$-condition.
Consider 
$\Mmc = ( \Wmc, \{\Nmc_i \}_{i \in J}, \Int)$,
where 
\begin{itemize}
\item $\Wmc = \{w, v\}$;
\item
$\Nmc_{i}(w) = \{\{v\},\Wmc\}$ and
$\Nmc_{i}(v) = \{\{w\},\Wmc\}$,
 for $i \in J$;
\item $\Imc_{w} = \Imc_{v}$, with $\Delta_{w} = \Delta_{v} = \{d\}$.
\end{itemize}
$\Mmc$ does not satisfy the $\mathbf{T}$-condition,
since $\{v\} \in \Nmc_{i}(w)$ but $w \notin \{v\}$.
We show that the $\mathbf{T}$-principle holds in $\Mmc$.

\begin{claim}
For all concepts $C$, 
$\llbracket C \rrbracket^{\Mmc}_{d} = \emptyset$ or
$\llbracket C \rrbracket^{\Mmc}_{d} = \Wmc$. 
\end{claim}
\begin{proof}[Proof of Claim]
By induction on the construction of $C$.

For the base case $C = A \in \NC$,
it follows from the definition that either
$d \in A^{\Imc_{w}}$ and $d \in A^{\Imc_{v}}$,
hence 
$\llbracket A \rrbracket^{\Mmc}_{d} = \Wmc$, 
or
$d \notin A^{\Imc_{w}}$ and $d \notin A^{\Imc_{v}}$,
hence 
$\llbracket A \rrbracket^{\Mmc}_{d} = \emptyset$.

We now show the inductive cases.
For $C = \lnot D, D \sqcap E$, the proof is immediate by applying the induction hypothesis.

For $C = \exists \role.D$, the proof follows by the application of the induction hypothesis
and the fact that $r^{\Imc_{w}} = r^{\Imc_{v}}$.

For $C = \B_{i} D$, by induction hypothesis
$\llbracket D \rrbracket^{\Mmc}_{d} = \emptyset$ or
$\llbracket D \rrbracket^{\Mmc}_{d} = \Wmc$. 
In the first case,
$\llbracket D \rrbracket^{\Mmc}_{d} \notin \Nmc_{i}(w)$
and
$\llbracket D \rrbracket^{\Mmc}_{d} \notin \Nmc_{i}(v)$,
hence $w \notin \llbracket \B_{i} D \rrbracket^{\Mmc}_{d}$
and $v \notin \llbracket \B_{i} D \rrbracket^{\Mmc}_{d}$,
thus $ \llbracket \B_{i} D \rrbracket^{\Mmc}_{d} = \emptyset$.
In the second case,
$\llbracket D \rrbracket^{\Mmc}_{d} \in \Nmc_{i}(w)$
and
$\llbracket D \rrbracket^{\Mmc}_{d} \in \Nmc_{i}(v)$,
hence $w \in \llbracket \B_{i} D \rrbracket^{\Mmc}_{d}$
and $v \in \llbracket \B_{i} D \rrbracket^{\Mmc}_{d}$,
thus $ \llbracket \B_{i} D \rrbracket^{\Mmc}_{d} = \Wmc$. 
\end{proof}

Now, given a concept $C$, suppose that $d \in (\B_{i} C)^{\Imc_{w}}$,
that is, $\llbracket C \rrbracket^{\Mmc}_{d} \in \Nmc_{i}(w)$.
By the claim, $\llbracket C \rrbracket^{\Mmc}_{d} = \emptyset$
or $\llbracket C \rrbracket^{\Mmc}_{d} = \Wmc$.
Since $\emptyset \notin \Nmc_{i}(w)$,
we have that $\llbracket C \rrbracket^{\Mmc}_{d} = \Wmc$,
and thus $w \in \llbracket C \rrbracket^{\Mmc}_{d}$.
This means that $d \in C^{\Imc_{w}}$.
By the same argument, we can show that 
$d \in (\B_{i} C)^{\Imc_{v}}$ implies
$d \in C^{\Imc_{v}}$.
Therefore $\Mmc \models \B_{i} C \sqsubseteq C$.
Similarly we can prove that $\Mmc \models \Box_{i} \varphi \to \varphi$,
for any formula $\varphi$.
}}

\emph{Point~2.}
Let $\Fmf = (W, \{ R_{i} \}_{i \in J})$ be a relational frame, and let $\Mmf = (F, \Delta, I)$ be a relational model based on $\Fmf$.


($\mathbf{E}$-principle) Follows directly from the ($\mathbf{M}$-principle) case below.

($\mathbf{M}$-principle) Assume $C \sqsubseteq D$ valid in $\Mmf$. Then $C^{I_w} \subseteq D^{I_w}$, for all $w \in W$. Now, suppose that $d\in (\B_{i} C)^{I_{w}}$, for $d \in \Delta$ and $w \in W$.
For all $v \in W$, $w \relations_{i} v$ implies $d\in C^{I_{v}}$,
hence $d\in D^{I_{v}}$.
Therefore, $d\in (\B_{i} D)^{I_{w}}$.
%

($\mathbf{C}$-principle) Assume $d\in(\B_{i} C \sqcap \B_{i} D)^{I_{w}}$, that is, $d\in(\B_{i} C)^{I_{w}}$ and $d\in(\B_{i} D)^{I_{w}}$.
Then, for all $v \in W$, $w\relations_{i} v$ implies $d\in C^{I_{v}}$ and $d\in D^{I_{v}}$,
that is $d\in (C\sqcap D)^{I_{v}}$.
Therefore, $d\in(\B_{i} ( C \sqcap D))^{I_{w}}$.

($\mathbf{N}$-principle) Assume $\top \sqsubseteq C$ valid in $\Mmf$.
Then, for all $w \in W$,  $C^{I_w} =\Delta$.
Thus, for all $d\in\Delta$ and all $v \in W$, $w \relations_{i} v$ implies $d\in C^{I_{v}}$.
In conclusion, for all $d \in \Delta$, $d\in (\B_{i} C)^{I_{w}}$. 

In conclusion, the
$\mathbf{E}$-,
$\mathbf{M}$-, $\mathbf{C}$-, and $\mathbf{N}$-principle
hold in $M$, and hence in $F$.

{{
($\mathbf{D}$-principle and $\mathbf{P}$-principle are equivalent)
It is easy to see that both the $\mathbf{D}$- and the $\mathbf{P}$-principle
hold if and only if 
$R_{i}$ is serial, for all $i \in J$
(that is, for all $w \in W$, there is $v\in W$ such that $w R_i v$).

($\mathbf{Q}$-principle does not hold)
The $\mathbf{Q}$-principle is incompatible with the $\mathbf{N}$-principle, whcih holds in relational models.
Indeed, if both principles hold in the relational model $\Mmf$,
then $\Mmf\models\top\sqsubseteq \Box_i\top\sqcap\lnot\Box_i\top$,
against the fact that $\ALC$ domains are non-empty.
}}
\end{proof}

\section{Proofs for Section~\ref{sec:tableaux}}

Here is an $\LnALC$ tableau algorithm example application.

\begin{example}
Consider the formula
$$\varphi = \lnot(\mathbb{D}_c\exists\mathsf{req}.(\mathsf{Prod} \sqcap \mathsf{InCatal}) \sqsubseteq
\mathbb{D}_c(\mathsf{Conf} \sqcup\lnot\mathsf{Conf})),$$
related to the discussion in Section~\ref{sec:model}.
We recall that the formula is unsatisfiable in models 
validating the $\mathbf{M}$-condition,
and it is satisfiable otherwise.
Here we show that the algorithm provides different answers depending whether $\mathbf{M}\in\Lvar$.
First, we rewrite $\p$ in NNF, using $\widehat{\mathbb{D}}_c$ as the dual operator of $\mathbb{D}_c$,
thus obtaining
$$\lnot(\top \sqsubseteq \widehat{\mathbb{D}}_c\forall \mathsf{req}.
(\lnot \mathsf{Prod} \sqcup \lnot \mathsf{InCatal}) \sqcup \mathbb{D}_c(\mathsf{Conf} \sqcup \lnot \mathsf{Conf})).$$
We then consider the following applications of the tableau algorithm.
In the first case we assume $\mathbf{M}\in\Lvar$:

%
%
%
%
%

\medskip
\noindent
$0 : \p$

\noindent
$0 : \mathbb{D}_c \exists\mathsf{req}.(\mathsf{Prod} \sqcap \mathsf{InCatal}) \sqcap
\widehat{\mathbb{D}}_c(\lnot\mathsf{Conf} \sqcap\mathsf{Conf})(v)$ 
\hfill ($\mathsf{R}_{\not\sqsubseteq}$)

\noindent
$1 : \exists\mathsf{req}.(\mathsf{Prod} \sqcap \mathsf{InCatal})(v)$
\hfill ($\mathsf{R}_{\mathit{L}}$)

\noindent
$1 : \lnot\mathsf{Conf} \sqcap\mathsf{Conf}(v)$
\hfill ($\mathsf{R}_{\mathit{L}}$)

\noindent
$1 : \lnot\mathsf{Conf}(v)$
\hfill ($\mathsf{R}_{\sqcap}$)

\noindent
$1 : \mathsf{Conf}(v)$
\hfill ($\mathsf{R}_{\sqcap}$)

\medskip

\noindent
The completion set constructed by the application of the $\LnALC$ tableau algorithm contains a clash,
hence the algorithm returns $\mathsf{unsatisfiable}$ on input $\p$. 
Now, assume $\mathbf{M}\notin\Lvar$:

\medskip

\noindent
$0 : \p$

\noindent
$0 : \mathbb{D}_c \exists\mathsf{req}.(\mathsf{Prod} \sqcap \mathsf{InCatal}) \sqcap
\widehat{\mathbb{D}}_c(\lnot\mathsf{Conf} \sqcap\mathsf{Conf})(v)$ 
\hfill ($\mathsf{R}_{\not\sqsubseteq}$)

\noindent
$1 : \forall \mathsf{req}.(\lnot \mathsf{Prod} \sqcup \lnot \mathsf{InCatal})(v)$
\hfill ($\mathsf{R}_{\mathit{L}}$)

\noindent
$1 : \mathsf{Conf} \sqcup \lnot \mathsf{Conf}(v)$
\hfill ($\mathsf{R}_{\mathit{L}}$)

\noindent
$1 : \mathsf{Conf}(v)$
\hfill ($\mathsf{R}_{\sqcup}$)

\medskip

\noindent
The completion set 
is clash-free and $\LnALC$-complete,
hence the algorithm returns $\mathsf{satisfiable}$ on input $\p$.
Note that the
latter
applications of $\mathsf{R}_{\mathit{L}}$
are only possible if $\mathbf{M}\notin\Lvar$.

\end{example}

In this appendix we prove termination, soundness, and completeness of the $\LnALC$ tableau algorithm.
We define the \emph{weight} $|C|$ of a concept $C$ in NNF as follows: $|A| = |\lnot A| = 0$; $|\exists r.D| = |\forall r.D| = |\Box_{i}D| = |\Diamond_{i}D| = |D| + 1$; $|D \sqcap E| = |D \sqcup E| = |D| + |E| + 1$. The \emph{weight $|\p|$} of a formula $\p$ in NNF is defined as:
$|C(a)| = |r(a,b)| = |\lnot r(a,b)| =
| (C \sqsubseteq D) | = | \lnot (C \sqsubseteq D) | = 0$; 
$|\Box_{i} \psi| = |\Diamond_i\psi| = | \psi | + 1$;
$| \psi \land \chi | = | \psi \lor \chi | = | \psi | + | \chi | + 1$. 
Observe that, for a concept or formula $\gamma$, we have that 
$| \gamma | = | \dnot \gamma |$.


\begin{restatable}[Termination]{theorem}{Termination}
	\label{thm:termination}
		The $\LnALC$ tableau algorithm for $\p$ terminates after at most $2^{p(|\fg(\p)|)}$ steps, where $p$ is a polynomial function.
\end{restatable}
\begin{proof}
We first require the following claims.
\begin{claim}
\label{cla:termlocal}
Let $\T$ be a completion set obtained by applying the $\LnALC$ tableau algorithm for $\p$.
For each $n \in \mathsf{L}_{\T}$,
the number of 
$n$-labelled constraints
for $\p$
in $\T$ does not exceed $2^{q(|\fg(\p)|)}$, where $q$ is a polynomial function. 
\end{claim}
\begin{proof}[Proof of Claim]
We remark that, for each $S_n\subseteq\T$, the $\LnALC$ tableaux algorithm behaves exactly like a standard (non-modal) $\ALC$ tableaux algorithm (cf. e.g.~\cite[Theorem 15.4]{GabEtAl03}), 
{{except possibly for the additional rule $\mathsf{R}_{\mathit{L}\mathbf{T}}$ which introduces at most
$| \fg(\p) |$ $n$-labelled contraints}}.
\end{proof}

\begin{claim}
\label{cla:termglobal}
{{Let $\T$ be a completion set obtained by applying the $\LnALC$ tableau algorithm for $\p$.
Then $|\mathsf{L}_{\T}| \leq r( |\fg(\p)| )$
if $\mathbf{C} \notin \mathit{L}$, and 
$|\mathsf{L}_{\T}| \leq 2^{r'(|\fg(\p)|)}$
if $\mathbf{C} \in \mathit{L}$,
for some for some polynomial functions $r$ and $r'$.
}}
\end{claim}
\begin{proof}[Proof of Claim]
{{
Labels $n$ are generated in $\T$ by means
of the application of the rules $\mathsf{R}_{\mathit{L}}$,
$\mathsf{R}_{\mathit{L}\mathbf{N}}$,
$\mathsf{R}_{\mathit{L}\mathbf{P}}$,
$\mathsf{R}_{\mathit{L}\mathbf{Q}}$,
$\mathsf{R}_{\mathit{L}\mathbf{D}}$.
If $\mathbf{C}\notin\mathit{L}$,
these rules are applied to either one or two $n$-labelled contraints,
while if $\mathbf{C}\in\mathit{L}$,
they are applied to $k$, $k+1$ or $k + h$ $n$-labelled contraints.
By the application conditions of the rules,
each such combination
of constraints
generates at most one label $m$.
Therefore, the number of
labels that can be generated in $\T$ is bounded by the number of possible
such
combinations,
which is at most $2 \cdot |\fg(\p)|^2 + 3 \cdot |\fg(\p)|$,  if $\mathbf{C}\not\in\mathit{L}$,
and at most 
$2^{|\fg(\p)|} \cdot |\fg(\p)| +
|\fg(\p)| +
2^{|\fg(\p)| + 1} + 
2^{2|\fg(\p)|}$, 
if $\mathbf{C}\in\mathit{L}$.
}}
\end{proof}

{{
The theorem is then a consequence of the following observations.
Given a completion set $\T$ constructed by the $\LnALC$ tableau algorithm for $\p$,
we have by Claim~\ref{cla:termglobal} that
the number of applications of 
the rules
$\mathsf{R}_{\mathit{L}}$, $\mathsf{R}_{\mathit{L}\mathbf{N}}$, $\mathsf{R}_{\mathit{L}\mathbf{P}}$, $\mathsf{R}_{\mathit{L}\mathbf{Q}}$, and $\mathsf{R}_{\mathit{L}\mathbf{D}}$ 
is bounded by $| \mathsf{L}_{\T} |$, which is at most 
$r( |\fg(\p)| )$,
for
$\mathit{L}$ such that $\mathbf{C}\notin\Lvar$,
and at most
$2^{r'(|\fg(\p)|)}$, for $\Lvar$ such that
 $\mathbf{C} \in \mathit{L}$,
where $r$ and $r'$ are polynomial functions.
Moreover, for a given label $n$, 
the number of possible applications of the rules
$\mathsf{R}_{\land}$, $\mathsf{R}_{\lor}$ and $\mathsf{R}_{\mathit{L}\mathbf{T}}$ to constraints of the form $n: \psi$
is linearly bounded by $\fg(\p)$,
hence there are at most $| \mathsf{L}_{\T} | \cdot q'(| \fg(\p) |)$ such rule applications, where $q'$ is a polynominal function.
Finally, by Claim~\ref{cla:termlocal}, 
for each label $n$, the number of applications of the rules
$\mathsf{R}_{\sqcap}, \mathsf{R}_{\sqcup}, \mathsf{R}_{\forall}, \mathsf{R}_{\exists}, \mathsf{R}_{\sqsubseteq}, \mathsf{R}_{\not\sqsubseteq}$  and $\mathsf{R}_{\mathit{L}\mathbf{T}}$ 
 to constraints of the form $n: C(x)$ or $n: r(x, y)$ is bounded by $2^{q(|\fg(\p)|)}$, where $q$ is a polynomial function,
hence there are at most $| \mathsf{L}_{\T} | \cdot 2^{q(|\fg(\p)|)}$ such rule applications.
It follows that the overall number of rule applications is bounded by 
$2^{p(|\fg(\p)|)}$
for some polynomial function $p$.
}}
\end{proof}

We now proceed to prove that the $\LnALC$ tableau algorithm is sound. 

\begin{restatable}[Soundness]{theorem}{Soundness}
	\label{thm:soundness}
	If there exists an execution of the $\LnALC$ tableau algorithm for $\p$ that constructs a complete and clash-free completion set, then $\p$ is $\LnALC$ satisfiable.
\end{restatable}
%

\begin{proof}
Suppose that  the $\LnALC$ tableau algorithm for $\p$
constructs
 an $\LnALC$-complete and clash-free completion set $\T$ for $\p$.
We define, for $n \in \mathsf{L}_{\T}$, $\psi \in \forneg(\p)$, $C \in \conneg(\p)$, and $x$ occurring in $\T$,

\begin{align*}
	\lfloor C \rfloor_{x} & = \{ n \in \mathsf{L}_{\T} \mid n : C(x) \in S_{n} \}, \\
	\lceil C \rceil_{x} & = \mathsf{L}_{\T} \setminus \{ n \in \mathsf{L}_{\T} \mid n : \dnot C(x) \in S_{n} \}, \\
	\lfloor \psi \rfloor & = \{ n \in \mathsf{L}_{\T} \mid n : \psi \in S_{n} \}, \\
	\lceil \psi \rceil & = \mathsf{L}_{\T} \setminus \{ n \in \mathsf{L}_{\T} \mid n : \dnot\psi \in S_{n}\}. \\
\end{align*}

\noindent
Moreover, define $\Gamma^{x}_{n} =  \{ \psi \mid n : \psi \in S_{n} \} \cup \{ C \mid  n : C(x) \in S_{n} \}$ and let $\gamma, \delta$ range over $\MLnALC$ formulas or concepts,
where: $\lfloor \gamma \rfloor_{x} = \lfloor \psi \rfloor$, if $\gamma = \psi$, and $\lfloor \gamma \rfloor_{x}  = \lfloor C \rfloor_{x} $, if $\gamma = C$; and similarly for $\lceil \gamma \rceil_{x}$.
%
We set $\Mmc = (\Fmc, \Imc)$, with $\Fmc = (\Wmc, \{ \Nmc_{i} \}_{i \in J})$ and $\Imc_{n} = (\Delta_{n}, \cdot^{\Imc_{n}})$, for $n \in \Wmc$, defined as follows:
\begin{itemize}
	\item $\Wmc =  \mathsf{L}_{\T}$;
	\item for every $i \in J = \{1, \ldots, n\}$, we set $\Nmc_{i} \colon \W \rightarrow 2^{2^{\Wmc}}$ such that: 

{{
		\begin{align*}
			\Nmc_{i}(n) =
			\big\{ \alpha \mid & \textnormal{\ for some \ }
			\Box_{i}\gamma_{1}  \in \Gamma^{{x}_{1}}_{n}, \ldots, \Box_{i}\gamma_{\mathsf{k}} \in \Gamma^{{x}_{\mathsf{k}}}_{n}
			\colon \\
			& \mathsf{LB}(\overline{\gamma_{\mathsf{k}}})
				\subseteq \alpha \subseteq
				\mathsf{UB}(\overline{\gamma_{\mathsf{k}}})
				\big\}
				\cup \mathsf{S};
		\end{align*}
	where:
		\begin{itemize}
			\item $\mathsf{LB}(\overline{\gamma_{\mathsf{k}}}) = \bigcap^{\mathsf{k}}_{j = 1} \lfloor \gamma_{j} \rfloor_{{{x}_{j}}}$;
			\item
			$\mathsf{UB}(\overline{\gamma_{\mathsf{k}}}) =
			\begin{cases}
				\Wmc, & \text{if $\mathbf{M} \in L$} \\
				\bigcap^{\mathsf{k}}_{j = 1} \lceil \gamma_{j} \rceil_{{{x}_{j}}}, & \text{if $\mathbf{M} \not \in L$}
			\end{cases};
			$
			\item
			$\mathsf{k}
			\begin{cases}
				\geq 1, & \text{if $\mathbf{C} \in L$} \\
				= 1, & \text{if $\mathbf{C} \not \in L$}
			\end{cases};
			$
			\item
			$\mathsf{S} =
			\begin{cases}
				\{ \Wmc \}, & \text{if $\mathbf{N} \in L$} \\
				\emptyset, & \text{if $\mathbf{N} \not \in L$}
			\end{cases};
			$
		\end{itemize}

}}

	\item
	{{
	$\Delta_{n} = \{ x \mid x \ \text{is a term occurring in} \ S_{n} \}$;
	}}
	\item $A^{\Imc_{n}} = \{ x \in \Delta_{n} \mid n : A(x) \in S_{n} \}$;
	\item
	{{
	$a^{\Imc_{n}} =
	\begin{cases}
		a, & \text{if $a$ occurs in $S_{n}$} \\
		\text{arbitrary}, & \text{otherwise}
	\end{cases}$
	;
	}}
	\item $r^{\Imc_{n}} = \{ (x, y) \in \Delta_{n} \times \Delta_{n} \mid n : r(x, y) \in S_{n} \ \text{or} \ n : r(z, y) \in S_{n}, $
	for some $z$ blocking $x$ in $S_{n}\}$.
\end{itemize}

We require the following claims.

\begin{claim}
\label{cla:modelcond}
  For $\mathbf{X}\in\{\mathbf{M,C,N,T,P,Q,D}\}$,
  if $\mathbf{X}\in\Lvar$,
  then $\Mmc$ satisfies the $X$-condition.
\end{claim}
\begin{proof}[Proof of Claim]

\quad

\begin{enumerate}[leftmargin=*, align=left]
	\item[$\mathbf{M}\in\Lvar$.] Suppose that $\alpha\in \Nmc_{i}(n)$ and $\alpha \subseteq \beta \subseteq \Wmc$. 
By definition, there are
$\Box_{i}\gamma_{1}  \in \Gamma^{{x}_{1}}_{n}, \ldots, \Box_{i}\gamma_{k} \in \Gamma^{{x}_{k}}_{n}$ 
such that $\mathsf{LB}(\overline{\gamma_{k}}) \subseteq \alpha$. 
Then $\mathsf{LB}(\overline{\gamma_{k}}) \subseteq \beta$, 
hence $\beta \in \Nmc_{i}(n)$.

	\item[$\mathbf{C}\in\Lvar$.] Suppose that $\alpha,\beta\in \Nmc_{i}(n)$.
Then there are
$\Box_{i}\gamma_{1}  \in \Gamma^{{x}_{1}}_{n}, \ldots, \Box_{i}\gamma_{k} \in \Gamma^{{x}_{k}}_{n}$ 
such that $\mathsf{LB}(\overline{\gamma_{k}}) \subseteq \alpha \subseteq \mathsf{UB}(\overline{\gamma_{k}})$,
and there are
$\Box_{i}\delta_{1}  \in \Gamma^{{x}_{1}}_{n}, \ldots, \Box_{i}\delta_{h} \in \Gamma^{{x}_{h}}_{n}$ 
such that $\mathsf{LB}(\overline{\delta_{h}}) \subseteq \beta \subseteq \mathsf{UB}(\overline{\delta_{h}})$.
Then 
$\mathsf{LB}(\overline{\gamma_{k}}) \cap \mathsf{LB}(\overline{\delta_{h}}) = \mathsf{LB}(\overline{\gamma_{k},\delta_h}) \subseteq \alpha\cap\beta
\subseteq \mathsf{UB}(\overline{\gamma_{k}}) \cap \mathsf{UB}(\overline{\delta_{h}}) = \mathsf{UB}(\overline{\gamma_{k},\delta_h})$,
which implies $\alpha\cap\beta\in \Nmc_{i}(n)$.
				 
	\item[$\mathbf{N}\in\Lvar$.] By construction, $\Wmc \in \Nmc_{i}(n)$ for all $n \in \Wmc$.
	
	\item[$\mathbf{P}\in\Lvar$.] Suppose that $\alpha\in\Nmc_{i}(n)$. 
Then there are
$\Box_{i}\gamma_{1}  \in \Gamma^{{x}_{1}}_{n}, \ldots, \Box_{i}\gamma_{k} \in \Gamma^{{x}_{k}}_{n}$ 
such that $\mathsf{LB}(\overline{\gamma_{k}}) \subseteq \alpha \subseteq \mathsf{UB}(\overline{\gamma_{k}})$.
Since $\T$ is $\LnALC$-complete,
by the rule $\mathsf{R}_{\mathit{L}\mathbf{P}}$,
there is 
$m$
such that 
$m: \gamma_{j} \in \T$ for all $1 \leq j \leq k$, that is $m\in\mathsf{LB}(\overline{\gamma_{k}})$.
Then $\alpha\not=\emptyset$.

	\item[$\mathbf{Q}\in\Lvar$.] Suppose that $\alpha\in\Nmc_{i}(n)$. 
Then there are
$\Box_{i}\gamma_{1}  \in \Gamma^{{x}_{1}}_{n}, \ldots, \Box_{i}\gamma_{k} \in \Gamma^{{x}_{k}}_{n}$ 
such that $\mathsf{LB}(\overline{\gamma_{k}}) \subseteq \alpha \subseteq \mathsf{UB}(\overline{\gamma_{k}})$.
Since $\T$ is $\LnALC$-complete,
by the rule $\mathsf{R}_{\mathit{L}\mathbf{Q}}$,
there is $m$ 
such that 
$m: \dot{\lnot}\gamma_{j} \in \T$ for some $1 \leq j \leq k$, that is $m\in\Wmc$ and $m\notin\lceil \gamma_{j} \rceil_{{{x}_{j}}}$,
hence $m\notin\mathsf{UB}(\overline{\gamma_{k}})$.
Then $\alpha\not=\Wmc$.

	\item[$\mathbf{D}\in\Lvar$.] Suppose that $\alpha,\beta\in \Nmc_{i}(n)$.
Then there are
$\Box_{i}\gamma_{1}  \in \Gamma^{{x}_{1}}_{n}, \ldots, \Box_{i}\gamma_{k} \in \Gamma^{{x}_{k}}_{n}$ 
such that $\mathsf{LB}(\overline{\gamma_{k}}) \subseteq \alpha \subseteq \mathsf{UB}(\overline{\gamma_{k}})$,
and there are
$\Box_{i}\delta_{1}  \in \Gamma^{{x}_{1}}_{n}, \ldots, \Box_{i}\delta_{h} \in \Gamma^{{x}_{h}}_{n}$ 
such that $\mathsf{LB}(\overline{\delta_{h}}) \subseteq \beta \subseteq \mathsf{UB}(\overline{\delta_{h}})$.
Since $\T$ is $\LnALC$-complete,
b the rule $\mathsf{R}_{\mathit{L}\mathbf{D}}$,
there is $m$ 
such that 
$m: \gamma_{j}, m: \delta_{\ell} \in \T$
for all $1\leq j\leq k$, $1\leq \ell \leq h$;
or
$m: \dot{\lnot}\gamma_{j}, m: \dot{\lnot}\delta_{\ell} \in \T$
for some $1\leq j\leq k$, $1\leq \ell \leq h$.
In the first case,
$m \in \mathsf{LB}(\overline{\gamma_{k}}) \cap \mathsf{LB}(\overline{\delta_{h}})$,
hence $m\in\alpha\cap\beta$.
In the second case, 
$m\in(\Wmc\setminus\mathsf{UB}(\overline{\gamma_{k}})) \cap (\Wmc\setminus\mathsf{UB}(\overline{\delta_{h}}))$,
hence 
$m\in(\Wmc\setminus\alpha)\cap(\Wmc\setminus\beta)$.
In either case $\beta\neq\Wmc\setminus\alpha$.

	\item[$\mathbf{T}\in\Lvar$.] Suppose that $\alpha\in\Nmc_{i}(n)$. 
Then there are
$\Box_{i}\gamma_{1}  \in \Gamma^{{x}_{1}}_{n}, \ldots, \Box_{i}\gamma_{k} \in \Gamma^{{x}_{k}}_{n}$ 
such that $\mathsf{LB}(\overline{\gamma_{k}}) \subseteq \alpha \subseteq \mathsf{UB}(\overline{\gamma_{k}})$.
Since $\T$ is $\LnALC$-complete,
by the rule $\mathsf{R}_{\mathit{L}\mathbf{T}}$,
$n: \gamma_{j} \in \T$ for all $1 \leq j \leq k$,
then $n\in \mathsf{LB}(\overline{\gamma_{k}})$,
thus $n\in\alpha$.\qedhere
\end{enumerate}
\end{proof}

%
\begin{claim}
\label{cla:conind}
For every $n \in \Wmc$, $C \in \conneg(\p)$, and $x \in \Delta_{n}$: if $n : C(x) \in S_{n}$, then $x \in C^{\Imc_{n}}$.
\end{claim}
\begin{proof}[Proof of Claim]
We show the claim by induction on the weight of $C$ (in NNF).
The base case of $C = A$ comes immediately from the definitions.
For the base case of $C = \lnot A$, suppose that $n : \lnot A(x) \in S_{n}$. Since $\T$ is clash-free, we have that $n : A(x) \not \in S_{n}$, and thus $x \not \in A^{\Imc_{n}}$ by definition of $A^{\Imc_{n}}$, meaning $x \in (\lnot A)^{\Imc_{n}}$.
The inductive cases of $C = D \sqcap E$ and $C = D \sqcup E$ come from the fact that $S_{n}$ is closed under $\mathsf{R}_{\sqcap}$ and $\mathsf{R}_{\sqcup}$, respectively, and straightforward applications of the inductive hypothesis.
We show the remaining cases (cf. also~\cite[Claim 15.2]{GabEtAl03}).

\begin{enumerate}[leftmargin=*, align=left]
	\item[$C = \exists r.D$.]
Let $n : \exists r.D(x) \in S_{n}$, meaning that $\exists r.D \in \Gamma^{x}_{n}$. We distinguish two cases.
\begin{itemize}
\item[$(i)$] $x$ is not blocked by any variable in $S_{n}$. Since $S_{n}$ is closed under $\mathsf{R}_{\exists}$, there exists $y$ occurring in $S_{n}$ such that $n : r(x,y) \in S_{n}$ and $n : D(y) \in S_{n}$. Thus, by definition, $(x, y) \in r^{\Imc_{n}}$ and $n : D(y) \in S_{n}$. By inductive hypothesis, we obtain that $x \in (\exists r.D)^{\Imc_{n}}$.

\item[$(ii)$] $x$ is blocked by a variable in $S_{n}$, implying that there exists a $<$-minimal (since $<$ is a well-ordering) $y$ occurring in $S_{n}$ such that $y < x$ and $\{ E \mid n : E(x) \in S_{n} \} \subseteq \{ E \mid n : E(y) \in S_{n} \}$.
In turn, this implies that $y$ is not blocked by any other variable $z$ in $S_{n}$ (for otherwise $z$ would block $x$, with $z < y$, against the fact that $y$ is $<$-minimal).
By reasoning as in the case above, since $y$ is not blocked and $S_{n}$ is closed under $\mathsf{R}_{\exists}$, we have a variable $z$ occurring in $S_{n}$ such that $n : r(y,z) \in S_{n}$ and $n : D(z) \in S_{n}$.
Since $y$ blocks $x$, by definition we have that $(x, z) \in r^{\Imc_{n}}$, and by inductive hypothesis we get from $n : D(z)$ that $z \in D^{\Imc_{n}}$.
Thus, $x \in (\exists r.D)^{\Imc_{n}}$.
\end{itemize}

\item[$C = \forall r.D$.]
Let $n : \forall r.D(x) \in S_{n}$, meaning that $\forall r.D \in \Gamma^{x}_{n}$, and suppose that $(x, y) \in r^{\Imc_{n}}$. By definition, either $n : r(x,y) \in S_{n}$ or $n : r(z,y) \in S_{n}$, for some $z$ blocking $x$ in $S_{n}$.
In the former case, since $S_{n}$ is closed under $\mathsf{R}_{\forall}$, we get that $n : D(y) \in S_{n}$.
In the latter case, since $z$ blocks $x$ in $S_{n}$, we obtain $n : \forall r.D(z) \in S_{n}$; again, since $S_{n}$ is closed under $\mathsf{R}_{\forall}$, this implies that $n : D(y) \in S_{n}$.
Hence, in both cases, we have $n : D(y) \in S_{n}$.
By inductive hypothesis, this means that $y \in D^{\Imc_{n}}$.
Since $y$ was arbitrary, we conclude that $x \in (\forall r.D)^{\Imc_{n}}$.


%

\item[$C = \Box_{i} D$.]
Let $n : \Box_{i} D(x) \in S_{n}$, meaning that $\Box_{i} D \in \Gamma^{x}_{n}$.
		We have by inductive hypothesis that
	$\lfloor D \rfloor_{x} = \{ n \in \Wmc \mid n : D(x) \in S_{n} \} \subseteq \{ n \in \Wmc \mid x \in D^{\Imc_{n}} \} = \llbracket D \rrbracket^{\Mmc}_{x}$.
	By inductive hypothesis 
	(since $| \dnot D | = | D |$), 
	we also have that
	$\{ n \in \Wmc \mid n : \dnot D(x) \in S_{n} \} \subseteq \{ n \in \Wmc \mid x \in (\dnot D)^{\Imc_{n}} \} = \llbracket \dnot D \rrbracket^{\Mmc}_{x} = \Wmc \setminus \llbracket D \rrbracket^{\Mmc}_{x}$.
	Hence, $\llbracket D \rrbracket^{\Mmc}_{x} \subseteq \Wmc \setminus \{ w \in \Wmc \mid n : \dnot D(x) \in S_{n} \} = \lceil D \rceil_{x}$. In conclusion, we have $\Box_{i} D \in \Gamma^{x}_{n} $ such that $\lfloor D \rfloor_{x} \subseteq \llbracket D \rrbracket^{\Mmc}_{x} \subseteq \lceil D \rceil_{x}$. Thus, by definition, $\llbracket D \rrbracket^{\Mmc}_{x} \in \Nmc_{i}(n)$, as required.
(If $\mathbf{M}\in\Lvar$, $\lfloor D \rfloor_{x} \subseteq \llbracket D \rrbracket^{\Mmc}_{x}$, and by definition this means $\llbracket D \rrbracket^{\Mmc}_{x} \in \Nmc_{i}(n)$, as required.)

	\item[$C = \Diamond_{i} D$.]
Let $n : \Diamond_{i}D(x) \in S_{n}$. 
We distinguish two cases.

\begin{itemize}
\item[$(i)$] There exists no $\Box_{i} \gamma \in \Gamma^{y}_{n}$.
Then if $\mathbf{N}\notin\Lvar$, $\Nmc_{i}(n) = \emptyset$, thus $\Wmc \setminus \llbracket D \rrbracket^{\Mmc}_{x} \not \in \Nmc_{i}(n)$, meaning that $x \in (\Diamond_{i}D)^{\Imc_{n}}$.
If instead $\mathbf{N}\in\Lvar$,
then $\Nmc_{i}(n) = \Wmc$.
Moreover, since $\T$ is $\LnALC$-complete, 
by the rule $\mathsf{R}_{\mathit{L}\mathbf{N}}$,
there is 
$m$
such that 
$m : D(x) \in S_{m}$. By inductive hypothesis, this implies $x \in D^{\Imc_{m}}$, that is, $\llbracket D \rrbracket^{\Mmc}_{x} \neq \emptyset$. Then we have $\Wmc \setminus \llbracket D \rrbracket^{\Mmc}_{x} \neq \Wmc$, and thus $\Wmc \setminus \llbracket D \rrbracket^{\Mmc}_{x} \not \in \Nmc_{i}(n)$. Hence, $x \in (\Diamond_{i}D)^{\Imc_{n}}$.

\item[$(ii)$] There exist $\Box_{i} \gamma_{1} \in \Gamma^{y_{1}}_{n}, \ldots, \Box_{i} \gamma_{k} \in \Gamma^{y_{k}}_{n}$.
		Since $\T$ is $\LnALC$-complete, there exists $m \in \Wmc$ such that:
$\gamma_{1}  \in \Gamma^{y_{1}}_{m}, \ldots, \gamma_{k} \in \Gamma^{y_{k}}_{m}$ and $D \in \Gamma^{x}_{m}$; or
$\dnot \gamma_{j} \in \Gamma^{y_{j}}_{m}$ and $\dnot D \in \Gamma^{x}_{m}$, for some $j\leq k$.
				By inductive hypothesis, the previous step implies that there exists $m \in \Wmc$ such that:
$\gamma_{1} \in \Gamma^{y_{1}}_{m}, \ldots, \gamma_{k} \in \Gamma^{y_{k}}_{m}$ and $x \in D^{\Imc_{m}}$; or
$\dnot \gamma_{j} \in \Gamma^{y_{j}}_{m}$ and $x \in \dnot D^{\Imc_{m}}$, for some $j\leq k$.
Thus
$\bigcap_{j = 1}^{k} \lfloor \gamma_{j} \rfloor_{y_{j}} \not\subseteq \Wmc \setminus \llbracket D \rrbracket^{\Mmc}_{x}$; or
$\Wmc \setminus \llbracket D \rrbracket^{\Mmc}_{x} \not\subseteq \bigcap_{j = 1}^{k}\lceil \gamma_{l} \rceil_{y_{l}}$.
Since this holds for every $\Box_{i} \gamma_{1} \in \Gamma^{y_{1}}_{n}, \ldots, \Box_{i} \gamma_{k} \in \Gamma^{y_{k}}_{n}$, we conclude that $\Wmc \setminus  \llbracket D \rrbracket^{\Mmc}_{x} \not \in \Nmc_{i}(n)$, i.e., $x \in (\Diamond_{i}D)^{\Imc_{n}}$, as required.
(If $\mathbf{M}\in\Lvar$, 
there exists $m \in \Wmc$ such that
$\gamma_{1}  \in \Gamma^{y_{1}}_{m}, \ldots, \gamma_{k} \in \Gamma^{y_{k}}_{m}$ and $D \in \Gamma^{x}_{m}$,
thus
$x \in D^{\Imc_{m}}$,
hence
$\bigcap_{j = 1}^{k} \lfloor \gamma_{j} \rfloor_{y_{j}} \not\subseteq \Wmc \setminus \llbracket D \rrbracket^{\Mmc}_{x}$,
therefore $\Wmc \setminus  \llbracket D \rrbracket^{\Mmc}_{x} \not \in \Nmc_{i}(n)$.)\qedhere
\end{itemize}
\end{enumerate}
\end{proof}

\begin{claim}
\label{cla:forind}
For every $n \in \Wmc$ and $\psi \in \conneg(\p)$: if $n : \psi \in S_{n}$, then $\Mmc, n \models \psi$.
\end{claim}
\begin{proof}[Proof of Claim]
We prove the claim by induction on the weight of $\p$ (in NNF).


\begin{enumerate}[leftmargin=*, align=left]
	\item[$\psi = C(a)$.] 
	{{Let $n : C(a) \in S_{n}$. By definition of $\Imc_{n}$ and Claim~\ref{cla:conind}, we have that $a^{\Imc_{n}} \in C^{\Imc_{n}}$, hence $\Mmc, n \models C(a)$. (For $\psi = \lnot C(a)$, recall that $\lnot C(a)$ is equivalent to $D(a)$ with $D = \lnot C$).}}

	\item[$\psi = r(a,b)$.]
{{Let $n: r(a,b) \in S_{n}$. By definition of $\Imc_{n}$, this implies $(a^{\Imc_{n}}, b^{\Imc_{n}}) \in r^{\Imc_{n}}$, hence $\Mmc, n \models r(a,b)$.}}

	\item[$\psi = \lnot r(a,b)$.]
{{ Let $n: \lnot r(a,b) \in S_{n}$. Since $\mathbf{T}$ is clash-free, we have that $n: r(a,b) \not \in S_{n}$. Thus, by definition
of $\Imc_{n}$
we have $(a^{\Imc_{n}}, b^{\Imc_{n}}) \not \in r^{\Imc_{n}}$, meaning that $\Mmc, n \not \models r(a,b)$.
}}

	\item[$\psi = (\top \sqsubseteq C)$.] Let $n : \top \sqsubseteq C \in S_{n}$ and let $x \in \Delta_{n}$. Since $S_{n}$ is closed under $(\mathsf{R}_{\sqsubseteq})$ and $x$ occurs in $S_{n}$, we have that $n : C(x) \in S_{n}$. By Claim~\ref{cla:conind}, we have that $x \in C^{\Imc_{n}}$. Given that $x$ is arbitrary, we conclude that $\Mmc, n \models \top \sqsubseteq C$.

	\item[$\psi = \lnot (\top \sqsubseteq C)$.] Let $n : \lnot (\top \sqsubseteq C) \in S_{n}$. Since $S_{n}$ is closed under $(\mathsf{R}_{\not\sqsubseteq})$, there exists $x$ occurring in $S_{n}$ such that $n : \dnot C(x) \in S_{n}$. By Claim~\ref{cla:conind}, we obtain that $x \in (\dnot C)^{\Imc_{n}}$, for some $x \in \Delta_{w}$. Hence, $\Mmc, n \models \lnot (\top \sqsubseteq C)$.
\end{enumerate}

The inductive cases of 
$\psi = \chi \land \vartheta$
and
$\psi = \chi \lor \vartheta$ follow from the definitions and straighforward applications of the inductive hypothesis.
%
Moreover the inductive cases of 
$\psi = \Box_{i} \chi$
and 
$\psi = \Diamond_{i} \chi$ can be proved analogously to Claim~\ref{cla:conind}.
\end{proof}

Since,
by definition,
we have 
$0 : \p \in S_{0} \subseteq \mathbf{T}$,
thanks to Claim~\ref{cla:forind} we obtain $\Mmc, 0 \models \p$.
{{Moreover,
by Claim~\ref{cla:modelcond}, $\Mmc$ is a
$\mathit{L}^n$ model.
Therefore $\p$ is $\LnALC$ satisfiable.}}
\end{proof}


We finally show completeness of the $\LnALC$ tableau algorithm.

\begin{restatable}[Completeness]{theorem}{Completeness}
	\label{thm:completeness}
	If $\p$ is $\LnALC$ satisfiable, then there exists an execution of the $\LnALC$ tableau algorithm for $\p$ that constructs a complete and clash-free completion set.
\end{restatable}
%

\begin{proof}
In the proof we assume $\mathbf{C}\in\Lvar$,
for the case $\mathbf{C}\notin\Lvar$ consider $k = h = 1$.
Let $\Mmc = (\Fmc, \Imc)$ be an $\LnALC$-model satisfying $\p$, with $\Fmc = (\Wmc, \{ \Nmc \}_{i \in J})$, i.e.,
$\Mmc, w_{\p} \models \p$, for some $w_{\p} \in \Wmc$.
We require the following definitions and technical results.
%
First, we let $\gamma, \delta$ (possibly indexed) range over $\MLnALC$ concepts and formulas, with $\llbracket \gamma \rrbracket^{\Mmc}_{d} = \llbracket \psi \rrbracket^{\Mmc}$, if $\gamma = \psi$, and $\llbracket \gamma \rrbracket^{\Mmc}_{d} = \llbracket C \rrbracket^{\Mmc}_{d}$, if $\gamma = C$.
Then, for $w \in \Wmc$ and $d \in \bigcup_{v \in \Wmc} \Delta_{v}$, define
$\Phi^{d}_{w} = \{ \psi \in \forneg(\p) \mid \Mmc, w \models \psi \} \cup \{ C \in \conneg(\p) \mid d \in C^{\Imc_{w}} \}$.
Observe that, if $C \in \Phi^{d}_{w}$, then $d \in \Delta_{w}$.
Moreover, given a completion set $\T$ for $\p$
and $S_{n} \subseteq \T$,
let $\Gamma^{x}_{n} = \{ \psi \mid n : \psi \in S_{n} \} \cup \{ C \mid n : C(x) \in S_{n} \}$.
We say that a completion set $\T$ for $\p$ is \emph{$\Mmc$-compatible} if
there exists a function $\pi$ from $\mathsf{L}_{\T}$ to $\Wmc$, and, for every $n \in \mathsf{L}_{\T}$, there exists a function $\pi_{n}$ from the set of 
{{terms}}
occurring in $S_{n}$ to $\Delta_{\pi(n)}$, such that
$\gamma \in \Gamma^{x}_{n}$ implies $\gamma \in \Phi^{\pi_{n}(x)}_{\pi(n)}$.
We  require the following claim.

\begin{claim}
\label{cla:compatible}
If a completion set $\T$ for $\p$ is $\Mmc$-compatible
{{and $\Mmc$ is an $\LnALC$-model}},
then for every  $\LnALC$-rule $\mathsf{R}$ applicable to $\T$, there exists a completion set $\T'$ obtained from $\T$ by an application of $\mathsf{R}$ such that $\T'$ is $\Mmc$-compatible.
\end{claim}
\begin{proof}
{{Given an $\Mmc$-compatible completion set $\T$ for $\p$ and a label $n \in \mathsf{L}_{\T}$, let $\pi$ and $\pi_{n}$ be the functions provided by the definition of $\Mmc$-compatibility.
We need to consider each $\LnALC$-rule $\mathsf{R}$.
 For $\mathsf{R} \in \{ \mathsf{R}_{\land}, \mathsf{R}_{\lor}, \mathsf{R}_{\sqcap}, \mathsf{R}_{\sqcup}, \mathsf{R}_{\forall}, \mathsf{R}_{\exists}, \mathsf{R}_{\sqsubseteq}, \mathsf{R}_{\not\sqsubseteq} \}$, we proceed similarly to~\cite[Claim 15.14]{GabEtAl03}.
Here we consider 
the modal rules.

\begin{enumerate}[leftmargin=*, align=left]
	\item[($\mathsf{R}_{\mathit{L}}$)]
Suppose that $\mathsf{R}_{\mathit{L}}$ is applicable to $\T$.
	Then there are $\Box_{i} \gamma_{1} \in \Gamma^{x_{1}}_{n}, \ldots, \Box_{i} \gamma_{k} \in \Gamma^{x_{k}}_{n}, \Diamond_{i} \delta \in \Gamma^{y}_{n}$.
	 Since $\T$ is $\Mmc$-compatible,
	 we have that $\Box_{i}\gamma_{1} \in \Phi^{\pi_{n}(x_{1})}_{\pi(n)}, \ldots, \Box_{i} \gamma_{k} \in \Phi^{\pi_{n}(x_{k})}_{\pi(n)}$ and $\Diamond_{i} \delta \in \Phi^{\pi_{n}(y)}_{\pi(n)}$,
	 meaning that $\llbracket \gamma_{j} \rrbracket^{\Mmc}_{d_{j}} \in \Nmc_{i}(\pi(n))$, for $j = 1, \ldots, k$, 
	 hence by the $C$-condition  $\bigcap_{j = 1}^{k} \llbracket \gamma_{j} \rrbracket^{\Mmc}_{d_{j}}  \in \Nmc_{i}(\pi(n))$,
	 and $\Wmc \setminus \llbracket \delta \rrbracket^{\Mmc}_{e} \not \in \Nmc_{i}(\pi(n))$, i.e., $\llbracket \dnot \delta \rrbracket^{\Mmc}_{e} \not \in \Nmc_{i}(\pi(n))$.
	 Then 
	 $\bigcap_{j = 1}^{k} \llbracket \gamma_{j} \rrbracket^{\Mmc}_{d_{j}} \neq \llbracket \dnot \delta \rrbracket^{\Mmc}_{e}$
	 (if $\mathbf{M}\in\Lvar$, 
	 $\bigcap_{j = 1}^{k} \llbracket \gamma_{j} \rrbracket^{\Mmc}_{d_{j}} \not\subseteq \llbracket \dnot \delta \rrbracket^{\Mmc}_{e}$).
It follows that there exists $v \in \Wmc$ such that
$\gamma_{1} \in \Phi^{\pi_{n}(x_{1})}_{v}, \ldots, \gamma_{k} \in \Phi^{\pi_{n}(x_{k})}_{v}$ and $\delta \in \Phi^{\pi_{n}(y)}_{v}$; or
$\dnot \gamma_{j} \in \Phi^{\pi_{n}(x_{j})}_{v}$ and $\dnot \delta \in \Phi^{\pi_{n}(y)}_{v}$,
for some $j\leq k$.
%
Then by applying the rule $\mathsf{R}_{\mathit{L}}$ 
accordingly, 
we expand $\T$ to $\T'$ with
$m : \gamma_1, \ldots, m : \gamma_k, m : \delta$, or with $m : \dot{\lnot}\gamma_j, m : \dnot \delta $, for some $j\leq k$,
for some $m$ satisfying the application condition of $\mathsf{R}_{\mathit{L}}$.
Since $m$ is fresh, we can extend $\pi$ with $\pi(m) = v$, and $\pi_{m}$ with $\pi_{m}(x_{1}) = \pi_{n}(x_{1})$, \ldots, $\pi_{m}(x_{k}) = \pi_{n}(x_{k})$, $\pi_{m}(y) = \pi_{n}(y)$, thus obtaining that $\T'$ is $\Mmc$-compatible.

	\item[($\mathsf{R}_{\mathit{L}\mathbf{N}}$)]
	Suppose that $\mathsf{R}_{\mathit{L}\mathbf{N}}$ is applicable to $\T$.
	Then there is $\Diamond_{i} \delta \in \Gamma^{y}_{n}$.
	 Since $\T$ is $\Mmc$-compatible,
	 we have that $\Diamond_{i} \delta \in \Phi^{\pi_{n}(y)}_{\pi(n)}$,
	 meaning that $\Wmc \setminus \llbracket \delta \rrbracket^{\Mmc}_{e} \not \in \Nmc_{i}(\pi(n))$.
	 At the same time, by the $N$-condition, $\Wmc \in \Nmc_{i}(\pi(n))$,
	  hence $\llbracket \delta \rrbracket^{\Mmc}_{e} \not= \emptyset$,
	 that is there exists $v \in \Wmc$ such that $\delta \in \Phi^{\pi_{n}(y)}_{v}$. 
	Then we expand $\T$ with $m : \delta$, for some $m$ satisfying the application condition of $\mathsf{R}_{\mathit{L}\mathbf{N}}$.
Since $m$ is fresh, we can extend $\pi$ with $\pi(m) = v$, and $\pi_{m}$ with $\pi_{m}(y) = \pi_{n}(y)$, thus obtaining that $\T'$ is $\Mmc$-compatible.
	
	\item[($\mathsf{R}_{\mathit{L}\mathbf{P}}$)]
Suppose that $\mathsf{R}_{\mathit{L}\mathbf{P}}$ is applicable to $\T$.
	Then there are $\Box_{i} \gamma_{1} \in \Gamma^{x_{1}}_{n}, \ldots, \Box_{i} \gamma_{k} \in \Gamma^{x_{k}}_{n}$.
	 Since $\T$ is $\Mmc$-compatible,
	 we have that $\Box_{i}\gamma_{1} \in \Phi^{\pi_{n}(x_{1})}_{\pi(n)}, \ldots, \Box_{i} \gamma_{k} \in \Phi^{\pi_{n}(x_{k})}_{\pi(n)}$, 
	 meaning that $\llbracket \gamma_{j} \rrbracket^{\Mmc}_{d_{j}} \in \Nmc_{i}(\pi(n))$, for $j = 1, \ldots, k$,
	 hence $\bigcap_{j = 1}^{k} \llbracket \gamma_{j} \rrbracket^{\Mmc}_{d_{j}} \in \Nmc_{i}(\pi(n))$.
	 At the same time, by the $P$-condition, 
	 $\bigcap_{j = 1}^{k} \llbracket \gamma_{j} \rrbracket^{\Mmc}_{d_{j}} \neq \emptyset$,
	that is there exists $v \in \Wmc$ such that
	$\gamma_{1} \in \Phi^{\pi_{n}(x_{1})}_{v}, \ldots, \gamma_{k} \in \Phi^{\pi_{n}(x_{k})}_{v}$.
	Then we expand $\T$ with $m : \gamma_1, \ldots, m : \gamma_k$,
	for some $m$ satisfying the application condition of $\mathsf{R}_{\mathit{L}\mathbf{P}}$.
	Since $m$ is fresh, we can extend $\pi$ with $\pi(m) = v$, and $\pi_{m}$ with $\pi_{m}(x_{1}) = \pi_{n}(x_{1})$, \ldots, $\pi_{m}(x_{k}) = \pi_{n}(x_{k})$, thus obtaining that $\T'$ is $\Mmc$-compatible.

	\item[($\mathsf{R}_{\mathit{L}\mathbf{Q}}$)]
	Suppose that $\mathsf{R}_{\mathit{L}\mathbf{Q}}$ is applicable to $\T$.
	Then there are $\Box_{i} \gamma_{1} \in \Gamma^{x_{1}}_{n}, \ldots, \Box_{i} \gamma_{k} \in \Gamma^{x_{k}}_{n}$.
	 Since $\T$ is $\Mmc$-compatible,
	 we have that $\Box_{i}\gamma_{1} \in \Phi^{\pi_{n}(x_{1})}_{\pi(n)}, \ldots, \Box_{i} \gamma_{k} \in \Phi^{\pi_{n}(x_{k})}_{\pi(n)}$, 
	 meaning that $\llbracket \gamma_{j} \rrbracket^{\Mmc}_{d_{j}} \in \Nmc_{i}(\pi(n))$, for $j = 1, \ldots, k$,
	 hence $\bigcap_{j = 1}^{k} \llbracket \gamma_{j} \rrbracket^{\Mmc}_{d_{j}} \in \Nmc_{i}(\pi(n))$.
	 At the same time, by the $Q$-condition, 
	 $\bigcap_{j = 1}^{k} \llbracket \gamma_{j} \rrbracket^{\Mmc}_{d_{j}} \neq \Wmc$,
	that is there exists $v \in \Wmc$ such that
	$\gamma_{j} \notin \Phi^{\pi_{n}(x_{j})}_{v}$ for some $j\leq k$.
	Then by applying $\mathsf{R}_{\mathit{L}\mathbf{Q}}$ accordingly, we expand $\T$ with $m : \dot{\lnot}\gamma_j$,
	for some $m$ satisfying the application condition of $\mathsf{R}_{\mathit{L}\mathbf{Q}}$.
	Since $m$ is fresh, we can extend $\pi$ with $\pi(m) = v$, and 
	{{$\pi_{m}$ with $\pi_{m}(x_{j}) = \pi_{n}(x_{j})$,  thus obtaining that $\T'$ is $\Mmc$-compatible.}}

	\item[($\mathsf{R}_{\mathit{L}\mathbf{D}}$)]
	Suppose that $\mathsf{R}_{\mathit{L}\mathbf{D}}$ is applicable to $\T$.
	Then there are $\Box_{i} \gamma_{1} \in \Gamma^{x_{1}}_{n}, \ldots, \Box_{i} \gamma_{k} \in \Gamma^{x_{k}}_{n}, \Box_{i} \delta_{1} \in \Gamma^{y_{1}}_{n}, \ldots, \Box_{i} \delta_{h} \in \Gamma^{y_{h}}_{n}$.
	Since $\T$ is $\Mmc$-compatible,
	we have that $\Box_{i}\gamma_{1} \in \Phi^{\pi_{n}(x_{1})}_{\pi(n)}, \ldots, \Box_{i} \gamma_{k} \in \Phi^{\pi_{n}(x_{k})}_{\pi(n)}$,
	and $\Box_{i}\delta_{1} \in \Phi^{\pi_{n}(y_{1})}_{\pi(n)}, \ldots, \Box_{i} \delta_{h} \in \Phi^{\pi_{n}(y_{h})}_{\pi(n)}$, 
	meaning that $\llbracket \gamma_{j} \rrbracket^{\Mmc}_{d_{j}} \in \Nmc_{i}(\pi(n))$, for $j = 1, \ldots, k$,
	and $\llbracket \delta_{\ell} \rrbracket^{\Mmc}_{e_{\ell}} \in \Nmc_{i}(\pi(n))$, for $\ell = 1, \ldots, h$;
	hence $\bigcap_{j = 1}^{k} \llbracket \gamma_{j} \rrbracket^{\Mmc}_{d_{j}} \in \Nmc_{i}(\pi(n))$,
	and $\bigcap_{\ell = 1}^{h} \llbracket \delta_{\ell} \rrbracket^{\Mmc}_{e_{\ell}} \in \Nmc_{i}(\pi(n))$.
	By the $D$-condition, 
	$\bigcap_{\ell = 1}^{h} \llbracket \delta_{\ell} \rrbracket^{\Mmc}_{e_{\ell}} \not=\Wmc \setminus \bigcap_{j = 1}^{k} \llbracket \gamma_{j} \rrbracket^{\Mmc}_{d_{j}}$
	(if $\mathbf{M}\in\Lvar$, $\bigcap_{\ell = 1}^{h} \llbracket \delta_{\ell} \rrbracket^{\Mmc}_{e_{\ell}} \not \subseteq \Wmc \setminus \bigcap_{j = 1}^{k} \llbracket \gamma_{j} \rrbracket^{\Mmc}_{d_{j}}$).
	This means that there exists $v \in \Wmc$ such that
$\gamma_{1} \in \Phi^{\pi_{n}(x_{1})}_{v}, \ldots, \gamma_{k} \in \Phi^{\pi_{n}(x_{k})}_{v},
\delta_{1} \in \Phi^{\pi_{n}(y_{1})}_{v}, \ldots, \delta_{h} \in \Phi^{\pi_{n}(y_{h})}_{v}$; or
$\dnot \gamma_{j} \in \Phi^{\pi_{n}(x_{j})}_{v}$ and $\dnot \delta_{\ell} \in \Phi^{\pi_{n}(y_{\ell})}_{v}$,
for some $j\leq k$, $\ell \leq h$.
Then by applying the rule $\mathsf{R}_{\mathit{L}\mathbf{D}}$ 
accordingly, 
we expand $\T$ to $\T'$ with
$m : \gamma_1, \ldots, m : \gamma_k, m : \delta_1, \ldots, m : \delta_h$, or with $m : \dot{\lnot}\gamma_j, m : \dnot \delta_{\ell} $, for some $j\leq k$, $\ell \leq h$,
and some $m$ satisfying the application condition of $\mathsf{R}_{\mathit{L}\mathbf{D}}$.
Since $m$ is fresh, we can extend $\pi$ with $\pi(m) = v$, 
and $\pi_{m}$ with $\pi_{m}(x_{j}) = \pi_{n}(x_{j})$, for $j \leq k$, 
$\pi_{m}(y_{\ell}) = \pi_{n}(y_{\ell})$, for $\ell \leq h$; thus obtaining that $\T'$ is $\Mmc$-compatible.

	\item[($\mathsf{R}_{\mathit{L}\mathbf{T}}$)]
	Suppose that $\mathsf{R}_{\mathit{L}\mathbf{T}}$ is applicable to $\T$.
	Then there is $\Box_{i} \gamma \in \Gamma^{x}_{n}$.
	 Since $\T$ is $\Mmc$-compatible,
	 we have that $\Box_{i}\gamma \in \Phi^{\pi_{n}(x)}_{\pi(n)}$, 
	 meaning that $\llbracket \gamma \rrbracket^{\Mmc}_{d} \in \Nmc_{i}(\pi(n))$.
	By the $T$-condition, 
	$\pi(n) \in \llbracket \gamma \rrbracket^{\Mmc}_{d}$,
	that is $\gamma \in \Phi^{\pi_{n}(x)}_{\pi(n)}$.
	Then the expansion $\T'$ of $\T$ with $n: \gamma$, obtained by the application of $\mathsf{R}_{\mathit{L}\mathbf{T}}$,
	 is $\Mmc$-compatible.\qedhere
\end{enumerate}
}}
By the argument of Theorem~\ref{thm:termination},
it can be seen that
after finitely many steps we obtain a complete completion set $\Tmc'$. Moreover, $\Tmc'$ is $\Mmc$-compatible, hence clearly clash-free.
\end{proof}

By Theorem~\ref{thm:termination}, we have that the $\LnALC$ tableau algorithm terminates after exponentially many steps in the size of the input formula. By Theorems~\ref{thm:soundness} and~\ref{thm:completeness}, the non-deterministic decision procedure based on the $\LnALC$ tableau algorithm is sound and complete with respect to satisfiability in varying domain neighbourhood models. 
Thus, we obtain the following result.

\begin{theorem}
	\label{thm:upperbound}
	Satisfiability in $\LnALC$  on varying domain neighbourhood models is decidable in $\NExpTime$.
\end{theorem}

To conclude, let $\T_{\p} = \{0 : \p
\}$ be the initial completion set for $\p$.
Define $\pi(0) = w_{\p}$ (where $\Mmc, w_{\p} \models \p$) and $\pi_{0}(x) = d$, for an arbitrary $d \in \Delta_{w_{\p}}$.
Clearly, these functions ensure that $\T_{\p}$ is $\Mmc$-compatible.
By Claim~\ref{cla:compatible}, we can apply the $\LnALC$-rules so that the obtained completion sets are $\Mmc$-compatible as well.
From Theorem~\ref{thm:termination}, we have that the $\LnALC$ tableau algorithm eventually terminates, 
providing an $\LnALC$-complete completion set for $\p$ that is clash-free by construction.
\end{proof}

\Fmp*
\begin{proof}
By Theorem \ref{thm:completeness}, if $\p$ is $\LnALC$ satisfiable, then 
there is a $\LnALC$-complete and clash-free completion set $\T$ for it.
Then by Theorem \ref{thm:soundness},
there exists a model 
$\Mmc = (\Wmc, \{ \Nmc_{i} \}_{i \in J}, \Imc)$
for $\p$ where $\Wmc =  \mathsf{L}_{\T}$
and for each $n\in\Wmc$, $\Delta_{n} = \{ x \in \mathsf{N_{V}} \mid x \ \text{occurs in} \ S_{n} \}$.
By Theorem~\ref{thm:termination}, Claim~\ref{cla:termglobal}, it follows
$|\Wmc| \leq p(|\fg(\p)|)$,
if $\mathbf{C}\notin\Lvar$, 
and $|\Wmc| \leq 2^{q(|\fg(\p)|)}$, 
if $\mathbf{C}\in\Lvar$,
where $p$ and $q$ are polynomial functions.
Finally by Theorem~\ref{thm:termination}, Claim~\ref{cla:termlocal}, 
for each $n\in\Wmc$, $|\Delta_n|$ does not exceed $2^{r(|\fg(\p)|)}$,
where $r$ is a polynomial function.
\end{proof}

\section{Proofs for Section~\ref{sec:fragvardom}}

\LemmapropL*
\begin{proof}
If an
$\MLnALCg$
formula $\varphi$ is $\LnALCg$ satisfiable
on varying domain neighbourhood models
then, clearly,
$\prop{\varphi}$ is satisfied in a $\setsymbols_\varphi$-consistent $L^{n}$ model.  
We now argue about the converse direction. 
Suppose that $\prop{\varphi}$ is satisfied in a $\setsymbols_\varphi$-consistent $L^{n}$ model
$\propmodel = (\Wmc, \{ \Nmc_{i} \}_{i \in J}, \Vmc)$. 
%
As $\propmodel$ is $\setsymbols_\varphi$-consistent, we have that, for every $w\in \propdomain$,
the $\ALC$ formula
$\alcform$
is satisfied by an $\ALC$ interpretation, say $\Imc_{w} = (\Delta_{w}, \cdot^{\Imc_{w}})$.
We define the
varying domain neighbourhood model $\Mmc=(\Fmc,\Imc)$, where the $L^{n}$ frame $\Fmc = ( \W, \{ \Nmc_{i} \}_{i \in J} )$ is as above,
and where $\Imc$ is a function associating with each $w \in \Wmc$ the $\ALC$ interpretation $\Imc_{w}$.
By induction on the structure of subformulas $\psi$ of $\varphi$, it can be shown that,
for every $w \in \Wmc$, we have
$\propmodel, w \models \prop{\psi}$
iff
 $\Mmc, w \models \psi$. We show this in Claim~\ref{cl:ind}.
\begin{claim}\label{cl:ind}
For every subformula $\psi$ of $\varphi$ and every $w \in \Wmc$, we have
$\propmodel, w \models \prop{\psi}$ iff
$\Mmc, w \models \psi$.
\end{claim}
\begin{proof}
In the base case $\psi$ is an $\ALC$ atom $\pi$ in $\varphi$ and $\prop{\psi}$ is a propositional symbol $p_\pi$. By the semantics of propositional neighbourhood models,
$\propmodel, w \models \prop{\psi}$ iff $w\in\Vmc(p_\pi)$. 
For every \ALC atom $\pi$ in $\varphi$, $w\in\Vmc(p_\pi)$ iff 
$\pi$ is a conjunct of $\hat{\varphi}_{\Vmc,w}$.
As $\propmodel$ is $\setsymbols_\varphi$-consistent, we have that, for every $w\in \propdomain$,
the $\ALC$ formula
$\hat{\varphi}_{\Vmc,w}$
is satisfied by the $\ALC$ interpretation $\Imc_{w} = (\Delta_{w}, \cdot^{\Imc_{w}})$.
 Thus, $\pi$ is a conjunct of $\hat{\varphi}_{\Vmc,w}$ iff $\Imc_w\models\pi$.
 By the semantics of $\MLnALCg$  neighbourhood models,
 $\Imc_w\models\pi$ iff $\Mmc, w \models \psi$.
 Suppose that Claim~\ref{cl:ind} holds for $\psi_1,\psi_2$. 
 For the inductive step, we make the following case distinction on 
 the format of $\psi$. 
 \begin{itemize}
 	\item $\psi=\neg\psi_1$: By the semantics of  propositional neighbourhood models,
 	$\propmodel, w \models \prop{\neg{\psi_1}}$ iff $\propmodel, w \not\models \prop{{\psi_1}}$. By the inductive hypothesis, Claim~\ref{cl:ind} holds for $\psi_1$.
 	By the contrapositive in each direction, $\propmodel, w \not\models \prop{{\psi_1}}$
 	iff $\Mmc, w \not\models \psi_1$. By the semantics of  $\MLnALCg$ neighbourhood models, $\Mmc, w \not\models \psi_1$ iff $\Mmc, w \models \neg\psi_1$.
\item $\psi=\psi_1\wedge\psi_2$: By the semantics of  propositional neighbourhood models,
$\propmodel, w \models \prop{{(\psi_1\wedge\psi_2)}}$ iff $\propmodel, w \models \prop{{\psi_1}}$ and $\propmodel, w \models \prop{{\psi_2}}$. By the inductive hypothesis, Claim~\ref{cl:ind} holds for $\psi_1,\psi_2$.
So, $\propmodel, w \models \prop{{\psi_i}}$
iff $\Mmc, w \models \psi_i$, for $i\in \{1,2\}$. By the semantics of  $\MLnALCg$ neighbourhood models, $\Mmc, w \models \psi_1$ and $\Mmc, w \models \psi_2$ iff $\Mmc, w \models \psi_1\wedge \psi_2$.
\item $\psi=\B_{i} \psi_1$: By the semantics of  propositional neighbourhood models,
$\propmodel, w \models \prop{{(\B_{i} \psi_1)}}$ iff $\llbracket \prop{{\psi_1}}\rrbracket^{\propmodel} \in \Nmc_{i}(w)$ where
$\llbracket \prop{{\psi_1}} \rrbracket^{\propmodel} = \{ v \in \Wmc \mid \propmodel, v \models \prop{{\psi_1}} \}$. By the inductive hypothesis, Claim~\ref{cl:ind} holds for $\psi_1$.
So, $\propmodel, v \models \prop{{\psi_1}}$
iff $\Mmc, v \models \psi_1$, for every $v\in\Wmc$. 
Thus, $\llbracket \prop{{\psi_1}} \rrbracket^{\propmodel}=\llbracket {{\psi_1}} \rrbracket^{\Mmc}$. By definition of $\propmodel$ and \Mmc, we have that $\Nmc_{i}(w)$
is the same in both $\propmodel$ and \Mmc, for every $w\in\Wmc$ and $i\in J$.
So $\llbracket \prop{{\psi_1}}\rrbracket^{\propmodel} \in \Nmc_{i}(w)$
iff $\llbracket {{\psi_1}} \rrbracket^{\Mmc} \in \Nmc_{i}(w)$.
By the semantics of  $\MLnALCg$ neighbourhood models, $\llbracket {{\psi_1}} \rrbracket^{\Mmc} \in \Nmc_{i}(w)$ iff $\Mmc, w \models \B_{i} \psi_1$.
 \end{itemize}
We have thus shown that for every subformula $\psi$ of $\varphi$ and every $w \in \Wmc$, we have
$\propmodel, w \models \prop{\psi}$ iff
$\Mmc, w \models \psi$.
\end{proof}
Since $\propmodel , v \models \prop{\varphi}$, for some $v \in \Wmc$, we conclude that $\varphi$ is $\LnALCg$ satisfiable. 
\end{proof}

\Lemmapropvardi*
\begin{proof}
{{In this proof, for any set $S\subseteq\{\mathbf{E,M,C,N,T,P,Q,D}\}$, we call $S$ model any neighbourhood model satisfying all conditions in $S$.}}
	We start with proving ($\Rightarrow$). 
	We consider the more complex case where $\mathbf{C}\in\Lvar$.
	For $\mathbf{C}\not\in\Lvar$ the proof simplifies by taking $k = 1$.
	Suppose that $\phi$ is satisfied in a world $w$ of a $\setsymbols$-consistent $\Lvar^{n}$ model 
	$\propmodel = (\propdomain, \{ \propneigh_{i} \}_{i \in J}, \propassign)$. That is, 
	$\propmodel, w\models \phi$. We define a $\setsymbols$-consistent valuation for 
	$\phi$
	by setting, for all $\psi \in {\sf sub}(\phi)$,
	$\nu(\psi)=1$ if $\propmodel, w\models \psi$ and $\nu(\psi) = 0$
	if  $\propmodel, w\not\models \psi$. 
	It is easy to check that $\nu$ is indeed a 
	$\setsymbols$-consistent valuation   (given that $\propmodel$ is a  
			$\setsymbols$-consistent $\Lvar^{n}$ model).   
	Now assume that $\B_i\psi_1, \dots, \B_i\psi_k, \B_i\chi\in{\sf sub}(\phi)$,
	$\valuation(\B_i\psi_j)=1$ for all $1\leq j \leq k$,
	and $\valuation(\B_i\chi)=0$.
	Then $\propmodel, w\models \B_i\psi_1 \land ... \land \B_i\psi_n\land\neg\B_i\chi$.
	Since $\propmodel$ is a $\mathbf{EC}$ model, this means that
	$\propmodel\not\models \psi_1\land ... \land \psi_n \leftrightarrow \chi$, that is,
	there is a worlds $u$ such that 
	$\propmodel, u\models (\bigwedge^{k}_{j=1}\psi_j\wedge\neg\chi) \vee \bigvee^{k}_{j=1} (\neg\psi_j\wedge\chi)$.
	(If $\mathbf{M}\in\Lvar$, then 
	$\propmodel\not\models \psi_1\land ... \land \psi_n \to \chi$, that is,
	there $u$ such that 
	$\propmodel, u\models (\bigwedge^{k}_{j=1}\psi_j\wedge\neg\chi)$.)
	Since $\propmodel$ is $\setsymbols$-consistent this concludes the proof.
	Now we prove that $\nu$ satifies $(\mathbf{X})$ if $\mathbf{X}\in\Lvar$, for $\mathbf{X}\in\{\mathbf{N,T,P,Q,D}\}$.
	\begin{itemize}
		\item[($\mathbf{N}$)]
		If $\nu(\B_i\psi)=0$, then $\propmodel, w\not\models \B_i\psi$.
		Since $\propmodel$ is a $\mathbf{EN}$ model, this means that $\propmodel\not\models\psi$
		(otherwise $\propmodel\models\B_i\psi$),
		that is there is $u$ such that $\propmodel, u \models \neg\psi$.
		
		\item[($\mathbf{T}$)]
		If $\nu(\B_i\psi)=1$, then $\propmodel, w\models \B_i\psi$, thus since $\propmodel$ is a $\mathbf{ET}$ model, $\propmodel\models\B_i\psi\to\psi$,
		hence $\propmodel, w \models \psi$, that is $\nu(\psi)=1$.
		
		\item[($\mathbf{P}$)]
		If $\valuation(\B_i\psi_1) = ... = \valuation(\B_i\psi_k) = 1$, 
		then $\propmodel, w \models \B_i\psi_1\land ... \land \B_i\psi_k$.
		Since $\propmodel$ is a $\mathbf{ECP}$ model, $\propmodel, w \models \B_i(\psi_1\land ... \land \psi_k)$,
		and $\propmodel\models\neg\B_i\falseprop$.
		Then $\propmodel\not\models \psi_1\land ... \land \psi_k \leftrightarrow \falseprop$,
		thus there is $u$ such that $\propmodel, u \models \psi_1\land ... \land \psi_k$.
		
		\item[($\mathbf{Q}$)]
		If $\valuation(\B_i\psi_1) = ... = \valuation(\B_i\psi_k) = 1$, 
		then $\propmodel, w \models \B_i\psi_1\land ... \land \B_i\psi_k$.
		Since $\propmodel$ is a $\mathbf{ECQ}$ model, $\propmodel, w \models \B_i(\psi_1\land ... \land \psi_k)$,
		and $\propmodel\models\neg\B_i(\trueprop)$.
		Then $\propmodel\not\models \psi_1\land ... \land \psi_k \leftrightarrow \trueprop$,
		thus there is $u$ such that $\propmodel, u \models \neg\psi_1\lor ...\lor\neg\psi_k$.
		
		\item[($\mathbf{D}$)]
		If $\valuation(\B_i\psi_j)= \valuation(\B_i\chi_\ell)=1$ for all $1\leq j \leq k$, $1\leq \ell \leq h$,
		then $\propmodel, w \models \bigwedge^{k}_{j=1}\B_i\psi_j \land \bigwedge^{h}_{\ell=1}\B_i\chi_\ell$.
		Since $\propmodel$ is a $\mathbf{ECD}$ model, 
		$\propmodel, w \models \B_i(\psi_i\land ... \land \psi_k) \land \B_i(\chi_1\land ... \land \chi_h)$, and
		$\propmodel\models\B_i\zeta\to\neg\B_i\neg\zeta$.
		Then $\propmodel\not\models \psi_1\land ... \land \psi_k \leftrightarrow \neg(\chi_1\land ... \land \chi_h)$,
		thus there is $u$ such that 
		$\propmodel, u \models (\psi_1\land ... \land \psi_k \land \chi_1\land ... \land \chi_h) \lor 
		(\neg(\psi_1\land ... \land \psi_k) \land \neg(\chi_1\land ... \land \chi_h))$.
		If $\mathbf{M}\in\Lvar$, then 
		$\propmodel\not\models \psi_1\land ... \land \psi_k \to \neg(\chi_1\land ... \land \chi_h)$,
		hence there is $u$ such that 
		$\propmodel, u \models \psi_1\land ... \land \psi_k \land \chi_1\land ... \land \chi_h$.
	\end{itemize}
	
	The proof of the converse ($\Leftarrow$) is as follows. 
	Suppose there is a $\setsymbols$-consistent valuation $\nu$ for $\phi$
	satisfying the conditions stated by the lemma.
	We construct a $\Lvar^{n}$ model $\propmodel$ and a world $w$ such that $\propmodel,w\models\phi$.
	By the condition, it follows that for all sets $\Psi$ of formulas
	$\B_i\psi_1, \dots, \B_i\psi_k$  in ${\sf sub}(\phi)$
	such that $\valuation(\B_i\psi_j)=1$ for all $1\leq j \leq k$,
	and all $\B_i\chi$ in ${\sf sub}(\phi)$ such that $\valuation(\B_i\chi)=0$, 
	there is a $\setsymbols$-consistent model \[\propmodel_{\Psi,\chi}=(\propdomain_{\Psi,\chi},
	\{ \propneigh_{{(\Psi,\chi)}_{i}} \}_{i \in J},\propassign_{\Psi,\chi})\]
	and a world 
	$w_{\Psi,\chi}\in \propdomain_{\Psi,\chi}$ such that 
	$\propmodel_{\Psi,\chi},w_{\Psi,\chi}\models(\bigwedge^{k}_{j=1}\psi_j\wedge\neg\chi) \vee \boldsymbol\vartheta$; moreover 
	if $\mathbf{X}\in\Lvar$, for $\mathbf{X}\in\{\mathbf{N,P,Q,D}\}$, the following hold: 
	\begin{itemize}
		\item[$(\mathbf{N})$]
		for all $\B_i\psi$ in ${\sf sub}(\phi)$ such that $\valuation(\B_i\psi)=0$, 
		there is a $\setsymbols$-consistent $\Lvar^{n}$ model 
		$\propmodel_{\psi}=(\propdomain_{\psi},  \{ \propneigh_{{\psi}_{i}} \}_{i \in J},\propassign_{\psi})$
		and a world 
		$w_{\psi}\in \propdomain_{\psi}$ such that 
		$\propmodel_{\psi}, w_{\psi} \models \neg\psi$;
		\item[$(\mathbf{P})$]
		for all $\Psi = \{\B_i\psi_1, \dots, \B_i\psi_k\}\subseteq{\sf sub}(\phi)$
		such that $\valuation(\B_i\psi_j)=1$ for all $1\leq j \leq k$,
		there is a $\setsymbols$-consistent $\Lvar^{n}$ model 
		$\propmodel_{\Psi}=(\propdomain_{\Psi},  \{ \propneigh_{{\Psi}_{i}} \}_{i \in J},\propassign_{\Psi})$
		and a world 
		$w_{\Psi}\in \propdomain_{\Psi}$ such that 
		$\propmodel_{\Psi}, w_{\Psi} \models \psi_1\land...\land\psi_k$;
		
		\item[$(\mathbf{Q})$]
		for all $\Psi = \{\B_i\psi_1, \dots, \B_i\psi_k\}\subseteq{\sf sub}(\phi)$
		such that $\valuation(\B_i\psi_j)=1$ for all $1\leq j \leq k$,
		there is a $\setsymbols$-consistent $\Lvar^{n}$ model 
		$\propmodel_{\Psi}=(\propdomain_{\Psi},  \{ \propneigh_{{\Psi}_{i}} \}_{i \in J},\propassign_{\Psi})$
		and a world 
		$w_{\Psi}\in \propdomain_{\Psi}$ such that 
		$\propmodel_{\Psi}, w_{\Psi} \models \neg\psi_1\lor ... \lor\neg\psi_k$;
		
		\item[$(\mathbf{D})$]
		for all $\Psi = \{\B_i\psi_1, \dots, \B_i\psi_k\}$, $\Lambda = \{\B_i\chi_1, \dots, \B_i\chi_h\}$,
		$\Psi,\Lambda\subseteq{\sf sub}(\phi)$
		such that 
		$\valuation(\B_i\psi_j)=\valuation(\B_i\chi_\ell)=1$ for all $1\leq j \leq k$,  $1\leq \ell \leq h$,
		there is a $\setsymbols$-consistent $\Lvar^{n}$ model 
		$\propmodel_{\Psi,\Lambda}=(\propdomain_{\Psi,\Lambda},  \{ \propneigh_{{(\Psi,\Lambda)}_{i}} \}_{i \in J},\propassign_{\Psi,\Lambda})$
		and a world 
		$w_{\Psi,\Lambda}\in \propdomain_{\Psi,\Lambda}$ such that 
		$\propmodel_{\Psi,\Lambda}, w_{\Psi,\Lambda} \models (\bigwedge^{k}_{j=1}\psi_j \land \bigwedge^{h}_{\ell=1}\chi_\ell) \vee \boldsymbol\eta$.
	\end{itemize}
	
	Let $\propmodel_1, ..., \propmodel_m$
	be an enumeration of all $\Lvar^{n}$ models listed above,
	where 
	$\propmodel_j = (\propdomain_j, \{ \propneigh_{j_{i}} \}_{i \in J},\propassign_j)$.
	That is, we take one model $\propmodel_{\Psi,\chi}$ 
	for each pair $(\Psi,\B_i\chi)$,
	where $\Psi = \{\B_i\psi_1, \dots, \B_i\psi_k\}\subseteq{\sf sub}(\phi)$,
	$\valuation(\B_i\psi_j)=1$ for all $1\leq j \leq k$,
	$\B_i\chi$ in ${\sf sub}(\phi)$, and
	$\valuation(\B_i\chi)=0$;
	and similarly 
	we take one model $\propmodel_\psi$, $\propmodel_\Psi$, or $\propmodel_{\Psi,\Lambda}$
	for all formulas or sets of formulas 
	listed in items $(\mathbf{N})$, $(\mathbf{P})$, $(\mathbf{Q})$, $(\mathbf{D})$. 
	Assume without loss of generality that 
	$\propdomain_j\cap \propdomain_\ell=\emptyset$ 
	for $j\neq \ell$. 
	We define a $\setsymbols$-consistent $\Lvar^{n}$ model   
	$\propmodel = (\propdomain,\{ \propneigh_{i} \}_{i \in J}, \propassign)$ for $\phi$
	as follows.
	\begin{itemize}
		\item $\propdomain = \bigcup_{j = 1}^{m} \propdomain_j \cup \{w\}$, where $w$ is a new world.
		
		\item 
		Consider a function $\ext{\cdot}: {\sf sub}(\phi)\rightarrow \Pmc(\Wmc)$
		with $\ext{\psi}=\bigcup_{j = 1}^{m} \llbracket \psi \rrbracket^{\propmodel_j} \cup \ext{\psi}_0$ for all $\psi\in {\sf sub}(\phi)$, where 
		$\ext{\cdot}_0: {\sf sub}(\varphi)\rightarrow  \Pmc(\{w\})$ is the function
		that assigns $\psi$ to $\{w\}$, if $\nu(\psi)=1$, 
		and to $\emptyset$, otherwise.
		By construction, we have that $\ext{\neg \psi}=\propdomain\setminus \ext{\psi}$
		and $\ext{\psi_1\wedge \psi_2} =\ext{\psi_1}\cap \ext{\psi_2}$. 
		We define the assignment $\propassign$ as the function 
		$\propassign: \NPr(\varphi)\rightarrow \Pmc(\Wmc)$ satisfying 
		$\propassign(p_\elaxiom)=\ext{p_\elaxiom}$ for all $p_\elaxiom\in \NPr(\varphi)$. 
		
		\item It remains to define $\propneigh_i$, for $1 \leq i \leq n$.
		For $u\in \propdomain_j$,
		we define $\alpha\in\propneigh_i(u)$ if and only if 
		there is $\B_i\psi$ in ${\sf sub}(\phi)$ such that
		$\propmodel_j, u \models \B_i\psi$ and $\ext{\psi} = \alpha$;
		and
		we define $\alpha\in\propneigh_i(w)$ if and only if 
		there is $\B_i\psi$ in ${\sf sub}(\phi)$ such that
		$\valuation(\B_i\psi)=1$ and $\ext{\psi} = \alpha$.
		Then if $\mathbf{C}\in\Lvar$,
		we close $\propneigh_i$ under intersection, 
		if $\mathbf{M}\in\Lvar$,
		we close $\propneigh_i$ under supersets,
		and if $\mathbf{N}\in\Lvar$,
		we extend $\propneigh_i(u)$ with $\propdomain$ for all
		$u\in\propdomain$,
		so that $\propmodel$ is a $\EC$, respectively a $\EM$,
		respectively a $\EN$, model.
	\end{itemize}
	
	We prove the following claim which ensures that
	$\propneigh_i$ is well-defined.
	
	\begin{claim}
		(i) For $u\in\propdomain_j$, if $\beta \in \Nmc_i(u)$ and 
		$\beta = \ext{\chi}$ for some 
		$\B_i\chi$ in ${\sf sub}(\phi)$,
		then $\propmodel_j,u\models\B_i\chi$.
		(ii) If $\beta \in \Nmc_i(w)$ and  $\beta = \ext{\chi}$ for some 
		$\B_i\chi$ in ${\sf sub}(\phi)$,
		then $\valuation(\B_i\chi) = 1$.
	\end{claim}
	\begin{proof}[Proof of Claim]
		We consider the case where $\mathbf{C},\mathbf{N}\in\Lvar$ and $\mathbf{M}\notin\Lvar$,
		for the other cases the proof can be easily adapted.
		
		(i) If $\beta \in \Nmc_i(u)$, then by definition 
		$\beta=\propdomain$, or
		$\beta = \bigcap_{\ell=1}^k \ext{\chi_\ell}$ for some $\B_i\chi_1, ..., \B_i\chi_k$ in ${\sf sub}(\phi)$ such that
		$\propmodel_j,u\models\bigwedge_{\ell=1}^k\B_i\chi_\ell$.
		If $\beta=\propdomain$, then $\ext{\chi}=\propdomain$,
		thus in particular 
		$\llbracket \chi \rrbracket^{\propmodel_j} = \propdomain_j$,
		and since $\propdomain_j\in\propneigh_{j_{i}}(u)$
		it holds
		$\propmodel_j, u \models \B_i\chi$.
		Otherwise
		$\ext{\chi} = \bigcap_{\ell=1}^k \ext{\chi_\ell}$,
		which implies 
		$\llbracket \chi \rrbracket^{\propmodel_j} =
		\bigcap_{\ell=1}^k\llbracket \chi_\ell \rrbracket^{\propmodel_j}$
		(because $\propdomain_j \cap \propdomain_k = \emptyset$ for  
		$k \neq j$).
		Since $\propmodel_j$ is a $\mathbf{EC}$ model,
		$\propmodel_j,u\models\B_i\bigwedge_{\ell=1}^k\chi_\ell$,
		then 
		$\llbracket \bigwedge_{\ell=1}^k\chi_\ell \rrbracket^{\propmodel_j} 
		\in\propneigh_{j_{i}}(u)$,
		where
		$\llbracket \bigwedge_{\ell=1}^k\chi_\ell \rrbracket^{\propmodel_j} =
		\bigcap_{\ell=1}^k\llbracket \chi_\ell \rrbracket^{\propmodel_j} =
		\llbracket \chi \rrbracket^{\propmodel_j}$,
		therefore
		$\propmodel_j,u\models\B_i\chi$.
		
		(ii) If $\beta \in \Nmc_i(u)$, then by definition 
		$\beta=\propdomain$, or
		$\beta = \bigcap_{\ell=1}^k \ext{\chi_\ell}$ for some $\B_i\chi_1, ..., \B_i\chi_k$ in ${\sf sub}(\phi)$ such that
		$\valuation(\chi_\ell) = 1$ for all $1 \leq \ell \leq k$.
		If $\beta=\propdomain$, then $\ext{\chi} = \propdomain$,
		thus $\ext{\chi}_0 = \{w\}$, that is $\valuation(\chi) = 1$.
		By contradiction, suppose that 
		$\valuation(\B_i\chi) = 0$.
		Then by $(N)$, there are a $\setsymbols$-consistent $\Lvar^{n}$ model 
		$\propmodel_{\chi}$  and a world  $w_{\chi}$ such that 
		$\propmodel_{\chi}, w_{\chi} \not\models \chi$.
		One such model is enumerated among 
		$\propmodel_1, ..., \propmodel_m$, let it be $\propmodel_o$.
		Then $\llbracket \chi \rrbracket^{\propmodel_o} \not= \propdomain_o$, 
		thus $\ext{\chi} \not= \propdomain$, giving a contradiction.
		Therefore $\valuation(\B_i\chi) = 1$.
		If instead $\beta\not=\propdomain$,
		then $\ext{\chi} = \bigcap_{\ell=1}^k \ext{\chi_\ell}$.
		Suppose that $\valuation(\chi) = 0$.
		By the hypothesis of the lemma,
		there are a $\setsymbols$-consistent model $\propmodel_{\Lambda,\chi}$ and a world 
		$w_{\Lambda,\chi}$ such that 
		$\propmodel_{\Lambda,\chi},w_{\Lambda,\chi}\models(\bigwedge^{k}_{\ell=1}\chi_\ell\wedge\neg\chi) \vee \bigvee^{k}_{\ell=1} (\neg\chi_\ell\wedge\chi)$,
		where $\Lambda = \{\B_i\chi_1,...,\B_i\chi_k\}$.
		Then one such model is enumerated among 
		$\propmodel_1, ..., \propmodel_m$, let it be $\propmodel_o$.
		Thus 
		$\llbracket \chi \rrbracket^{\propmodel_o} \not=
		\llbracket \bigwedge_{\ell=1}^k\chi_\ell \rrbracket^{\propmodel_o}$, where
		$\llbracket \bigwedge_{\ell=1}^k\chi_\ell \rrbracket^{\propmodel_\ell} =
		\bigcap_{\ell=1}^k\llbracket \chi_\ell \rrbracket^{\propmodel_\ell}$.
		Since  $\propdomain_j \cap \propdomain_k = \emptyset$ for  
		$k \neq j$,  
		this implies $\ext{\chi} \not = \bigcap_{\ell=1}^k \ext{\chi_\ell}$,
		giving a contradiction.
		Therefore $\valuation(\B_i\chi) = 1$. 
	\end{proof}
	
	\begin{claim}
		For all $\psi\in{\sf sub}(\phi)$,
		$\llbracket \psi \rrbracket^{\propmodel} = \ext{\psi}$.
	\end{claim}
	\begin{proof}[Proof of Claim]
		By induction on the structure of formulas.
		For $p_\elaxiom\in \NPr(\varphi)$,
		$\llbracket p_\elaxiom \rrbracket^{\propmodel}=\ext{p_\elaxiom}$
		by definition of $\propassign$.
		For boolean connectives, the claim follows immediately from the inductive hypothesis and the fact
		that $\ext{\neg \chi} =\propdomain\setminus \ext{\chi}$ and $\ext{\chi_1\wedge \chi_2} = \ext{\chi_1}\cap \ext{\chi_2}$.
		Suppose  that $u\in \ext{\B_i\chi}$.  
		Then, either $u=w$ and $\nu(\B_i\chi)=1$
		or $u\in \propdomain_j$ and $\propmodel_j,u\models\B_i\chi$. 
		By definition of $\propneigh_i$, in either case 
		we have that $\ext{\chi}\in \propneigh_i(u)$. By inductive hypothesis,
		$\llbracket \chi \rrbracket^{\propmodel} = \ext{\chi}$, it follows that $\propmodel,u\models \B_i\chi$, 
		that is, $u\in \llbracket \B_i\chi \rrbracket^{\propmodel}$. 
		Suppose now that $u\in \llbracket \B_i\chi \rrbracket^{\propmodel}$, 
		that is, $\propmodel,u\models \B_i\chi$, or, equivalently,  $\llbracket \chi \rrbracket^{\propmodel}\in \propneigh_i(u)$.
		By inductive hypothesis,
		$\llbracket \chi \rrbracket^{\propmodel} = \ext{\chi}$, then
		by the previous claim, if $u=w$, then $\valuation(\B_i\chi)=1$,
		and if $u\in \propdomain_j$, then $\propmodel_j,u\models\B_i\chi$. 
		By definition of $\ext{\cdot}$, in either case we have that $u\in \ext{\B_i\psi}$. 
	\end{proof}

	\begin{claim}
		For $\mathbf{X}\in\{\mathbf{M,C,N,T,P,Q,D}\}$,
		if $\mathbf{X}\in\Lvar$,
		then $\propmodel$ satisfies the $\mathbf{X}$-condition.
	\end{claim}
	\begin{proof}[Proof of Claim]
		For $\mathbf{X}\in\{\mathbf{M,C,N}\}$, 
		the claim follows immediately from
		the definition of $\propneigh_i$.
		We consider the other cases,
		assuming 
		that 
		$\mathbf{C}\in\Lvar$ and $\mathbf{M}\notin\Lvar$
		(for $\mathbf{C}\notin\Lvar$ or $\mathbf{M}\in\Lvar$ the proof can be easily adapted).
		
		\begin{itemize}
			\item[$(\mathbf{T})$]  Suppose that $\alpha\in\propneigh_i(u)$.
			Then $\alpha = \propdomain$ (if $\mathbf{N}\in\Lvar$), 
			which implies $u\in\alpha$,
			or $\alpha = \bigcap_{\ell = 1}^k \ext{\psi_\ell}$ for $\B_i\psi_1, ..., \B_k\psi \in{\sf sub}(\phi)$.
			If $u\in\propdomain_j$, then $\propmodel_j, u \models \B_i\psi_1 \land ... \land \B_i\psi_k$. 
			Since $\propmodel_j$ is a $\mathbf{ET}$ model,
			$\propmodel_j\models \B_i\chi\to\chi$, thus $\propmodel_j, u \models \psi_1 \land ... \land \psi_k$,
			that is $u\in\llbracket \psi_\ell \rrbracket^{\propmodel_j}$ for all $1\leq \ell \leq k$.
			It follows $u\in \bigcap_{\ell = 1}^k \ext{\psi_\ell} = \alpha$. 
			If instead $u = w$, then $\valuation(\B_i\psi_\ell) = 1$  for all $1\leq \ell \leq k$.
			By the hypothesis of the proposition, $\valuation(\psi) = 1$ for all $1\leq \ell \leq k$, 
			thus $\ext{\psi_\ell}_0 = \{w\}$  for all $1\leq \ell \leq k$.
			therefore $u = w \in \bigcap_{\ell = 1}^k \ext{\psi_\ell} = \alpha$. 
			
			
			\item[$(\mathbf{P})$]  Assume that $\emptyset\in\propneigh_i(u)$.
			Then $\emptyset = \bigcap_{\ell = 1}^k \ext{\psi_\ell}$ for $\B_i\psi_1, ..., \B_k\psi \in{\sf sub}(\phi)$.
			If $u\in\propdomain_j$, then $\propmodel_j, u \models \B_i\psi_1 \land ... \land \B_i\psi_k$,
			that is $\llbracket \psi_\ell \rrbracket^{\propmodel_j}\in\propneigh_{j_i}(u)$ for all $1\leq \ell \leq k$.
			By the $\mathbf{C}$-condition, $\bigcap_{\ell = 1}^k \llbracket \psi_\ell \rrbracket^{\propmodel_j} \in\propneigh_{j_i}(u)$,
			and by construction of $J$, 
			$\bigcap_{\ell = 1}^k \llbracket \psi_\ell \rrbracket^{\propmodel_j} = \emptyset$,
			contradicting the fact that $\propmodel_j$ is a $\mathbf{EP}$ model.
			If instead $u = w$, then $\valuation(\B_i\psi_\ell) = 1$  for all $1\leq \ell \leq k$.
			By item $(\mathbf{P})$ above, there are  $\propmodel_{\Psi}$ and  $w_{\Psi}$ such that 
			$\propmodel_{\Psi}, w_{\Psi} \models \psi_1\land...\land\psi_k$.
			One such model is enumerated among 
			$\propmodel_1, ..., \propmodel_m$, let it be $\propmodel_o$.
			Then $\llbracket \bigwedge_{\ell=1}^k\psi_\ell \rrbracket^{\propmodel_o} =
			\bigcap_{\ell = 1}^k \llbracket \psi_\ell \rrbracket^{\propmodel_o} \not=\emptyset$,
			thus $\bigcap_{\ell = 1}^k \ext{\psi_\ell}  \not=\emptyset$. 
			In either case $\emptyset\notin\propneigh_i(u)$.
			
			\item[$(\mathbf{Q})$]  Suppose that $\alpha\in\propneigh_i(u)$.
			Then $\alpha = \bigcap_{\ell = 1}^k \ext{\psi_\ell}$ for $\B_i\psi_1, ..., \B_k\psi \in{\sf sub}(\phi)$.
			If $u\in\propdomain_j$, then $\propmodel_j, u \models \B_i\psi_1 \land ... \land \B_i\psi_k$,
			that is $\llbracket \psi_\ell \rrbracket^{\propmodel_j}\in\propneigh_{j_i}(u)$ for all $1\leq \ell \leq k$.
			Since $\propmodel_j$ is a $\mathbf{EQ}$ model, $\llbracket \psi_\ell \rrbracket^{\propmodel_j} \not = \propdomain_j$,
			then $\bigcap_{\ell = 1}^k \ext{\psi_\ell}\not=\propdomain$.
			If instead $u = w$, then $\valuation(\B_i\psi_\ell) = 1$  for all $1\leq \ell \leq k$.
			By item $(\mathbf{Q})$ above, there are  $\propmodel_{\Psi}$ and  $w_{\Psi}$ such that 
			$\propmodel_{\Psi}, w_{\Psi} \models \neg\psi_1\lor ... \lor\neg\psi_k$.
			One such model is enumerated among 
			$\propmodel_1, ..., \propmodel_m$, let it be $\propmodel_o$.
			Then $\llbracket \neg\psi_\ell \rrbracket^{\propmodel_o} \not=\emptyset$ for some $\B_i\psi_\ell$,
			that is $\llbracket \psi_\ell \rrbracket^{\propmodel_o} \not=\propdomain$,
			therefore $\bigcap_{\ell = 1}^k \llbracket \psi_\ell \rrbracket^{\propmodel} \not=\propdomain$.
			In either case $\alpha\not=\propdomain$, that is $\propdomain\notin\propneigh_i(u)$.
			
			\item[$(\mathbf{D})$]  Suppose that $\alpha,\beta\in\propneigh_i(u)$.
			Then $\alpha = \bigcap_{\ell = 1}^k \ext{\psi_\ell}$ for $\B_i\psi_1, ..., \B_i\psi_k \in{\sf sub}(\phi)$, and
			$\beta = \bigcap_{{\ell'} = 1}^k \ext{\chi_{\ell'}}$ for $\B_i\chi_1, ..., \B_i\chi_h \in{\sf sub}(\phi)$.
			If $u\in\propdomain_j$, then 
			$\propmodel_j, u \models \B_i\psi_1 \land ... \land \B_i\psi_k \land \B_i\chi_1 \land ... \land \B_i\chi_h$.
			Since $\propmodel_j$ is a $\mathbf{EC}$ model,
			$\propmodel_j, u\models \B_i \bigwedge_{\ell=1}^k\psi_\ell \land \B_i\bigwedge_{{\ell'}=1}^h\chi_{\ell'}$, and
			since $\propmodel_j$ is a $\mathbf{ED}$ model,
			$\propmodel_j\not\models \bigwedge_{\ell=1}^k\psi_\ell \leftrightarrow \neg\bigwedge_{{\ell'}=1}^h\chi_{\ell'}$,
			that is 
			$\bigcap_{\ell = 1}^k \llbracket \psi_\ell \rrbracket^{\propmodel_j} =
			\llbracket \bigwedge_{\ell=1}^k\psi_\ell \rrbracket^{\propmodel_j} \not=
			\llbracket \neg\bigwedge_{{\ell'}=1}^h\chi_{\ell'} \rrbracket^{\propmodel_j} =
			\propdomain_j \setminus \bigcap_{{\ell'}=1}^h \llbracket \chi_{\ell'} \rrbracket^{\propmodel_j}$.
			If instead $u = w$, 
			then  $\valuation(\B_i\psi_\ell)=\valuation(\B_i\chi_{\ell'})=1$ for all $1\leq \ell \leq k$,  $1\leq \ell' \leq h$.
			By item $(\mathbf{D})$ above, there are
			$\propmodel_{\Psi,\Lambda}$ and $w_{\Psi,\Lambda}$ such that
			$\propmodel_{\Psi,\Lambda}, w_{\Psi,\Lambda} \models (\bigwedge^{k}_{\ell=1}\psi_\ell \land \bigwedge^{h}_{\ell'=1}\chi_{\ell'}) \vee (\neg(\bigwedge^{k}_{\ell=1}\psi_\ell) \land \neg(\bigwedge^{h}_{\ell'=1}\chi_{\ell'}))$.
			One such model is enumerated among 
			$\propmodel_1, ..., \propmodel_m$, let it be $\propmodel_o$.
			Then $\bigcap_{\ell = 1}^k \llbracket \psi_\ell \rrbracket^{\propmodel_o} =
			\llbracket \bigwedge_{\ell=1}^k\psi_\ell \rrbracket^{\propmodel_o} \not=
			\llbracket \neg\bigwedge_{{\ell'}=1}^h\chi_{\ell'} \rrbracket^{\propmodel_o} =
			\propdomain_o \setminus \bigcap_{{\ell'}=1}^h \llbracket \chi_{\ell'} \rrbracket^{\propmodel_o}$.
			Thus in either case, 
			$\bigcap_{\ell = 1}^k \ext{\psi_\ell} \not=\propdomain \setminus \bigcap_{{\ell'} = 1}^k \ext{\chi_{\ell'}}$,
			that is, $\alpha\not=\propdomain\setminus\beta$.\qedhere
		\end{itemize} 
	\end{proof}

	\begin{claim}
		$\propmodel$ is $\setsymbols$-consistent.
	\end{claim}
	\begin{proof}[Proof of Claim]
		$\nu$, used to construct the assignment 
		related to $w$, is $\setsymbols$-consistent and 
		the models $\propmodel_1,\ldots,\propmodel_m$, used to define 
		the remaining worlds in $\Wmc$, are all $\setsymbols$-consistent. 
	\end{proof}
	
	Finally, since $\valuation(\phi)=1$, we have that $w\in \ext{\phi}$, 
	and consequently $\propmodel, w \models \phi$.
	Given that $\propmodel$ is a $\setsymbols$-consistent $\Lvar^{n}$ model, this concludes the proof.
\end{proof}

\begin{algorithm}[t]
	\KwIn{$L$,  $\setsymbols$, and an $\MLn$ formula $\phi$ built from $\setsymbols$.}
	\KwOut{$\mathsf{satisfiable}$, if $\psi$ is  satisfiable in a $\setsymbols$-consistent $L^n$ model; $\mathsf{unsatisfiable}$, otherwise.}
	\BlankLine
	
	\For{each $\setsymbols$-consistent valuation $\valuation$ for $\phi$}{
		\uIf{$\mathsf{Check}(L,\setsymbols,\valuation,\phi)=1$}{
			\uIf{$L\cap\{ \mathbf{N},\mathbf{T},\mathbf{P},\mathbf{Q}\} \neq \emptyset$}{
				\uIf{$\mathsf{CheckNTPQ}(L,\setsymbols,\valuation,\phi)=1$}{
					\uIf{$\mathbf{D}\notin L$}{
						\Return $\mathsf{satisfiable}$\;
					}\uElseIf{$\mathsf{CheckD}(L,\setsymbols,\valuation,\phi){=}1$}
						{\Return $\mathsf{satisfiable}$\;}
				}
			}\uElseIf{$\mathbf{D}\in L$, $\mathsf{CheckD}(L,\setsymbols,\valuation,\phi){=}1$
		}{\Return $\mathsf{satisfiable}$\;} 
		}
		
	}
\BlankLine
\Return $\mathsf{unsatisfiable}$\;
\caption{$\mathsf{Sat}$}
\label{alg:propSAT}
\end{algorithm}

\begin{algorithm}[t]
	\KwIn{$L$, $\setsymbols$, a $\setsymbols$-consistent valuation $\valuation$, and an $\MLn$ formula $\phi$ built from $\setsymbols$.}
	\KwOut{$\mathsf{1}$, if $\valuation$ satisfies the conditions of Lemma~\ref{lem:proplemmaL}; $0$, otherwise.}
	\BlankLine
	\uIf{$\mathbf{C} \in \Lvar$}{
		$\boldsymbol{\kappa}:= | {\sf sub}({\phi}) |$\;
	}
	\uElse{$\boldsymbol{\kappa}:=1$}
	\BlankLine
		\For{all $1\leq k\leq \boldsymbol{\kappa}$}{
			\For{ $\B_i\psi_1, \dots, \B_i\psi_k, \B_i\chi\in{\sf sub}(\phi)$,
				with $\valuation(\B_i\psi_j)=1$ for all $1\leq j \leq k$,
				and $\valuation(\B_i\chi)=0$}{ 
				\uIf{$\mathbf{M}\in L$}{ 
					\uIf{$\mathsf{Sat}(L,\setsymbols,\bigwedge^{k}_{j=1}\psi_j\wedge\neg\chi)= \mathsf{unsatisfiable}$}{\Return $0$\;} 
				}
				\uElseIf{ $\mathsf{Sat}(L,\setsymbols,(\bigwedge^{k}_{j=1}\psi_j\wedge\neg\chi)\vee (\bigvee^{k}_{j=1} (\neg\psi_j\wedge\chi)))=\mathsf{unsatisfiable}$\;}
				{\Return $0$\;}
			}		
		}
		\Return $1$\;
 		\BlankLine
		\caption{$\mathsf{Check}$}
	\label{alg:prop1}
\end{algorithm}

\begin{algorithm}[t]
	\KwIn{$L$, $\setsymbols$, a $\setsymbols$-consistent valuation $\valuation$, and an $\MLn$ formula $\phi$ built from $\setsymbols$.}
	\KwOut{$\mathsf{1}$, if $\valuation$ satisfies the conditions of  Lemma~\ref{lem:proplemmaL}; $0$, otherwise.}
	\BlankLine
	\uIf{$\mathbf{C} \in \Lvar$}{
		$\boldsymbol{\kappa}:= | {\sf sub}({\phi}) |$\;
	}
	\uElse{$\boldsymbol{\kappa}:=1$}
	\BlankLine
		 
		\For{all $1\leq k\leq \boldsymbol{\kappa}$}{
				
				\uIf{$\mathbf{N}\in L$}{	
					\For{ $\B_i\psi\in{\sf sub}(\phi)$  
						with $\valuation(\B_i\psi)=0$}{	
						\uIf{$\mathsf{Sat}(L,\setsymbols,\neg \psi)= 	\mathsf{unsatisfiable}$}{\Return $0$\;} 
					}
					\Return $1$\;	
				}
				
			\uIf{$\mathbf{T}\in L$}{
				\For{$\B_i\psi\in{\sf sub}(\phi)$ with  
					$\valuation(\B_i\psi)=1$}{
					\uIf{$\valuation(\psi)=0$}{
						\Return $0$\;}
				}
			}
			\uIf{$\mathbf{P}\in L$}{
				\For{$\B_i\psi_1, \dots, \B_i\psi_k\in{\sf sub}(\phi)$ with
					$\valuation(\B_i\psi_j)=1$ for all $1\leq j \leq k$}{
					\uIf{$\mathsf{Sat}(L,\setsymbols,\bigwedge^{k}_{j=1}\psi_j)= \mathsf{unsatisfiable}$}{\Return $0$\;} 
				}
			}
			\uIf{$\mathbf{Q}\in L$}{
				\For{$\B_i\psi_1, \dots, \B_i\psi_k\in{\sf sub}(\phi)$ with 
					$\valuation(\B_i\psi_j)=1$ for all $1\leq j \leq k$}{
					\uIf{$\mathsf{Sat}(L,\setsymbols,\bigvee^{k}_{j=1}\neg\psi_j)= \mathsf{unsatisfiable}$}{\Return $0$\;} 
				}
			}
		}
			\Return $1$\;

		\BlankLine
		\caption{$\mathsf{CheckNTPQ}$}
	\label{alg:prop}
\end{algorithm}

\begin{algorithm}[t]
	\KwIn{$L$, $\setsymbols$, a $\setsymbols$-consistent valuation $\valuation$, and an $\MLn$ formula $\phi$ built from $\setsymbols$.}
	\KwOut{$\mathsf{1}$, if $\valuation$ satisfies the conditions of Lemma~\ref{lem:proplemmaL}; $0$, otherwise.}
	\BlankLine
	\uIf{$\mathbf{C} \in \Lvar$}{
$\boldsymbol{\kappa}:= | {\sf sub}({\phi}) |$\;
	}
 \uElse{$\boldsymbol{\kappa}:=1$}
 \BlankLine

		\For{all $1\leq k,h\leq \boldsymbol{\kappa}$}{

		\For{$\B_i\psi_1, \dots, \B_i\psi_k, \B_i\chi_1, \dots, \B_i\chi_h\in{\sf sub}(\phi)$ with
			$\valuation(\B_i\psi_j)=1$, for all $1\leq j \leq k$, and
			$\valuation(\B_i\chi_\ell)=1$, for all $1\leq \ell \leq h$}{
			\uIf{$\mathbf{M}\in L$}{
				\uIf{$\mathsf{Sat}(L,\setsymbols,(\bigwedge^{k}_{j=1}\psi_j \land \bigwedge^{h}_{\ell=1}\chi_\ell))= \mathsf{unsatisfiable}$}{\Return $0$\;} 
			}
			\uElseIf  {$\mathsf{Sat}(L,\setsymbols,(\bigwedge^{k}_{j=1}\psi_j \land \bigwedge^{h}_{\ell=1}\chi_\ell)\vee (\neg(\bigwedge^{k}_{j=1}\psi_j) \land \neg(\bigwedge^{h}_{\ell=1}\chi_\ell)))= \mathsf{unsatisfiable}$}{\Return $0$\;}
			
		}
	}
	\Return $1$\;
	\BlankLine

	\caption{$\mathsf{CheckD}$}
	\label{alg:propD}
\end{algorithm}

\Satfragvardomexp*
\begin{proof}
Soundness and completeness of Algorithm~\ref{alg:propSAT} is given by Lemmas~\ref{lem:propL} and~\ref{lem:proplemmaL}. 
 We argue   that Algorithm~\ref{alg:prop} terminates in exponential time. 
 Since the \ALC satisfiability check is in exponential time, one can compute
 in exponential time (in the size of $\setsymbols$) all valuations $\valuation$ which are $\setsymbols$-consistent. The number of iterations in Line 1 of  Algorithm~\ref{alg:propSAT}
 is bounded by $2^{|\setsymbols|}$. It remains to argue that each iteration takes exponential time. Suppose $\prop{\varphi}$ is the original formula we want to check satisfiability.
 Since each iteration calls   the functions $\mathsf{Check}$, $\mathsf{CheckNTPQ}$, 
 $\mathsf{CheckD}$ and these functions can make recursive calls (to $\mathsf{Sat}$), we need to argue that (1) the number of recursive calls in exponentially bounded and (2) 
 the number of steps inside each function is also exponentially bounded.
 Regarding the latter, we argue that   the number of iterations of the ``for'' loops inside
 $\mathsf{Check}$, $\mathsf{CheckNTPQ}$, and $\mathsf{CheckD}$ is exponentially bounded by the number of subformulas of the formula given as input to each function (and each such formula has size linear in the size of $\prop{\varphi}$). If  the total number of recursive calls is exponentially bounded then Point (2) holds.
 So it remains to argue about Point (1). 
  Consider a computation tree where each node corresponds to a recursive call to  $\mathsf{Sat}$ and the parent relation in the tree is defined by the recursive calls.
 Since each recursive call reduces the number of nested epistemic 
 operators of the original formula $\prop{\varphi}$, any nested sequence  of   recursive calls
 is polynomial in the size of $\prop{\varphi}$. This means that the depth of such tree is polynomial in the size of $\prop{\varphi}$ and, since the number of children of each node is  exponentially bounded (see Point(2)), the total number of nodes of the tree is exponential in the size of $\prop{\varphi}$. We have thus shown that the number of recursive calls in exponentially bounded.
 As satisfiability in $\ALC$ is $\ExpTime$-hard, our upper bound is tight.
\end{proof}

\section{Proofs for Section~\ref{sec:reasoncondom}}


\Theoremcomplealc*
\begin{proof}
This theorem is a consequence of the following claim, and 
the complexity of formula satisfiability in \KnALC{3n} constant domain relational models~\cite{KraWol,GasHer,GabEtAl03}.
\begin{claim} 
The $\EnALC{n}$ formula satisfiability problem on constant domain neighbourhood models can be reduced in polynomial time to the \KnALC{3n} formula satisfiability problem on constant domain relational models.
\end{claim}
%
\begin{proof}[Proof of Claim]
Consider an \MLALC{n} formula $\p$
s.t.
$\Mmc, w \mdl \p$, for some constant domain neighbourhood model $\Mmc = (\Fmc, \Delta, \Int)$ with $\Fmc = (\Wmc, \{ \Nmc_{i} \}_{i \in J})$
and some $w\in\Wmc$.
We define a relational frame
$\Fmf = (W, \{ R_{i_{1}}, R_{i_{2}}, R_{i_{3}} \}_{i \in J})$
and an $\MLALC{3n}$ relational model 
$\Mmf = (\Fmf, \Delta, I)$
such that:
	\begin{itemize}
		\item $W = \{ (w, 0) \mid w \in \Wmc \} \cup \{ (\alpha, 1) \mid \alpha \in \bigcup_{v \in \Wmc} \Nmc_{i}(v) \}$
		\item $\relations_{i_1} = \{ ((w, 0), (\alpha, 1)) \mid \alpha \in \Nmc_{i}(w)\}$;
		\item $\relations_{i_2} = \{ ((\alpha, 1), (w, 0)) \mid w \in \alpha \}$
		\item $\relations_{i_3} = \{  ((\alpha, 1), (w, 0)) \mid w \not \in \alpha \}$ 
		\item for every $(w, 0) \in W$, $I_{(w, 0)} = \Imc_{w}$; for every $(\alpha, 1) \in W$, $X^{I_{(\alpha, 1)}} = \eset$, for all $X \in \NC \cup \NR$, and $a^{I_{(\alpha, 1)}} = a^{\Int}$, for all $a \in \NI$.
	\end{itemize}
The pairs $(w, 0), (\alpha, 1)$ are used to ensure that $W$ is the disjoint union of the sets of worlds $w$ and subsets $\alpha$ of $\Wmc$.


Firstly we show, by induction on the structure of concepts $C$, that for all $d \in \Delta$ and all $w \in \Wmc$:
\[
d \in C^{\Int_w} \text{ iff } d \in (C\tr)^{I_{(w, 0)}}.
\]
For the base case $C = A \in \NC$, 
the claim follows immediately from the definitions of $I$ and $\cdot\tr$. 
Assume the claim holds for $D$ and $E$. 
The inductive cases $C = \lnot D$ and $C = (D \sqcap E)$ are straightforward. 
We are left with the cases below.

$C = \exists r.D$.
We have that
$d \in (\exists r.D)^{\Imc_{w}}$
iff there is $d' \in D^{\Imc_{w}}$ such that $(d,d') \in r^{\Imc_{w}}$.
By i.h. and definition of $I$, this is equivalent to
$d' \in (D\tr)^{I_{(w, 0)}}$ and $(d,d') \in r^{I_{(w, 0)}}$,
which means that
$d \in (\exists r. (D\tr))^{I_{(w, 0)}}$. 
By definition of $\cdot\tr$, $d \in ((\exists r.D)\tr)^{I_{(w, 0)}}$.

$C = \B_{i} D$. 
By definition, we have that
$d \in (\B_{i} D)^{\Imc_{w}}$
iff 
$\llbracket D \rrbracket^{\Mmc}_{d} \in \Nmc_{i}(w)$.
Equivalently, iff there is
$\alpha \in \Nmc_{i}(w)$
s.t. for all $v \in \Wmc$,
$v \in \alpha \Leftrightarrow d \in D^{\Imc_{v}}$.
By i.h. and definitions of $\relations_{i_2}$ and $\relations_{i_3}$, this means that there is $\alpha \in \Nmc_{i}(w)$ s.t.
$(i)$ for every $v \in \Wmc : (\alpha, 1) \relations_{i_2} (v, 0) \Rightarrow d \in (D\tr)^{I_{(v, 0)}}$
and
$(ii)$ for every $v \in \Wmc: (\alpha, 1) \relations_{i_3} (v, 0) \Rightarrow d \not \in (D\tr)^{I_{(v, 0)}}$.
This holds iff there is
$\alpha \in \Nmc_{i}(w)$ s.t. $d \in (\B_{i_2} D\tr \sqcap \B_{i_3} \lnot D\tr)^{I_{(\alpha, 1)}}$. 
By definition of $\relations_{i_1}$, this means that there exists
$(\alpha, 1) \in W$ s.t. $(w, 0) \relations_{i_1} (\alpha, 1)$ and $d \in (\B_{i_2} D\tr \sqcap \B_{i_3} \lnot D\tr)^{I_{(\alpha, 1)}}$.
That is, $d \in \D_{i_1} (\B_{i_2} D\tr \sqcap \B_{i_3} \lnot D\tr)^{I_{(w, 0)}}$
iff, by definition of $\cdot\tr$,
$d \in ((\B_{i} D)\tr)^{I_{(w, 0)}}$.
\newline

We now show that for every $\MLALC{n}$ formula $\psi$ and every $w \in \Wmc$:
\[
	\Mmc, w \mdl \psi \text{ iff } \Mmf, (w, 0) \mdl \psi\tr
\]
For the case $\psi = C \sqs D$, it follows from the previous claim, while for $\psi = C(a)$ and $\psi = r(a,b)$, it is immediate from the definitions of $I$ and $\cdot\tr$, as well as from the claim above.
Assuming that the lemma holds for $\chi$ and $\zeta$, the inductive cases $\psi = \lnot \chi$ and $\psi = \chi \land \zeta$ are straightforward.
We prove the statement for modalised formulas.

$\psi = \B_{i} \chi$.
$\Mmc, w \mdl \B_{i} \chi$
iff, by definition,
$\llbracket \chi \rrbracket^{\Mmc}  \in \Nmc_{i}(w)$.
That is, iff there is $\alpha \in \Nmc_{i}(w)$ s.t. for all $v \in \Wmc : v \in \alpha \Leftrightarrow \Mmc, v \mdl \chi$.
By i.h. and definitions of $\relations_{i_2}, \relations_{i_3}$, this means that there is
$\alpha \in \Nmc_{i}(w)$ s.t.
$(i)$ for all $v \in \Wmc : (\alpha, 1) \relations_{i_2} (v, 0) \Rightarrow \Mmf, (v, 0) \mdl \chi\tr$
and
$(ii)$ for all $v \in \Wmc: (\alpha, 1) \relations_{i_3} (v, 0) \Rightarrow \Mmf, (v, 0) \not\mdl \chi\tr$,
iff there is $\alpha \in \Nmc_{i}(w)$ s.t.
$\Mmf, (\alpha, 1) \mdl \B_{i_2} \chi\tr \land \B_{i_3} \lnot \chi\tr$.
By definition of $\relations_{i_1}$, the previous step is equivalent to:
there is $(\alpha, 1) \in W$ s.t. $(w, 0) \relations_{1} (\alpha, 1)$ and $\Mmf, (\alpha, 1) \mdl \B_{i_2} \chi\tr \land \B_{i_3} \lnot \chi\tr$,
iff
$\Mmf, (w, 0) \mdl \D_{i_1} (\B_{i_2} \chi\tr \land \B_{i_3} \lnot \chi\tr)$.
By definition of $\cdot\tr$,
$\Mmf, (w, 0) \mdl (\B_{i} \chi)\tr$.

Thus, in particular, we obtain $\Mmf, (w, 0) \mdl \p\tr$.
\newline

Conversely, consider a $\MLALC{3n}$ formula $\p\tr$ s.t.
$\Mmf, w \mdl \p\tr$,
for some $\MLALC{3n}$ R-model
$\Mmf = (\Fmf, \Delta, I)$
based on
$\Fmf = (W, \{ \relations_{i_{j}} \}_{j \in [1, 3]})$,
and some
$w \in W$.
We define a $\MLnALC{n}$
neighbourhood model
$\Mmc = (\Fmc, \Delta, \Int)$
based on
$\Fmc = (\Wmc, \{ \Nmc_{i} \}_{i \in [1, n]})$
s.t.
$\Wmc = W$,
and for all $w \in W$:
\begin{itemize}
			\item $\alpha \in \Nmc_{i}(w)$ iff there is $v \in W$ s.t. $w \relations_{i_1} v$ and: $(i)$ for all $u \in W$, $v \relations_{i_2} u \Rightarrow u \in \alpha$, and $(ii)$ for all $u \in W$, $v \relations_{i_3} u \Rightarrow u \not \in \alpha$;
			\item $\Imc_{w} = I_{w}$.
\end{itemize}

Again, we show firstly, by induction on the structure of concepts $C$, that for all $d \in \Delta$ and all $w \in W$:
\[
d \in C^{\Imc_{w}} \text{ iff } d \in (C\tr)^{I_{w}}.
\]
For the base case $C = A \in \NC$, the claim follows from the definitions of $\Int$ and $\cdot\tr$. 
Assume the claim holds for $D$ and $E$. 
The inductive cases $C = \lnot D$ and $C = (D \sqcap E)$ are straightforward. 
We are left with the cases below.

$C = \exists r.D$.
We have that
$d \in (\exists r.D)^{\Imc_{w}}$
iff there is $d' \in D^{\Imc_{w}}$ such that $(d,d') \in r^{\Imc_{w}}$.
By i.h. and definition of $\Int$, this is equivalent to
$d' \in (D\tr)^{I_{w}}$ and $(d,d')\in r^{I_{w}}$, which means that $d \in (\exists r. (D\tr))^{I_{w}}$. 
By definition of $\cdot\tr$, $d \in ((\exists r.D)\tr)^{I_{w}}$.

$C = \B_{i} D$. 
We have  
$d \in (\B_{i} D)^{\Imc_{w}}$ iff 
$\llbracket D \rrbracket^{\Mmc}_{d} \in \Nmc_{i}(w)$.
By definition of $\Nmc_{i}$, this means that there is
$v \in W$ s.t. $w \relations_{i_1} v$
and:
$(i)$ for all $u \in W$: $v \relations_{i_2} u \Rightarrow d \in D^{\Imc_{u}}$
and
$(ii)$ for all $u \in W$: $v \relations_{i_3} u \Rightarrow d \not \in D^{\Imc_{u}}$.
By i.h. the previous step is equivalent to:
there is $v \in W$ s.t. $w \relations_{i_1} v$ and:
$(i)$ for all $u \in  W$: $v \relations_{i_2} u \Rightarrow d \in (D\tr)^{I(u)}$ and
$(ii)$ for all $u \in W$: $v \relations_{i_3} u \Rightarrow d \not \in (D\tr)^{I(u)}$.
Equivalently,
$d \in (\D_{i_1}(\B_{i_2} D\tr \sqcap \B_{i_3} \lnot D\tr))^{I_{w}}$
iff, by definition of $\cdot\tr$,
$d \in ((\B_{i} D)\tr)^{I_{w}}$.
\newline

We now prove, by induction on $\MLALC{n}$ formulas $\psi$, that for every $w \in W$:
\[
	\Mmc, w \mdl \psi \text{ iff } \Mmf, w \mdl \psi\tr
\]

For the case $\psi = C \sqs D$, it follows from the previous claim, while for $\psi = C(a)$ and $\psi = r(a,b)$, it is immediate from the definitions of $\Int$ and $\cdot\tr$, as well as from the claim above. Assuming that the lemma holds for $\chi$ and $\zeta$, the inductive cases $\psi = \lnot \chi$ and $\psi = \chi \land \zeta$ are straightforward. We prove the statement for modalised formulas.

$\psi = \B_{i} \chi$.
$\Mmc, w \mdl \Box_{i} \chi$ iff
$\llbracket \chi \rrbracket^{\Mmc}  \in \Nmc_{i}(w)$.
That is, there is $v \in W$ s.t. $w \relations_{i_1} v$ and $(i)$ for all $u \in  W$: $v \relations_{i_2} u \Rightarrow \Mmc, u \mdl \chi$ and $(ii)$ for all $u \in W$: $v \relations_{i_3} u \Rightarrow \Mmc, u \not \mdl \chi$.
By i.h., this is equivalent to: there is $v \in W$ s.t. $w \relations_{i_1} v$ and
$(i)$ for all $u \in  W$: $v \relations_{i_2} u \Rightarrow \Mmf, u \mdl \chi\tr$ and
$(ii)$ for all $u \in  W$: $v \relations_{i_3} u \Rightarrow \Mmf, u \not \mdl \chi\tr$.
The previous step means that: 
$\Mmf, w \mdl \D_{i_1} (\B_{i_2} \chi\tr \land \B_{i_3} \lnot \chi\tr)$
iff, by definition of $\cdot\tr$,
$\Mmf, w \mdl (\B_{i} \chi)\tr$.


Therefore,  in particular, $\Mmc, w \mdl \p$.
\end{proof}
\end{proof}

\Theoremcomplmalc*
\begin{proof}
This theorem is a consequence of the following claim, and the complexity 
of formula satisfiability in $\KnALC{2n}$ constant domain relational models~\cite{KraWol,GasHer,GabEtAl03}. 
\begin{claim} 
The $\MnALC{n}$ formula satisfiability problem on constant domain neighbourhood models can be reduced in polynomial time to the \KnALC{2n} formula satisfiability problem on constant domain relational models.
\end{claim}
 \begin{proof}[Proof of Claim]
The proof is analogous to that of Theorem~\ref{theor:complealc} (we show the cases for modalised concepts and formulas only).
Consider a $\MLALC{n}$ formula $\p$ satisfiable on supplemented neighbourhood frames, i.e., so that there is a neighbourhood model
$\Mmc = (\Fmc, \Delta, \Int)$
based on a supplemented neighbourhood frame
$\Fmc = (\Wmc, \{ \Nmc_{i} \}_{i \in J})$ and a $w$ in $\Mmc$ such that
$\Mmc, w \mdl \p$.
We define a relational frame
$\Fmf = (W, \{ R_{i_{1}}, R_{i_{2}} \}_{i \in J})$,
and an $\MLALC{2n}$ relational model
$\Mmf = (\Fmf, \Delta, I)$ based on $\Fmf$, such that:
	\begin{itemize}
		\item $W = \{ (w, 0) \mid w \in \Wmc \} \cup \{ (\alpha, 1) \mid \alpha \in \bigcup_{v \in \Wmc} \Nmc_{i}(v) \}$
		\item $\relations_{i_1} = \{ ((w, 0), (\alpha, 1)) \mid \alpha \in \Nmc_{i}(w)\}$;
		\item $\relations_{i_2} = \{ ((\alpha, 1), (w, 0)) \mid w \in \alpha \}$
		\item for every $(w, 0) \in W$, $I_{(w, 0)} = \Imc_{w}$; for every $(\alpha, 1) \in W$, $X^{I_{(\alpha, 1)}} = \eset$, for all $X \in \NC \cup \NR$, and $a^{I_{(\alpha, 1)}} = a^{\Int}$, for all $a \in \NI$.
	\end{itemize}

Firstly we show, by induction on the structure of concepts $C$, that for all $d \in \Delta$ and all $w \in \Wmc$:
\[
d \in C^{\Imc_{w}} \text{ iff } d \in (C\ttr)^{I(w,0)}.
\]
\noindent
$C = \B_{i} D$. 
We have  
$d \in (\B_{i} D)^{\Imc_{w}}$
iff 
$\llbracket D \rrbracket^{\Mmc}_{d} \in \Nmc_{i}(w)$.
Since $\Fmc$ is supplemented, the previous step means that there is
$\alpha \in \Nmc_{i}(w) \colon \alpha \sbs \llbracket D \rrbracket^{\Mmc}_{d}$.
This is equivalent to:
there is $\alpha \in \Nmc_{i}(w)$
s.t. for every
$v \in \Wmc : v \in \alpha \Rightarrow d \in D^{\Imc_{v}}$.
By i.h. and definition of $\relations_{i_2}$,
there is $\alpha \in \Nmc_{i}(w)$
s.t. for every
$v \in \Wmc : (\alpha, 1) \relations_{i_2} (v, 0) \Rightarrow d \in (D\ttr)^{I(v,0)}$.
Equivalently, there is
$\alpha \in \Nmc_{i}(w) \colon d \in (\B_{i_2} D\ttr)^{I(\alpha,1)}$.
By definition of $\relations_{i_1}$, this means:
there exists
$(\alpha,1) \in W$
s.t.
$(w,0) \relations_{i_1} (\alpha,1)$ and $d \in (\B_{i_2} D\ttr)^{I(\alpha,1)}$, 
iff
$d \in (\D_{i_1} \B_{i_2} D\ttr )^{I(w,0)}$.
That is, by definition of $\cdot\tr$,
$d \in ((\B D_{i})\ttr)^{I(w,0)}$.
\newline

Then, we prove that for every $\MLALC{n}$ formula $\psi$ and every $w \in \Wmc$:
\[
	\Mmc, w \mdl \psi \text{ iff } \Mmf, (w,0) \mdl \psi\ttr
\]
\noindent
$\psi = \B_{i} \chi$.
$\Mmc, w \mdl \B_{i} \chi$ iff
$\llbracket \chi \rrbracket^{\Mmc}  \in \Nmc_{i}(w)$.
Since $\Fmc$ is supplemented, this is equivalent to: 
there is $\alpha \in \Nmc_{i}(w)$ s.t. $ \alpha \sbs \llbracket \chi \rrbracket^{\Mmc}$,
iff
there is $\alpha \in \Nmc_{i}(w)$ s.t. for all $v \in \Wmc : v \in \alpha \Rightarrow \Mmc, v \mdl \chi$.
By i.h. and definition of $\relations_{i_2}$, there is $\alpha \in \Nmc_{i}(w)$
s.t. for all $v \in \Wmc : (\alpha,1) \relations_{i_2} (v,0) \Rightarrow \Mmf, (v,0) \mdl (\chi\ttr)$.
This means that, for some $\alpha \in \Nmc_{i}(w)$,
$\Mmf, (\alpha,1) \mdl \B_{i_2} \chi\ttr$.
By definition of $\relations_{i_1}$, the previous step is equivalent to:
there is $(\alpha,1) \in W$ s.t.
$(w,0) \relations_{i_1} (\alpha,1)$
and
$\Mmf, (\alpha,1) \mdl \B_{i_2} \chi\ttr$.
That is,
$\Mmf, (w,0) \mdl \D_{i_1} \B_{i_2} \chi\ttr$.
By definition of $\cdot\ttr$,
$\Mmf, (w,0) \mdl (\B_{i} \chi)\ttr$.

Thus, in particular, we have $\Mmf, (w,0) \mdl \p\ttr$.
\newline

Conversely, consider an $\MLALC{2n}$ formula $\p\ttr$ satisfiable on relational models:
$\Mmf, w \mdl \p\ttr$,
for some $\MLALC{2n}$ model
$\Mmf = (\Fmf, \Delta, I)$ based on $\Fmf = (W, \{ \relations_{i_{j}} \}_{j \in [1, 2]})$, and some $w \in W$. 
We define an $\MLALC{n}$ neighbourhood model
$\Mmc = (\Fmc, \Delta, \Int)$
based on
$\Fmc = (\Wmc, \{ \Nmc_{i} \}_{i \in [1, n]})$
s.t. $\Wmc = W$, and for all $w \in W$:
\begin{itemize}
			\item $\alpha \in \Nmc_{i}(w)$ iff there is $v \in W$ s.t. $w \relations_{i_1} v$ and for all $u \in W$, $v \relations_{i_2} u \Rightarrow u \in \alpha$;
			\item $\Imc_{w} =I_{w}$.
\end{itemize} 

Firstly, notice that $\Fmc$ is supplemented: 
for all $w \in W$, if $\alpha \in \Nmc_{i}(w)$ and $\alpha \sbs \beta \sbs W$, then there is $v \in W$ s.t. $w \relations_{i_1} v$ and for all $u \in W$, $v \relations_{i_2} u \Rightarrow u \in \beta$, i.e., $\beta \in \Nmc_{i}(w)$.
We now show, by induction on the structure of concepts $C$, that for all $d \in \Delta$ and all $w \in W$:
\[
d \in C^{\Imc_{w}} \text{ iff } d \in (C\ttr)^{I_{w}}.
\]
\noindent
$C = \B_{i} D$. 
We have  
$d \in (\B_{i} D)^{\Imc_{w}}$ iff 
$\llbracket D \rrbracket^{\Mmc}_{d} \in \Nmc_{i}(w)$.
By definition of $\Nmc$, this means that there is
$v \in W$ s.t. $w \relations_{i_1} v$ and for all $u \in  W$: $v \relations_{i_2} u \Rightarrow d \in D^{\Imc_{u}}$.
By i.h. the previous step is equivalent to: there is $v \in W$ s.t. $w \relations_{i_1} v$ and for all $u \in  W$: $v \relations_{i_2} u \Rightarrow d \in (D\ttr)^{I_{u}}$, iff
$d \in (\D_{i_1}\B_{i_2} D\ttr)^{I_{w}}$.
By definition of $\cdot\ttr$, $d \in ((\B_{i} D)\ttr)^{I_{w}}$.
\newline

Finally, we prove, by induction on $\MLALC{n}$ formulas $\psi$, that for every $w \in W$:
\[
	\Mmc, w \mdl \psi \text{ iff } \Mmf, w \mdl \psi\ttr
\]
\noindent
$\psi = \B_{i} \chi$.
$\Mmc, w \mdl \B_{i} \chi$ iff
$\llbracket \chi \rrbracket^{\Mmc}  \in \Nmc_{i}(w)$.
That is, there is $v \in W$ s.t. $w \relations_{i_1} v$ and for all $u \in  W$: $v \relations_{i_2} u \Rightarrow \Mmc, u \mdl \chi$.
By i.h., this is equivalent to: there is $v \in W$ s.t. $w \relations_{i_1} v$ and for all $u \in  W$: $v \relations_{i_2} u \Rightarrow \Mmf, u \mdl \chi\ttr$,
iff 
$\Mmf, w \mdl \D_{i_1} \B_{i_2} \chi\ttr$.
By definition of $\cdot\ttr$,
$\Mmf, w \mdl (\B_{i} \chi)\ttr$.

Therefore, we have in particular that $\Mmc, w \mdl \p$, i.e., $\p$ is satisfiable on a supplemented neighbourhood model.
\end{proof}
\end{proof}

\end{document}